\documentclass{amsart}

\usepackage[T1]{fontenc}
\usepackage{amsmath,a4wide}
\usepackage[foot]{amsaddr}
\usepackage{amssymb}
\usepackage{amscd}
\usepackage{epsfig}
\usepackage[T1]{fontenc}
\usepackage{epsfig}
\usepackage{amsmath}
\usepackage{graphics}
\usepackage{graphicx}
\usepackage{subfigure}
\usepackage{epstopdf}
\usepackage{color}
\usepackage{booktabs}
\usepackage{ifpdf}
\usepackage{cite}
\usepackage{amsthm}
\usepackage{amsmath}
\usepackage{tikz}
\usetikzlibrary{arrows,intersections}
\usepackage{boldline}
\usepackage{multirow}
\usepackage{bbm}
\usepackage[a4paper,top=1.5in, bottom=1.5in, left=1.2in, right=1.2in]{geometry}
\usepackage{varioref}
\usepackage{xr-hyper}
\usepackage{hyperref}
\usepackage{amsmath}
\usepackage{cleveref}
\usepackage{lipsum}
\hypersetup{
	colorlinks,linkcolor=blue,citecolor=blue,filecolor=red,urlcolor=blue}
\makeatletter
\newcommand*{\rom}[1]{\expandafter\@slowromancap\romannumeral #1@}
\makeatother
\definecolor{r}{rgb}{0.0, 0.28, 0.67}

\newcommand{\mathds}[1]{\text{\usefont{U}{dsrom}{m}{n}#1}}

\newtheorem{theorem}{Theorem}[section]
\newtheorem{proposition}[theorem]{Proposition}
\newtheorem{example}[theorem]{Example}
\newtheorem{definition}[theorem]{Definition}
\newtheorem{corollary}[theorem]{Corollary}
\newtheorem{lemma}[theorem]{Lemma}

\numberwithin{equation}{section}

\makeatother

\newtheorem{Condition}{Condition}{\bf}{\rm}

\makeatletter
\def\namedlabel#1#2{\begingroup
	\def\@currentlabel{#2}%
	\label{#1}\endgroup
}
\newenvironment{condition}[3]{\begin{Condition}[#3] \namedlabel{#2}{#1} }{\end{Condition}}
\newtheorem{assump}{Assumption}{\bf}{\rm}

\makeatother
\newtheorem{prop}{Proposition}

\newtheorem{lem}{Lemma}

\newtheorem{thm}{Theorem}

\makeatother
\newtheorem{defn}{Definition}{\bf}{\rm}

\makeatother

\begin{document}
	\title{Duality in optimal consumption--investment problems with alternative data}
	\author{Kexin Chen}
	\address{Department of Applied Mathematics, The Hong Kong Polytechnic University, Hung Hom, Hong Kong.}
	\email{kexinchen@polyu.edu.hk}
	\author{Hoi Ying Wong}
	\address{Department of Statistics, The Chinese University of Hong Kong, Shatin, N.T., Hong Kong.}
	\email{hywong@cuhk.edu.hk}

\begin{abstract}
		This study investigates an optimal consumption--investment problem in which the unobserved stock trend is modulated by a hidden Markov chain that represents different economic regimes. In the classical approach, the hidden state is estimated from historical asset prices, but recent advancements in technology enable investors to consider alternative data in their decision-making. These include social media commentary, expert opinions, COVID-19 pandemic data, and GPS data, which originate outside of the standard sources of market data but are considered useful for predicting stock trends.  We develop a novel duality theory for this problem and consider a jump-diffusion process for the alternative data series. This theory helps investors in identifying ``useful'' alternative data for dynamic decision-making by offering conditions to the filter equation that permit the use of a control approach based on the dynamic programming principle. We demonstrate an application for proving a unique smooth solution for a constant relative risk-averse agent once the distributions of the signals generated from alternative data satisfy a bounded likelihood ratio condition. In doing so, we obtain an explicit consumption--investment strategy that takes advantage of different types of alternative data that have not been addressed in the literature. 
		\medskip
	
	\medskip
	

\end{abstract}
	\maketitle 
	%----Section 1. 
\section{Introduction}
\indent The optimal consumption--investment problem is a classical problem in modern financial theory. The investor's objective is to maximize the expected utility of consumption and terminal wealth over a finite horizon. Merton's pioneering work \cite{merton1975optimum} on formulating the problem in a continuous-time framework has become the cornerstone of the development of a stochastic optimal control theory for solving this type of problem. Numerous generalizations of classical models have been studied in efforts to model the dynamics of asset prices more precisely. For example, \cite{elliott1997application, CJW2022, sotomayor2009explicit, yin2004markowitz, zhou2003markowitz} study the regime-switching model, in which the model coefficients are assumed to be modulated by a Markov chain. The different states of the chain are interpreted as different economic states or market modes. 
Moreover, \cite{bauerle2004portfolio, bauerle2007portfolio, honda2003optimal, sass2017expert, rieder2005portfolio} argue that the states of the Markov chain are not directly observable, so investors must learn and estimate them from observation, which results in partial information formulations. The literature refers to this model as a hidden Markov model (HMM).\\
\indent Traditionally, investors learn about the hidden economic state solely from easily accessible historical asset prices. However, investors are now actively acquiring alternative data through modern technology to supplement their decision-making. Social media commentary, internet search results, COVID-19 pandemic data, and GPS data are all examples of ``alternative data,'' that is, data that originate outside of the standard market data but are considered useful for predicting economic trends. Recent studies such as \cite{frey2012portfolio,callegaro2020optimal,fouque2015filtering, sass2017expert} support the use of aggregate consumption and macroeconomic indicators and expert opinions as additional sources of observation. The effective use of alternative data could improve estimation accuracy and risk-sensitive asset management performance \cite{davis2020debiased,davis2021risk}.\\
\indent There are novel technical difficulties associated with incorporating alternative data into dynamic decision-making because of the additional randomness that arises from alternative data. The aforementioned studies apply stochastic control techniques to an equivalent {\it primal} problem, the so-called separated problem, which is deduced from the original primal problem via filtering. Such a solution procedure is similar to that of the stochastic control problem with partial information, but the additional randomness complicates the mathematical analysis for the solvability of the problem and the eligibility of the solution procedure. Indeed, it is rarely discussed under what conditions alternative data and the corresponding filters permit a stochastic control framework, such as the dynamic programming principle (DPP), to become applicable to the underlying problem. One exception is the study of \cite{frey2012portfolio}, which requires that the density functions of signals generated from alternative data are continuously differentiable with common bounded support and uniformly bounded from below by a positive constant. Clearly, this excludes Gaussian signals and the most commonly used distributions. Under such a criterion, they prove the DPP and that there exists a unique value function for the case of power utility. In other words, the relevance of many types of alternative data to dynamic decisions remains unclear. The lack of rigorous results in a general setting limits our understanding of the optimal policies as well as the use of alternative data from various sources.\\
\indent To fill this theoretical gap, we propose a new methodology based on duality theory that is applicable for general types of alternative data in the context of consumption--investment problems with a more general class of utility functions, particularly the power utility function with a negative exponent. We offer concrete and novel results for specific problems that supplement those found in the literature. For instance, we provide a bounded likelihood ratio (BLR) condition for alternative data signals in a bull--bear regime economy for an agent with power utility. The BLR condition allows us to check for eligibility of signals from a wide range of distributions, such as Gaussian, exponential family, Gaussian mixture, and others. We describe three examples in Section \ref{sec:2_BLR}.\\
\indent Following the literature, we postulate the risky asset price as a geometric Brownian motion in which the drift is modulated by the hidden economic state, which also affects alternative data.  Inspired by \cite{davis2021jump}, the alternative data are sampled from a regime-switching jump-diffusion process with parameters depending on the hidden state. This consideration aims to capture the realistic nature of alternative data sources, such as ecosystem, electricity prices, manufacturing, and production prediction (see \cite{sethi2012hierarchical, xi2009asymptotic, xi2018martingale, yin2009hybrid, zhu2015feynman,weron2004modeling} and reference therein). It also covers the examples studied in recent literature \cite{callegaro2020optimal,frey2012portfolio}.  When the alternative date are incorporated to stock prediction for dynamic decision-making, our problem formulation involves a filtering scheme on the market and alternative data jointly to predict the stock trend so that suitable regularity is required for the alternative data generating process to ensure the use of DPP based on the adopted filter. Such use of alternative data makes it a clear difference from problem formulation with conventional jump-diffusion factor processes. \\
\indent \color{black}  Subject to the above general setup, our main theorem (Theorem \ref{thm:duality_main}) establishes an equivalence between the primal partial information problem and the dual problem, wherein the latter simply involves a minimization over a set of equivalent local martingale measures. To the best of our knowledge, this study is the first to extend the use of the duality approach from a partial information framework using a single observation process \cite{karatzas2001bayesian,Lakner1995,Lakner1998,pham2001optimal,Putschgl2008,Sass2004} to mixed-type observations using alternative data. The aforementioned studies characterize the dual formulation based on a single equivalent martingale measure, whereas we use non-unique equivalent martingale measures because of the additional randomness from alternative data. Once the dual problem is solved, the solution of the primal problem is obtained using convex duality. We discover that the dual problem, which is a stochastic control problem in itself but is very different from the primal problem, is more tractable. This enables us to use the DPP for the dual stochastic control problem under a general abstract condition on the filter equation. With regard to the application, such a condition outlines what type of alternative data can be considered ``useful'' for dynamic decision-making with DPP. More precisely, the dual problem can be read at the analytical level of the Hamilton--Jacobi--Bellman (HJB) equation, thereby providing a dual equation and improving our understanding of the optimal strategy. To demonstrate the whole solution procedure, we apply this general methodology to a concrete case study  and explicitly derive a novel feedback optimal consumption--investment strategy by analyzing the dual equation in Sect.~\ref{sec:1}. We prove a verification theorem (Theorem~\ref{thm:1}) that shows the dual value function is the unique smooth solution to the dual equation. These results are obtained under the mild condition (\ref{cond:BLR}) on alternative data signals which covers frequently seen examples that have not been addressed previously.\\ 
\indent  This study makes some technical contributions to overcome the mathematical challenges to achieve these novel results. In the framework of the aforementioned case study, the filter process is a jump-diffusion with L\'{e}vy type jumps, that is, the intensity of the jump measure depends on the filter process itself. This subtle feature creates analytical challenges in establishing the verification theorem.  One may expect to derive the dual equation via the DPP first heuristically and then, given the regularity of the dual equation solution (i.e., existence, uniqueness, and smoothness), verify the desired dual value function by formally applying the It\^{o}'s formula and a martingale argument. However, rigorously proving the regularity is surprisingly difficult because the dual equation is a degenerate partial integro-differential equation (PIDE) with an embedded optimization. To overcome these difficulties, we first show that the dual value function is a bounded Lipschitz continuous and, hence, a $C^1$ function of its arguments in Theorem~\ref{thm:lipctsx}. The result is technically innovative, as we introduce an auxiliary process and use the Radon--Nikodym derivatives to address the L\'{e}vy-type jumps of the filter process. As an immediate consequence, the filter process is shown to be {\it Feller} (Proposition~\ref{prop:filter_feller}), indicating that the DPP is valid and the solution procedure is eligible. We then show that the dual equation has a unique smooth ($C^{1,2}$) solution. The method is based on the link between viscosity solutions and classical solutions for PIDEs, which originate from \cite{pham1998optimal}, \cite{davis2010impulse} and \cite{davis2013jump}, but our context is different from theirs in that ours contains an optimization embedded in the nonlocal integro-differential operator in the PIDE. This distinctive nature leads to both nonlinearity and degeneracy on the state space boundaries, so we have to address both difficulties simultaneously. Finally, we obtain explicit formulas for the optimal strategies and wealth process in terms of functions of the dual equation solution in Proposition~\ref{prop:back_primal}.\\
\indent We believe that an extensive analysis of such a well-received case study is a valuable contribution to the literature in itself. Although there are studies in the stochastic control literature that deals with the controlled jump-diffusion model (see, e.g., \cite{barles2008second,davis2010impulse,davis2013jump,pham1998optimal,seydel2010general}), most jump mechanisms are exogenous and not dependent on the state process itself. To the best of our knowledge, the only related result presented in \cite{frey2012portfolio} proposes a distributional transformation and reconstructing the filter process as an exogenous jump type, so that techniques in the above-mentioned literature can be applied.   Their approach imposes restrictive
conditions on alternative data and a predominant constraint on trading strategies in order
to obtain the necessary technical estimates. However, this study derives technical estimates to develop empirically testable conditions that are consistent with the abstract general condition for the duality approach and then solves them for the dual problem under a more general setting.\\ 
\indent  The remainder of this paper is organized as follows. To simplify matters and allow for better illustration, Sect.~\ref{sec:1} begins with a concrete optimal consumption-investment problem in a bull-bear stock market, where expert opinions are regarded as alternative data. We detail the solution procedure for solving such a stochastic optimization problem and offer an explicit solution to the case of CRRA utility. This enables us to articulate the key mathematical challenges in the solution procedure and the advantage of the dual formulation. By considering a general regime-switching jump-diffusion model for the alternative data series and a general set of utility functions, Sect.~\ref{sec:3} develops the duality approach under partial information using alternative data. Specifically, we prove an equivalence between the primal and dual problems and present a condition on the filter equation that ensures DPP's validity in the dual. Sect.~\ref{sec:4} is devoted to the proof of verification theorem (Theorem~\ref{thm:1}) in Sect.~\ref{sec:1}, that is, to show that the dual value function is the unique classical solution to an HJB equation.
Sect.~\ref{sec:conclu} concludes the paper. 

\section{Expert opinions as alternative data}\label{sec:1}
Before developing a duality theory with alternative data in general setups, we specifically consider expert opinions and power utilities to exhibit the solution procedure of our duality approach in this section. This specification allows us to transparent the dual formulation and the regularity of the approach without an overwhelming burden in notation. We will see shortly that the solution procedure involves a stochastic optimal control problem in the dual and produces optimal solutions at the analytical level of the Hamilton-Jacobi-Bellman (HJB) equation in the dual. The investigation of general setups follows in Sect.~\ref{sec:3}. 
\subsection{A hidden Markov bull-bear financial market}
Within a fixed date $T>0$, which represents the fixed terminal time or investment horizon, we consider a filtered probability space $(\Omega,\mathcal{F},\mathbb{F},\mathbb{P})$, where $\mathbb{P}$ denotes the physical measure and $\mathbb{F}: = (\mathcal{F}_t)_{t\in[0,T]}$ denotes the full information filtration, satisfying the usual conditions: $\mathbb{F}$ is right-continuous and completed with $\mathbb{P}$-null sets. For a generic $\mathbb{F}$-adapted process $G$, we denote the filtration generated by $G$ as $\mathbb{F}^G$.\\
\indent We consider a two-regime hidden Markov financial market model, in which the transitions of the ``true'' regime are described by a two-state continuous-time hidden Markov chain $\alpha_t\in\mathcal{S}:=\{1,2\}$ on $(\Omega,\mathcal{F},\mathbb{F},\mathbb{P})$. This model provides a natural consideration of the bull and bear market, with $\alpha_t=1$ indicating the ``bull market'' state and $\alpha_t=2$ the ``bear market'' state at time $t$. The Markov chain $\alpha=(\alpha_t)_{t\in[0,T]}$ is characterized by the generator $\mathbf{A}$ of the form:
\begin{align}
	\label{eq:alpha_gen}
	\mathbf{A}=
	\begin{pmatrix} -a_1 & a_1 \\ a_2 & -a_2 \end{pmatrix},\;\;\; a_1, a_2 >0.
\end{align}
For time $t\in[0,T]$, we describe the financial market model as follows.\\
\noindent (i) The risk-free asset is given by $S^0_t = \displaystyle e^{rt}$, with risk-free interest rate $r > 0$.\\
\noindent (ii) The risky asset $S=(S_t)_{t\in[0,T]}$ satisfies the following stochastic differential equation (SDE):
\begin{align}\label{eq:st_1}
	dS_t = \mu(\alpha_t) S_t dt +\sigma S_t dW_t,
\end{align}	
where $(W_t)_{t\in[0,T]}$ is a standard $\mathbb{F}$-Brownian motion on $(\Omega,\mathcal{F},\mathbb{P})$ independent of $\alpha$. The risky asset's volatility $\sigma$ is a positive constant, and its drift process $\mu$ satisfies $ \mu(\alpha_t) \in \{\mu_1,\mu_2\}$, 
where $\mu_1>\mu_2$ are constant drifts under bull and bear markets, respectively.\\
\indent Unlike the Markov-modulated regime-switching model, which treats $\alpha$ as observable, we assume that the representative agent does not observe $\alpha$ directly. The {\it observation process} of the agent has two components: the asset price process $S$ and an alternative data process in the form of expert options. More precisely, the agent receives noisy signals about the current state of $\alpha$ at discrete time points $T_k$. The aggregated alternative data process $\eta$ is a standard marked point process that depends on the Markov chain $\alpha$, as described by the double sequence $\{T_k, Z_k \}_{k\ge 0}$ representing the time instants at which the signal arrives and complemented by a sequence of random variables, one for each time, which denote the signal size:
\begin{align}\label{eq:eta_1}
	\eta_t: = \sum_{T_k\le t} Z_{k}.
\end{align}
We assume that the intensity of signal arrivals is given by a constant $\lambda$. In other words, the signal arrival time is independent of the hidden state. The signal $Z_k$ takes values in some set $\mathcal{Z} \subseteq \mathbb{R}$, and given $\alpha_{T_k} = i\in\{1,2\}$, the distribution of $Z_k$ is absolutely continuous with Lebesgue-density $f_i(z)$. Equivalent to \eqref{eq:eta_1}, we have
\begin{align}\label{eq:eta_2}
	d\eta_t = \int_{\mathcal{Z}}z N(dt,dz), \quad N(dt, d z):=\sum_{k \geqslant 1} \delta_{\left(T_{k},  \Delta\eta _{T_{k}}\right)}(d t, d z) \mathds{1}_{\left\{T_{k}<\infty, \Delta\eta _{T_{k}} \neq 0\right\}},
\end{align}
where $\delta_{\left(T_{k},  \Delta\eta _{T_{k}}\right)}(\cdot, \cdot)$ denotes the Dirac measure at point $(T_k, \Delta\eta _{T_{k}}) \in [0,T]\times\mathcal{Z}$, so $N$ is an integer valued random measure on $[0,T]\times\mathcal{Z}$, where $(\mathcal{Z}, \mathcal{B}(\mathcal{Z}))$ is a given Borel space.
In particular, the $\mathbb{F}$-dual predictable projection (see Definition~\ref{A.def:dual_projec} in Appendix~\ref{App:E}) of random measure $N$ is given by $\lambda \sum_{i=1}^2 \mathds{1}_{\alpha_t = i} f_i(z)dz dt$.\\
\indent In other words, the information available to the agent is given by the {\it observation filtration} $\mathbb{H} = (\mathcal{H}_t)_{t\in[0,T]}$ with $\mathbb{H}: = \mathbb{F}^{S} \vee \mathbb{F}^{\eta} \subseteq \mathbb{F}$. This is a partial information setting because \eqref{eq:st_1} - \eqref{eq:eta_2} constitute a filtering system in which $\alpha$ and $(S, \eta)$ play the roles of state and observation, respectively.
\subsection{The optimal consumption-investment problem}
\indent Let $\varpi_t$ be the net amount of capital allocated in the risky asset, and $c_t$ the rate at which capital is consumed at time $t$. The agent's wealth process $V^{v,\varpi,c}$ corresponding to the choice $(\varpi,c)$ and initial wealth $v \in \mathbb{R}_+:=(0,\infty)$ evolves as
\begin{align} \label{eq:wealth_1}
	d V^{v,\varpi,c}_t = (\mu(\alpha_t) - r)\varpi_t dt+ (r V^{v,\varpi,c}_t -c_t ) dt + \varpi_t \sigma dW_t.
\end{align}
Formally, we define the choices of agents in the following way.\\
\noindent (h1) $\varpi=\left(\varpi_{t}\right)_{t \in[0, T]}$ is an investment process if it is a real $\left(\mathcal{H}_{t}\right)_{t \in[0, T]}$-measurable process with trajectories square-integrable in $[0, T]$.\\
\noindent (h2) $c=\left(c_{t}\right)_{t \in[0, T]}$ is a consumption process if it is a real nonnegative $\left(\mathcal{H}_{t}\right)_{t \in[0, T]}$-measurable process with trajectories integrable in $[0, T]$.\\
\indent As a class of admissible controls, we consider the pairs of processes $(\varpi,c)$ satisfying (h1), (h2) and such that the corresponding wealth process $V$ is nonnegative. The quantity to maximize in our optimization problem is 
\begin{align}\label{def:controlprob}
	\mathbb{E}\Big[U_1(V_T)+\int_{0}^{T} U_2(c_t) dt\Big],
\end{align}
where $U_i:(0,\infty)\rightarrow\mathbb{R}$, $i=1,2$ are utility functions of the form 
\begin{align}\label{eq:utilit_form}
	U_1(c) = U_2(c) = \frac{c^{\kappa}}{\kappa}, \quad \kappa \neq 0 \text{ and } \kappa <1. 
\end{align}
We highlight that the choice pair $(\varpi,c)$ is required to be adapted to the available information flow $\mathbb{H}$. Therefore, 
the stochastic control problem is under a partial information framework that has a larger available information, due to alternative data, than that of classical partial information problems. The alternative data improve  estimation of the state ($\alpha$) of the economy if they contain useful information. The estimation procedure is known as filtering and we have to study conditions for useful expert opinions under a filtering scheme. 
\subsection{Filtering}\label{sec:2.3_filter}
Following standard notations in the filtering literature, denote by $\hat{g_t} = \mathbb{E}[g_t| \mathcal{H}_t]$ the optional projection of a generic process $g=(g_t)_{t\in[0,T]}$ on the filtration $\mathbb{H}$. Let $\pi=(\pi_t)_{t\in[0,T]}$ be the filter of the hidden Markov chain $\alpha$ defined as $\pi_t = \mathbb{P}[\alpha_t=1|\mathcal{H}_t]$. For a process of the form $g_t = G(\alpha_t)$, its optional projection is given by $$\hat{G}(\pi_t):=\hat{g}_t = \pi_t G(1) + (1-\pi_t)G(2).$$\color{black}
We define the process $\widetilde{W}=(\widetilde{W}_t)_{t\in[0,T]}$ such that for any $t\in[0, T]$, 
\begin{align}\label{eq:W}
	\widetilde{W}_t: = \frac{1}{\sigma} \Big(\int_0^t\frac{dS_u}{S_u} - rdu\Big)-\int_0^t\hat{\theta}(\pi_u)du = W_t - \int_0^t (\hat{\theta}(\pi_u) - \theta(\alpha_u)) du.
\end{align}
where $\theta$ is the bounded function defined as:
\begin{align}\label{eq:theta}
	\theta(\alpha_t) : = (\mu(\alpha_t) - r)/\sigma \in\{ (\mu_1-r)/\sigma, (\mu_2-r)/\sigma\}.
\end{align}
By classical results from filtering theory (see e.g. \cite{kallianpur2013stochastic,liptser1977statistics}), $\widetilde{W}$ is  a $(\mathbb{P},\mathbb{H})$-Brownian motion (the so-called innovations process). We define the predictable random measure $\gamma^{\mathbb{H}}$ and function $\hat{f}:[0,1]\times\mathcal{Z}\rightarrow\mathbb{R}$ as follows
\begin{align}\label{eq:hatf}
	\gamma^{\mathbb{H}}(dt,dz):=\lambda \hat{f}(\pi_{t-},z) dz dt,\quad \hat{f}(x,z): = f_1(z)x+f_2(z)(1-x),
\end{align}
and $\gamma^{\mathbb{H}}$ is known as the $\mathbb{H}$-dual predictable projection of $N$ by standard results in filtering theory (see e.g. \cite{ceci2006risk,ceci2012nonlinear}). We thus introduce the $\mathbb{H}$-compensated jump measure of $N$ given by 
$$
\overline{N}^{\pi}(dt,dz): = N(dt,dz) - \gamma^{\mathbb{H}}(dt,dz)= N(dt,dz)  -\lambda\hat{f}(\pi_{t-},z)dtdz.
$$
By standard arguments in filtering theory (see e.g. \cite{kurtz1988unique,ceci2006risk,ceci2012nonlinear,callegaro2020optimal}), the filter $\pi$ is the unique strong solution of the following Kushner–Stratonovich equation
\begin{align}  \label{eq:filter1.5}
	d\pi_t =& (a_2-(a_1+a_2)\pi_t) dt + \pi_t(1-\pi_t)(\theta_1 - \theta_2) d\widetilde{W}_t \cr
	&+ \int_{\mathcal{Z}} (\xi(\pi_{t-},z)-\pi_{t-}) {\overline{N}^{\pi}(dt,dz)}, 
\end{align}
with initial value $\pi_0= x \in [0,1]$  and function $\xi:[0,1]\times\mathcal{Z} \rightarrow \mathbb{R} $ defined as:
\begin{align}\label{def:xi}
	\xi(x,z): = \frac{f_1(z)x}{f_1(z)x+f_2(z)(1-x)}.
\end{align}
We remark that the last term in \eqref{eq:filter1.5} can be  expressed as
\begin{align*}	\int_{\mathcal{Z}}& (\xi(\pi_{t-},z)-\pi_{t-}) \{ N(dt,dz) -\lambda( f_1(z)\pi_{t-}+f_2(z)(1-\pi_{t-}))dtdz \} \\
	=&	\int_{\mathcal{Z}} (\xi(\pi_{t-},z)-\pi_{t-}) N(dt,dz)- \lambda \left\{  \int_{\mathcal{Z}} f_1(z)\pi_{t-}-\pi_{t-}\hat{f}(\pi_{t-},z)dz \right\}dt,
\end{align*}
where the last term in the above equation equals to $0$ because both $f_1$ and $f_2$ are density functions defined on $z\in \mathcal{Z}$. Therefore, it is equivalent to write \eqref{eq:filter1.5} as
\begin{align}  \label{eq:filter2} 
	d\pi_t =& (a_2-(a_1+a_2)\pi_t) dt + \pi_t(1-\pi_t)(\theta_1 - \theta_2) d\widetilde{W}_t \cr
	&+ \int_{\mathcal{Z}} (\xi(\pi_{t-},z)-\pi_{t-}) {N(dt,dz)}.
\end{align}

\subsection{Primal and dual control problems}
Since we are going to apply dynamic programming techniques, we start by embedding the optimization problem in a family of problems indexed by generic time-space points $(t,x,v)\in [0, T]\times[0,1]\times\mathbb{R}_+$, the starting time, and the initial guess of the filter process, and initial wealth level. We denote the domain of $(t,x)$ by $\mathcal{U}_T:=[0,T)\times(0,1)$ and $\overline{\mathcal{U}}_{T}:=[0,T]\times[0,1]$.\\
\indent For given and fixed $(t,x)\in\overline{\mathcal{U}}_{T}$, we introduce the filtration $\mathbb{H}^{t}:=\left(\mathcal{H}_{s}^{t}\right)_{ s \in[t, T]}$:
$$
\mathcal{H}_{s}^{t}=\sigma\Big\{\widetilde{W}(r)-\widetilde{W}(t), N(r,A)-N(t,A), A \in \mathcal{B}({\mathcal{Z}}), t \leq r \leq s\Big\},
$$
where $\widetilde{W},N$ are defined in \eqref{eq:W} and \eqref{eq:eta_2}. Denote by $({\pi}_s)_{s\in[t,T]}$ the solution of \eqref{eq:filter2} on $[t,T]$ with initial guess $\pi_t=x$.  We introduce the measure $\mathbb{P}^{t,{x}}$ on $\mathcal{H}_T^t$ such that $\mathbb{P}^{t,{x}}({\pi}_t = {x})=1$, and denote by $\mathbb{E}^{t,{x}}$ the expectation operator under $\mathbb{P}^{t,{x}}$.\\
\indent For $v\in\mathbb{R}_+$, consider all the pairs of $\mathbb{H}^t$-measurable processes $(\varpi,c)$ that are defined analogous to (h1) and (h2), we denote by $V^{t,x,v,\varpi,c}$ the solution to \eqref{eq:wealth_1} starting at time $t$ from $v$ under the control $(\varpi,c)$. The class of admissible controls $\mathcal{A}(t,x,v)$ depending on the initial value $(t,x,v)\in\overline{\mathcal{U}}_{T}\times\mathbb{R}_+$ is defined as the set of pairs $(\varpi,c)$ satisfying the requirements above and such that 
\begin{align}\label{eq:budget_cons}
	V^{t,x,v,\varpi,c}_s \ge 0 \text{ a.s., }t < s\le T. \text{ (no-bankruptcy constraint)}
\end{align}
Clearly, the admissible set is nonempty for all $v\in \mathbb{R}_+$ because, for each initial value, the null strategy $(\varpi,c)\equiv (0,0)$ is always admissible. The agent's objective function is postulated to be
\begin{align}\label{def:J_object_func}
	\widetilde{J}(t, v,x;\varpi,c): = \mathbb{E}^{t,x}\Big[  U_1(V^{t,x,v,\varpi,c}_T) + \int_{t}^T U_2(c_t) dt \Big].
\end{align}
We define the {\it primal problem} as follows and denote by $J$ the value function associated with it, which we call the primal value function:
\begin{align}\label{def:primal_con_pro} 
	J(t,x,v) : = \sup_{(\varpi,c) \in \mathcal{A}(t,x,v)} \widetilde{J}(t,x,v;\varpi,c), ~{(t,x,v)\in \overline{\mathcal{U}}_{T}\times\mathbb{R}_+}. \tag{\bf P}
\end{align}
To apply the duality approach, we introduce the convex dual functions $\widetilde{U}_i$ of the concave utility functions $U_i$:
\begin{align}\label{eq:utilit_conv}
	\widetilde{U}_i (y) := \sup_{c>0} [U_i(c) -yc] = U_i(I_i(y)) - y I_i(y), ~ y\in\mathbb{R}_+,
\end{align}
where $ I_i(\cdot)$ is the inverse 
function of $\partial_cU_i(\cdot)$, $i=1,2$. For $U_i$ defined in \eqref{eq:utilit_form}, $\widetilde{U}_i(y)=-y^{\beta}/\beta$ with $I_i(y) = y^{\beta-1}$ and $\beta:=-\kappa/(1-\kappa)$, for $i=1,2$.
We also introduce the process $(Z_s^{\nu})_{s\in[t,T]}$ with initial value $Z_t^{\nu}=1$ defined for some $\mathbb{H}^t$-predictable process $(\nu(s,z))_{s\in[t,T]}$ indexed by $\mathcal{Z}$ (see Definition~\ref{A.def:func_index} in Appendix~\ref{App:E}):
\begin{align}\label{eq:z_dynamic}
	Z_s^{\nu} := & \exp\left( -\frac{1}{2}\int_{t}^{s} \hat{\theta}(\pi_u)^2 du -\int_{t}^{s} \hat{\theta}(\pi_u)d\widetilde{W}_u   \right)\cr
	&\exp\left( \int_{t}^{s}\int_{\mathcal{Z}} (1-e^{\nu(u,z)}) \hat{f}(\pi_{u-},z) dz du  +  \int_{t}^{s}\int_{\mathcal{Z}} \nu(u,z) N(du,dz) \right),
\end{align}
where $\hat{\theta}(\pi_t)$ is the optional projection of $\theta(\alpha_t)$ defined in \eqref{eq:theta} and $\hat{f}$ is defined in \eqref{eq:hatf}. We consider the admissible set of all $(\nu(s,z))_{s\in[t,T]}$ that satisfies the L\'{e}pingle-M\'{e}min condition (see e.g., Theorem 1.4 in \cite{ishikawa2016stochastic}):
\begin{align}
	& \int_{t}^{T}\int_{\mathcal{Z}} \left\{e^{2\nu(u,z)} + |\nu(u,z)|^2\right\} \lambda f_i(z)dz du < +\infty,~i=1,2;\label{cond:5}\\
	&\mathbb{E}^{t,x}\left[\exp\left(\int_{t}^{T}\int_{\mathbb{R}}\left\{e^{\nu(u,z)}\nu(u,z) + 1-e^{\nu(u,z)}\right\} \lambda\hat{f}(\pi_{u-},z) dz du \right)\right]<+\infty.\label{cond:6}
\end{align}
Let $\Theta^t$ be the admissible set of $(\nu(s,z))_{s\in[t,T]}$. Specifically,
\begin{align*}
	\Theta^t: = \{ \nu=(\nu(s,z))_{s\in[t,T]}| \nu \text{ is }\mathbb{H}^t\text{-predictable and  such that \eqref{cond:5}-\eqref{cond:6}  hold} \},
\end{align*}
which is not empty as $\nu\equiv 0 $ is admissible. Since $\hat{\theta}$ is bounded, the local martingale $Z^{\nu}$ is a martingale for every $\nu\in \Theta^t$. We thus define a $\mathbb{P}^{t,x}$-equivalent probability measure $\mathbb{Q}^{\nu}$ on $(\Omega,\mathcal{H}^t_T)$:
$\left.{d\mathbb{Q}^{\nu}}/{d\mathbb{P}^{t,x}}\right\vert_{\mathcal{H}^t_T} = Z^{\nu}_T$.
We observe that
\begin{align}\label{eq:changemeasure}
	Z_s^{\nu} = \mathbb{E}^{t,{x}}\Big[\frac{d\mathbb{Q}^{\nu}}{d\mathbb{P}^{t,{x}}} \Big\vert\mathcal{H}^t_s \Big],~ s\in[t,T],
\end{align}
and that $Z^{\nu}$ satisfies the SDE:
\begin{align*}
	d Z^{\nu}_s= -Z^{\nu}_s\Big(\hat{\theta}(\pi_s)d\widetilde{W}_s + \int_{\mathcal{Z}}(1-e^{\nu(s,z)} )\overline{N}^{\pi}(ds,dz)\Big),
\end{align*}
and $\pi_s$ is the solution of \eqref{eq:filter2} with $\pi_t= x$. In addition, for each $\nu\in\Theta^t$, 
\begin{align}\label{def:Q_set}
	\mathbb{Q}^{\nu} \in \mathcal{Q}: = \{\mathbb{Q}: \mathbb{P}^{t,x} \ll\mathbb{Q} \ll \mathbb{P}^{t,x}~ (e^{-r(s-t)}S_s)_{s\in[t,T]} \text{ is a } \mathbb{Q}\text{ martingale}.   \}
\end{align}
Let $(t,x,v)\in \overline{U}_T\times \mathbb{R}_+$ , $(\varpi,c) \in \mathcal{A}(t,x,v)$, $\nu\in\Theta^t$, and set $V=V^{t,x,v,\varpi,c}$. It\^{o}'s lemma yields
%\begin{align*}
%&d e^{-rs}V_sZ_s^{\nu} + c_s Z_s^{\nu} ds\\
%&=  e^{-rs}Z_s^{\nu} \left((V_s \hat{\theta}(\pi_s) + \varpi_s \sigma )d\widetilde{W}_s + \int_{\mathcal{Z}}(1-e^{\nu}(s,z) )\overline{N}^{\pi}(ds,dz) \right). 
%\end{align*}
that $
e^{-rT}Z^{\nu}_TV_T + \int_{t}^{T} e^{-rs}Z^{\nu}_sc_s dt
$ is a $(\mathbb{P},\mathbb{H})$-supermartingale (as a positive local martingale), which implies that (due to arbitrariness of $\nu\in\Theta^t$)
\begin{align*}
	\sup_{\nu\in\Theta^t}  \mathbb{E}^{t,x}\Big[ e^{-rT}Z^{\nu}_TV_T + \int_{t}^{T} e^{-rs}Z^{\nu}_sc_s ds \Big] \le e^{-rt} v.
\end{align*}
Together with the definition of $\widetilde{U}_i$ in \eqref{eq:utilit_conv}, we derive that for all $y\in\mathbb{R}_+$, $(\varpi,c) \in \mathcal{A}(t,x,v)$, and $\nu\in\Theta^t$, the agent's objective function $\widetilde{J}$ defined in \eqref{def:J_object_func} satisfies
\begin{align*}	
	\widetilde{J}(t, v,x; \varpi,c) \le& \mathbb{E}^{t,x}\Big[ U_1(V_T) +  \int_{t}^{T} U_2(c_s) ds\Big]\\
	&- y\mathbb{E}^{t,x}\Big[e^{-r(T-t)}Z_T^{\nu}V_T + \int_{t}^{T} e^{-r(s-t)}Z_s^{\nu}c_s ds\Big]+ vy \nonumber\\
	\le &\, \mathbb{E}^{t,x}\Big[\widetilde{U}_1(y e^{-r(T-t)}Z_T^{\nu}) + \int_{t}^{T} \widetilde{U}_2(  y e^{-r(s-t)}Z_s^{\nu}) ds\Big] + vy.
\end{align*}
Further taking supremum of $\widetilde{J}$ over $(\varpi,c)\in \mathcal{A}(t, v,x)$, we have that the primal value function $J$ defined in \eqref{def:primal_con_pro} satisfies:
\begin{align}\label{eq:dualderive}
	J(t,x,v) \le \mathbb{E}^{t,x}\Big[\widetilde{U}_1( y e^{-r(T-t)}Z_T^{\nu}) + \int_{t}^{T} \widetilde{U}_2(y e^{-r(s-t)}Z_s^{\nu}) ds\Big] + vy, 
\end{align}
for any $\nu \in \Theta^t$ and $y\in\mathbb{R}_+$. This calculation shows that the right-hand side (r.h.s) of \eqref{eq:dualderive} is an upper bound for $J$. Taking the infimum over $\nu \in \Theta^t$ on the r.h.s of \eqref{eq:dualderive} inspires us to consider the following {\it dual optimization problem} defined for $(t,x,y)\in \overline{\mathcal{U}}_T\times\mathbb{R}_+$, 
\begin{align}\label{def:dualprob} 
	\inf_{\nu \in \Theta^t} \widetilde{L}(t,x, y;\nu),\tag{\bf D}
\end{align}
where 
$
\widetilde{L}(t,x,y;\nu): = \mathbb{E}^{t,x}[\widetilde{U}_1(y e^{-r(T-t)}Z_T^{\nu}) + \int_{0}^{T} \widetilde{U}_2( y e^{-r(s-t)}Z_s^{\nu}) dt].
$
We denote by $\hat{L}$ the value function associated with this problem and call it the {\it dual value function}. Specifically,
\begin{align}\label{def:dual_value_func}
	\hat{L}(t,x, y):= \inf_{\nu \in \Theta^t} \widetilde{L}(t, x, y;\nu),~(t,x,y)\in \overline{\mathcal{U}}_T\times\mathbb{R}_+.
\end{align}
It then follows from \eqref{eq:dualderive} that
\begin{align}\label{eq:dualgap}
	J(t,x,v)  \le 	\inf_{y\in\mathbb{R}_+} \left\{ \hat{L}(t,x,y)+ vy\right\},~(t,x,v)\in \overline{\mathcal{U}}_{T}\times\mathbb{R}_+.
\end{align}
There is no {\it duality gap} between the primal problem \eqref{def:primal_con_pro} and the dual problem \eqref{def:dualprob} once the equality in \eqref{eq:dualgap} holds. The current formulation suggests that one can first work on the dual problem and then transform it back to the primal by closing the duality gap. Here, the dual problem \eqref{def:dualprob} is also a stochastic control problem.
\subsection{HJB in the dual}\label{sec:approach_HJB}
The dual problem reduces the original agent's problem with two control variables to only one control process $\nu\in\Theta^t$. The natural choice for solving it is a heuristic use of DPP: for $\mathbb{H}^t$-stopping time $\tau\in[t,T]$, the following holds,
\begin{align}\label{eq:DPP}
	\hat{L}(t,x,y) = \inf_{\nu \in \Theta^{t}} \mathbb{E}^{t,x}\Big[ \hat{L}(\tau, {\pi}_{\tau}, ye^{-r(\tau-t)}Z_{\tau}^{\nu} ) + \int_{t}^{\tau} \widetilde{U}_2(ye^{-r(s-t)}Z_{s}^{\nu}) ds\Big].
\end{align}
In this way, the HJB equation of the dual value function is derived as
\begin{align}\label{eq:hatL_HJB}
	\partial_t \hat{L}(t,x,y) + \inf_{\nu} \overline{\mathcal{L}}^{\nu} \hat{L} + \widetilde{U}_2(y) = 0,
\end{align}
where the dynamics \eqref{eq:z_dynamic} for $ye^{-r(s-t)}Z^{\nu}_s$ and \eqref{eq:filter1.5} for $\pi_s$ produces the generator 
\begin{align*}
	\overline{\mathcal{L}}^{\nu}\hat{L}: = & \lambda\int_{\mathcal{Z}} \left(\hat{L}(t,\xi(x,z),e^{\nu}y) - \hat{L}(t,x,y) +(1-e^{\nu})y\partial_y \hat{L}(t,x,y)\right)\hat{f}(x,z) dz  \\
	&+ \left( \frac{1}{2}x^2(1-x)^2(\theta_1 - \theta_2)\partial_{xx}  +  \hat{\theta}(x)\partial_x +\frac{1}{2}y^2\hat{\theta}(x)^2 \partial_{yy} - ry\partial_y\right)\hat{L}(t,x,y).
\end{align*}
Intuitively, the optimal dual optimizer $\nu^*$ could be constructed in a feedback form through the first-order conditions in the HJB equation \eqref{eq:hatL_HJB} if the candidate process is admissible, i.e., fulfilling conditions \eqref{cond:5}-\eqref{cond:6}.  The remaining task is to determine regularity conditions under which the alternative data and the corresponding filter allow for the above prescriptions.
\subsection{Regularity: Bounded likelihood ratio}\label{sec:2_BLR}
In Sect.~\ref{sec:3}, we study the regularity in much greater generality, in terms of the choice of utility functions and alternative data processes, that the above prescription is true. However, the regularity happens to be more abstract. %As a specific case of Theorem~\ref{thm:DPP_feller}, the {\it Feller} condition of filter $\pi$ (as derived in Proposition~\ref{prop:filter_feller} under \ref{cond:BLR}) guarantees the validity of the DPP. 
Under the setting of expert opinions in this section, we offer concrete technical conditions on the probability density functions of alternative data signals that validate \eqref{eq:DPP} and the proposed solution procedure. 
\begin{condition}{Condition of BLR}{cond:BLR}{Bounded Likelihood Ratio (BLR)} The probability density functions $f_1$ and $f_2$ of signals in \eqref{eq:eta_1} have the same support $\mathcal{Z}$ and admit finite second moments such that the following bounded likelihood ratio condition holds,
	\begin{align*}
		b_{\min} < \frac{f_2(z)}{f_1(z)} < b_{\max}, \quad \forall z\in \mathcal{Z},
	\end{align*}
	for some $0\le b_{\min}<1<b_{\max}$. We also expect the dissimilarity between two distributions to be reasonably bounded.  Specifically, we use {\bf a-divergence measure}  $D_{a}(f_1 \| f_2)$ (see e.g., \cite{amari2012differential}) indexed by $a=3$ to characterize such dissimilarity and require that for some constant $L_F>0$, 
	\begin{align*}\displaystyle
		D_{3}(f_1 \| f_2): =\int_{\mathcal{Z}} \frac{1}{6}(\frac{f_1(z)^3}{f_2(z)^2}-1) dz< L_{F}.
	\end{align*}
\end{condition}
The interpretation of \ref{cond:BLR} is that we should not expect the arriving signals to be particularly powerful in terms of distinguishing between the two regimes. Otherwise, the situation becomes similar to directly observe the state ($\alpha$). We should emphasize that \ref{cond:BLR} based on the duality covers a wider range of signals than those based on the primal in the literature. Indeed, it clearly covers those examples in assumption 5.1 and remark 5.2 in \cite{frey2012portfolio}, i.e., densities that are continuously differentiable with common bounded support and uniformly bounded below by a strictly positive constant.   Besides, \ref{cond:BLR} covers more examples of discrete distributions and continuous distributions defined in unbounded domains. We list a few of them below.
\begin{example}[Exponential family] \label{ex:BLR1} Let
	\begin{align*}
		\begin{cases}
			&f_1(z)= \exp(\sum_j g_j(z)v_j^{(1)}),\\
			&f_2(z) =  \exp(\sum_j g_j(z)v_j^{(2)}),
		\end{cases}
	\end{align*}
	where $v_j$ are the parameters of the distribution and $g_j$ are fixed features of the family, such as $(1,x,x^2)$ in the Gaussian case. \ref{cond:BLR} holds if there exists some constant $C$ such that $\sum_j (v_j^{(2)}-v_j^{(1)})g_j(z) \le C$ and $\sum_j (3v_j^{(1)}-2v_j^{(2)})g_j(z) \le C$ for all $z$. The following Gaussian examples clearly satisfy the conditions:
	\begin{align*}
		\begin{cases}
			&f_1(z) = \frac{\sqrt{1.6}}{\sqrt{2\pi}}e^{-0.8(z+1)^2},\\
			&f_2(z) = \frac{\sqrt{2}}{\sqrt{2\pi}}e^{-(z-1)^2}. 
		\end{cases}
	\end{align*}
\end{example}
\begin{example}[Mixture of Gaussian] \label{ex:BLR2} As a direct extension of Example 1, \ref{cond:BLR} holds for Gaussian density $f_2(z)$ and $f_1(z):=\sum_{j=1}^n a_j f_1^{(j)}(z)$ with $\sum_j a_j = 1$, which is a mixture of Gaussian distribution density, when each pair of $f_1^{(j)}$, $f_2$ fulfills the conditions in Example~\ref{ex:BLR1}.
\end{example}
\begin{example}[A mixture distribution and Gamma distribution] \label{ex:BLR3} Consider a mixture distribution and a Gamma distribution defined on $ \mathbb{R}_+$:
	\begin{align*}
		\begin{cases}
			&f_1(z)=a_2a_1z^{a_1-1}\mathds{1}_{z\in(0,1)}+(1-a_2)e^{1-z}\mathds{1}_{z\in(1,\infty)},\\
			&f_2(z)=z^{a_1-1}e^{-z}\mathcal{G}(a_1)^{-1}, 
		\end{cases}
	\end{align*}
	where $a_1,a_2\in(0,1)$ and $\mathcal{G}(a):=\int_{\mathbb{R}_+} z^{a-1}e^{-z}dz$ is the Gamma function. \ref{cond:BLR} holds with $b_{\min}=0$, $b_{\max}=\max\{ 1/a_2a_1,1/(1-a_2)e\}\mathcal{G}(a_1)^{-1}$ and $L_F=e^3\mathcal{G}(a_1)^2\mathcal{G}(2-2a_1)+2e^2\mathcal{G}(a_1)^2a_1^2$. 
\end{example}
Under \ref{cond:BLR}, we derive the following two useful properties of the filter process $\pi$, whose proofs are placed in Appendix~\ref{A.proof_Prop1}.
\begin{proposition}\label{prop:filter_boundary}
	Both $0$ and $1$ are unattainable boundaries for the filter process $\pi$, the solution of \eqref{eq:filter2}. In other words, they cannot be achieved from the interior of the state space $(0,1)$.
\end{proposition}
\begin{proposition}\label{prop:filter_feller}
	The Markov filter process $\pi^{x_0}:=(\pi^{x_0}_t)_{t\ge0}$ which is defined as the solution of \eqref{eq:filter2} starting from time 0 and a given starting point $x_0\in(0,1)$ is {\it Feller}. That is, following \cite{dynkin}, the function $P_tf(x): = \mathbb{E}[f({\pi}^x_t)]$ satisfies that for any bounded and continuous function $f$, 
	\begin{align}
		&\forall t\ge0,~  P_tf(x) \text{ is continuous},\label{eq:feller1}\\
		&\lim_{t\downarrow 0}P_tf(x) = f(x). \label{eq:feller2}
	\end{align}
\end{proposition} 
Proposition~\ref{prop:filter_boundary} implies that when characterizing the dual value function $\hat{L}$ via HJB method, no conditions should be imposed on the boundaries of the filter, neither on the value of the function nor on its partial derivatives (see definition 2.5 and remark 2.6 in \cite{bayraktar2012valuation} for a detailed discussion). 
Proposition~\ref{prop:filter_feller} implies that a similar initial guess of the hidden state will lead to similar developments in the filtering, and the filter itself changes in a reasonably continuous manner. The filter process's {\it Feller} property further validates DPP \eqref{eq:DPP} (see Theorem~\ref{thm:DPPfeller} in a general setup in Sect.~\ref{sec:3}). As a result of Propositions~\ref{prop:filter_boundary} and \ref{prop:filter_feller}, we have the following main result. 
\begin{theorem}[verification]\label{thm:1} Under \ref{cond:BLR}, the dual value function $\hat{L}$ is the unique classical ($C(\overline{\mathcal{U}}_T\times\mathbb{R}_+)\cap C^{1,2,2}({\mathcal{U}}_T\times\mathbb{R}_+)$) solution of the HJB equation \eqref{eq:hatL_HJB}, subject to the boundary condition $\hat{L}(T,x,y) = \widetilde{U}_2(y)$. It takes the form: 
	\begin{align}\label{eq:Lambda_form}
		\hat{L}(t,x,y)= -\frac{y^{\beta}}{\beta}\hat{\Lambda}(t,x),
	\end{align} where $\beta=-\kappa/(1-\kappa)$ and $\kappa$ is the risk aversion parameter of utility functions defined in \eqref{eq:utilit_form}; $\hat{\Lambda}$ is smooth. For $y\in\mathbb{R}_+$, the dual problem \ref{def:dualprob} admits a dual optimizer $\nu^*\in{\Theta}^t$:
	\begin{align}\label{def:opt_nu}
		{\nu}^*_s:=\hat{\nu}(s,\pi_{s-},z)=\frac{1}{1-\beta}	\ln\Big[\frac{\hat{\Lambda}(s,\xi(\pi_{s-},z))}{\hat{\Lambda}(s,\pi_{s-})}\Big], ~s\in[t,T].
	\end{align}
\end{theorem}
Proof of this theorem is given in Sect.~\ref{sec:4}.\\
Given that the dual optimizer in $\Theta^t$ for \eqref{def:dualprob} exists, we now turn to the proof of no {\it duality gap}.  We 
have the following result (a special case of Theorem~\ref{thm:duality_main} below) that closes the duality gap and derives the optimal controls for the primal problem. The proof is placed in Appendix~\ref{A.proof_Prop1}.
\begin{proposition}\label{prop:back_primal} Under \ref{cond:BLR}, there is no duality gap between the primal \eqref{def:primal_con_pro} and dual \eqref{def:dualprob} problems. Fix $(t,x,v)\in\mathcal{U}_T\times\mathbb{R}_+$, the optimal wealth process is 
	\begin{align*}
		{V}^*_s = v (e^{-r(s-t)}Z_s^{{\nu}^*})^{\beta-1} \frac{\hat{\Lambda}(s,\pi_s)}{\hat{\Lambda}(t,x)},~ s\in[t,T], 
	\end{align*}
	where $\hat{\Lambda}$ and ${\nu}^*$ are given in Theorem~\ref{thm:1} and $(Z_s^{{\nu}^*})_{s\in[t,T]}$ satisfies the SDE:
	\begin{align}\label{eq:z_2}
		d Z^{{\nu}^*}_s:= -Z^{{\nu}^*}_s\Big( \hat{\theta}(\pi_s)d\widetilde{W}_s + \int_{\mathcal{Z}}(1-e^{{\nu}^*_s}) \overline{N}^{\pi}(ds,dz)\Big),~Z_t^{{\nu}^*} = 1.
	\end{align}
	The optimal	controls of primal problem $({\varpi}^*,{c}^*)$ take the following feedback forms,
	\begin{align*}
		\begin{cases}
			&{\varpi}^*_s = \hat{\varpi}(s, \pi_s,{V}^*_s) := \frac{{V}^*_s}{\sigma}\left[(1-\beta)\hat{\theta}(\pi_s) + \frac{\partial_{x} \hat{\Lambda}(s,\pi_s)}{\hat{\Lambda}(s,\pi_s)} \right],~ s\in[t,T],\\
			&{c}^*_s = \hat{c}(s, \pi_s,{V}^*_s)  \frac{{V}^*_s}{\hat{\Lambda}(s,\pi_s)}, ~ s\in[t,T].
		\end{cases}
	\end{align*}
\end{proposition}

%We state below a condition that closes the duality gap, whose sufficiency is shown by Theorem~\ref{thm:duality_main} in Sect.~\ref{sec:3}. 
%\color{r}\begin{condition}{Condition of no duality gap}{cond:4}{Condition of no duality gap} 
%		For given and fixed $(t,x)\in\overline{\mathcal{U}}_{T}$, for every $y\in\mathbb{R}_+$, the {\it dual problem} \eqref{def:dualprob} admits a minimizer $\nu^y \in \Theta^t$. \label{ass:4}
%\end{condition}
%We prove \ref{cond:4} as a corollary of our next results. The key step is to apply dynamic programming techniques to the dual problem \eqref{def:dualprob}, which involves a minimization over a set of equivalent martingale measures. This requires us to check for the validity of the dynamic programming principle (DPP) of dual problem \eqref{def:dualprob}. The validity of the DPP under a more general situation is discussed in Sect.~\ref{sec:3}. The duality theorem with alternative data (Theorem~\ref{thm:duality_main}) plays an important role in the derivation below. The DPP is then directly used in Sect.~\ref{sec:4} to characterize the dual value function $\hat{L}$ at analytical level of the Hamilton-Jacobi-Bellman (HJB) equation.

%----Section 3. 
\section{Duality with alternative data: a general dynamic programming approach}\label{sec:3}
In this section, we present a general dynamic programming approach for solving the optimal choice problem based on duality, under a wider class of time-dependent utility functions (Assumption~\ref{ass:utilitu_gen}) and more general alternative data situations. We note that our results can be easily extended to the case of more than two economic states, which corresponds to a finite-state hidden Markov chain $\alpha$.\\ 
\indent We start by describing the general model of alternative data $\eta$ that serves as the setting for our (abstract) result.  In numerous practical scenarios, systems exhibit non-continuous trajectories and structural changes. Commonly used jump-diffusion models in financial asset price modeling (see e.g. \cite{tankov2003financial}) may fail to account for structural changes in alternative data that originate outside of the standard financial market. For this purpose, we are interested in a regime-switching jump-diffusion model because it incorporates discontinuous changes with regime-switching jump sizes and intensities.  Mathematically, we model $\eta$ by the SDE:
\begin{align}\label{eq:eta}
	d\eta_t = b_1(\eta_t,\alpha_t)dt + \sigma_1(\eta_t) dW_t + \sigma_2(\eta_t) dB_t + \int_{\mathcal{Z}} b_2(\eta_{t-},z) N_{\eta}(dt,dz),  
\end{align}
where $B$ is an $\mathbb{F}$-standard Brownian motion, $N_{\eta}(dt,dz)$ is an $\mathbb{F}$-adapted integer valued random measure on $[0,T] \times \mathcal{Z}$, jointly independent of Brownian motion $W$, and hidden Markov chain $\alpha$. In particular, the intensity measure of $N_{\eta}$ is given by $dt\gamma(\alpha_{t-},dz)dz$, which depends on the hidden state. 
To avoid unduly technicalities, we simply assume what we require: {\bf \eqref{eq:eta} has a unique strong solution}. Sufficient conditions are summarized in Assumption~\ref{A.ass:eta} in Appendix~\ref{App:E}.\\ 
\indent To proceed, we need to know the structure of $\mathcal{Q}$ introduced in \eqref{def:Q_set}, i.e, the set of all $\mathbb{P}$-equivalent probability measures $\mathbb{Q}$ on $\mathcal{H}_T$ for which the discounted risky asset price is a $\mathbb{Q}$-martingale. This requires us to define the innovation processes associated with the diffusion part and jump part  of \eqref{eq:eta}.  Recalling the notations introduced at the beginning 
of  Sect.~\ref{sec:2.3_filter}, together with \eqref{eq:W}, we define the following $(\mathbb{P},\mathbb{H})$-Brownian motion $\widetilde{B}$ and $\mathbb{H}$-compensated jump measure $\overline{m}^{\pi}$: 
\begin{align}
	\begin{cases}
		&\widetilde{B}_t: = B_t - \int_{0}^{t}\hat{\vartheta}(\eta_u,\pi_u)-\vartheta(\eta_u,\alpha_u) du,~ \vartheta(\cdot,i): = \frac{(b_1(\cdot,i) - {\theta(i)\sigma_1(\cdot)}/{\sigma})}{\sigma_2(\cdot)},\label{eq:B}\\
		&\overline{m}^{\pi}(dt,dq) := m(dt,dq) - \hat{\lambda}(\pi_{t-})\hat{\phi}_t(\pi_{t-},dq)dt,
	\end{cases}
\end{align}
where $m(d t, d q):=\sum_{s: \Delta \eta_{s} \neq 0} \delta_{\left(s, \Delta \eta_{s}\right)}(d t, d q)$ is the integer-valued random measure associated to the jumps of the process $\eta$; $\lambda_t(\alpha_{t-})\phi_t(\alpha_{t-}, dq)dt$ is the $\mathbb{P}$-dual predictable projection of $m$ (see proposition 3 in \cite{ceci2006risk}) satisfying
$$\lambda_t(\alpha_{t-})\phi_t(\alpha_{t-}, A)dt = \gamma(\alpha_{t-},\{z\in\mathcal{Z}:b_2(\eta_{t-},z)\in A/\{0\}\}), \text{ for} A\in\mathcal{B}(\mathbb{R});$$ \color{black} 
the filter $\pi_t = \mathbb{E}[\alpha_t=1|\mathcal{H}_t]$ is the unique solution to the Kushner–Stratonovich system:
\begin{align}
	&d \pi_t = (a_2-(a_1+a_2)\pi_t) dt + \pi_t(1-\pi_t)(\theta_1 - \theta_2) d\widetilde{W}_t+ \pi_t(1-\pi_t)(\vartheta_1 - \vartheta_2)d\widetilde{B}_t\nonumber\\ 
	&\qquad +\int_{0}^{t}\int_{\mathbb{R}} ( \xi(s,q) -\pi_{s-} )\overline{m}^{\pi}(ds,dq),~\xi(s,q) = \frac{d \pi_{s-}\lambda_s(1)\phi_s(1,dq)}{d(\hat{\lambda}(\pi_{s-})\hat{\phi}_s(\pi_{s-},dq))}.\label{filter:KS}
\end{align}
We have that all $(\mathbb{P},\mathbb{H})$-local martingales can be constructed through the triplet $(\widetilde{W},\widetilde{B},\overline{m}^{\pi})$ (see Proposition~\ref{prop:martg_represen} in Appendix~\ref{App:E} for formal statements). Hence, for given $t\in[0,T]$, $\mathbb{Q}$ is in $\mathcal{Q}$ if  and
only if its Radon–Nikodym derivative w.r.t $\mathbb{P}^{t,x}$ on $\mathcal{H}^t_T$ is given by Dol\'{e}ans’ exponential $\overline{Z}^{\mathbb{Q}}$, where for $s\in[t,T]$, 
\begin{align}\label{eq:z_gene}
	\overline{Z}^{\mathbb{Q}}_s=&
	\mathcal{E}( -\int_{t}^{\cdot}  \hat{\theta}(\pi_u)d\widetilde{W}_u - \int_{t}^{\cdot}\nu_D(u)d\widetilde{B}_u - \int_{t}^{\cdot}\int_{\mathbb{R}} (1-e^{\nu_J(u,q)}) \overline{m}^{\pi}(du,dq) )_{s},
\end{align}
for some $\mathbb{H}^t$-predictable processes $\nu_D$ and $\mathbb{H}^t$-predictable process $\nu_J$ indexed by $\mathbb{R}$, satisfying the L\'{e}pingle-M\'{e}min condition.\\
\indent Under the current general setup, the dual optimization problem is posed as 
\begin{align}\label{def:gene_dual_L}
	\hat{L}(t,x,y)= \inf_{\mathbb{Q}\in \mathcal{Q}}
	\mathbb{E}^{t,x}\Big[\widetilde{U}_1\big(T, ye^{-r(T-t)}\widetilde{Z}_{T}^{\mathbb{Q}} \big)+ \int_{t}^{T} \widetilde{U}_2\big(s, ye^{-r(s-t)}  \widetilde{Z}_{s}^{\mathbb{Q}}\big) ds \Big].\tag{\bf D'}
\end{align}
We highlight that a notable advantage of solving the dual problem in the context of general alternative data setups is the broad applicability of the DPP approach. We cite the following abstract result which shows DPP is valid when the filter process is {\it Feller}. Such a condition offers important insight into what type of alternative data is considered ``useful'' in terms of verification of the problem, that is, the solution procedure as demonstrated in Sect.~\ref{sec:approach_HJB}.
\begin{theorem}[theorem 3.17 in \cite{zitkovic2014dynamic}]\label{thm:DPPfeller}
	Suppose that filter process $({\pi}_t)_{t\in[0,T]}$ as the unique solution to the Kushner–Stratonovich system \eqref{filter:KS} is {\it Feller}.  Then DPP holds for the dual value function $\hat{L}$ defined in \eqref{def:gene_dual_L}; specifically, 
	\begin{itemize}
		\item[i.]
		for any $\mathbb{H}^t$-stopping time $\tau\in[t,T]$ and each $(t,x,y)\in \overline{\mathcal{U}}_T \times \mathbb{R}_+ $,
		\begin{align}\label{eq:DDP}
			\hat{L}(t,x,y) = \inf_{\mathbb{Q}\in\mathcal{Q}} \mathbb{E}^{t,x}[ \hat{L}(\tau, \pi_{\tau},ye^{r(t-\tau)}\widetilde{Z}_{\tau}^{\mathbb{Q}}) + \int_{t}^{\tau} \widetilde{U}_2(s,   ye^{r(t-s)}\widetilde{Z}_{s}^{\mathbb{Q}} ) ds].
		\end{align}
		\item[ii.] For $\epsilon>0$, an $\epsilon$-optimal $\mathbb{Q}^* \in \mathcal{Q}$ can be associated with each $(t,x,y)\in \overline{\mathcal{U}}_T \times \mathbb{R}_+$ in a universally measurable way. 
	\end{itemize}
\end{theorem}
\indent We are now ready to state the main result of this section, with  proof placed in Appendix~\ref{A.proof_thm1}. It establishes the equivalence between the primal and dual  problems. 
\begin{theorem}\label{thm:duality_main} 
	For a class of time-dependent utility functions with suitable growth conditions (Assumption~\ref{ass:utilitu_gen}), { suppose that the dual optimizer $\mathbb{Q}^y\in\mathcal{Q}$ for \eqref{def:gene_dual_L} exist for all $y\in\mathbb{R}_+$}, then for every initial wealth level $v\in\mathbb{R}_+$, there exists a real number $y^*= y(v) >0$ such that
	$$
	J(t,x,v) = \hat{L}(t,x,y^*)+vy^* =\widetilde{L}(t,x,y^*;\nu^{y^*})+vy^* = \inf_{y\in\mathbb{R}_+} \{ \hat{L}(t,x,y)+vy \},
	$$
	where $J$ is the primal value function and $\nu^{y^*}$ is the dual optimizer of \eqref{def:gene_dual_L} for $y^*$. In particular, there is no duality gap. There exists a pair $(\varpi^*,c^*)\in \mathcal{A}(t,x,v)$ with $c^*_s = I_2(s, e^{-r(s-t)}y^*\widetilde{Z}_s^{\mathbb{Q}^{y^*}})$ and $V_T^{t,x,v, \varpi^*,c^*} = I_1( T, e^{-r(T-t)}y^*\widetilde{Z}_T^{\mathbb{Q}^{y^*}})$, that is optimal to the primal problem \eqref{def:primal_con_pro}.
\end{theorem}

%----Section 4. 
\section{The dual value function as a classical solution of the HJB equation}\label{sec:4}
\subsection{Proof of Theorem~\ref{thm:1}}\label{sec:4.1} The main difficulty stems from the nonlinear integro-differential term and degeneracy induced by the filter process, making it very difficult to tackle directly via the PDE theory of classical solutions. 
We first deduce the form of  $\hat{L}$ given by \eqref{eq:Lambda_form}. Recall that $\widetilde{U}_i(y) =  - y^{\beta}/\beta$, it is clear by definition \eqref{def:dual_value_func} that $\hat{L}$ is written as: 
\begin{align*}
	\hat{L}(t,x,y) = y^{\beta}\inf_{\nu \in \Theta^t} - \frac{1}{\beta}\mathbb{E}^{t,x}\Big[(e^{-r(T-t)} Z_T^{\nu})^{\beta} + \int_{t}^{T} (e^{-r(s-t)}Z_s^{\nu})^{\beta}  ds\Big].
\end{align*}
For fixed $y\in\mathbb{R}_+$, the dual optimization in \eqref{def:dualprob} is therefore reduced to the following {\it auxiliary dual problem}:
\begin{align*}
	\text{maximize [minimize] }
	\Lambda(t,x;\nu):= \mathbb{E}^{t,x}\Big[(e^{-r(T-t)}Z_T^{\nu})^{\beta} + \int_{t}^{T} (e^{-r(s-t)}Z_s^{\nu})^{\beta}  ds\Big]
\end{align*}
over $\nu\in\Theta^t$, where maximize or minimize depends on the sign of the utility parameter $\kappa$ in \eqref{eq:utilit_form}. With a change of measure, we write $\Lambda$ in the following form:
\begin{align}\label{eq:Lambda2}
	\Lambda(t,x;\nu)= \widetilde{\mathbb{E}}^{t,x,\nu}\Big[ e^{\int_{t}^{T}\Gamma( \pi_u,\nu)du} + \int_{t}^{T} e^{\int_{t}^{s}\Gamma( \pi_u,\nu)du} ds \Big],
\end{align}
where $\widetilde{\mathbb{E}}^{t,x,\nu}$ denotes the expectation associated to measure $\widetilde{\mathbb{P}}^{t,x,\nu}$ defined via
\begin{align*}
	&	\frac{d\widetilde{\mathbb{P}}^{t,x,\nu}}{d\mathbb{P}^{t,x}}\Big\vert_{\mathcal{H}^t_T} = \widetilde{Z}_T^{\nu}, ~\widetilde{Z}_T^{\nu}:= \exp\Big( -\int_{t}^{T} \beta\hat{\theta}(\pi_u) d\widetilde{W}_u - \frac{1}{2}\int_{t}^{s}\beta^2 \hat{\theta}(\pi_u)^2 du\nonumber\\
	&+ \int_{t}^{T}\int_{\mathcal{Z}}\beta\nu(u,z) N(du,dz) + \lambda \int_{t}^{T}\int_{\mathcal{Z}} (1-e^{\beta\nu(u,z)})\hat{f}(\pi_{u-},z)dzdu \Big);\\
	&\displaystyle \Gamma(x,\nu):= -\beta r-\frac{1}{2}\beta(1-\beta) \hat{\theta}(x)^2 + \lambda\int_{\mathcal{Z}} (e^{\beta\nu(u,z)}-1+ \beta(1-e^{\nu(u,z)}))\hat{f}(x,z)dz,
\end{align*}
and recalling that  $\hat{\theta}(x) = \theta(1)x+ \theta(2)(1-x)$, $\hat{f}(x,z)= f_1(z)x+ f_2(z)(1-x)$, $z \in \mathcal{Z}$. 
In addition, under $\widetilde{\mathbb{P}}^{t,x,\nu}$, the dynamic of the filter process $\pi$ evolves as
\begin{align}\label{eq:filter_measure_nu}
	&d\pi_s = \overline{\mu}(\pi_s) ds + \overline{\sigma}(\pi_s)dW^{\beta}_s+\int_{\mathcal{Z}} (\xi (\pi_{s-},z)-\pi_{s-}) N(ds,dz),~ \pi_t = x, \\
	&\text{ with }\overline{\mu}(x):= (a_2-(a_1+a_2)x) -\beta \overline{\sigma}(x)\hat{\theta}(x), \text{ and } \overline{\sigma}(x):= x(1-x)(\theta_1-\theta_2). \nonumber
\end{align}
Here, $ W_s^{\beta}: = \widetilde{W}_s + \beta \int_{t}^{s} \hat{\theta}(\pi_u) du $ is a standard $(\mathbb{P}^{t,x,\nu},\mathbb{H})$-Brownian motion, and $\widetilde{N}^{\beta}(ds,dz):= N(ds,dz)-e^{\beta\nu(s,z)}\lambda \hat{f}(\pi_{s-},z)dzds$ is the $\mathbb{H}$-compensated Poisson random measure under  $\widetilde{\mathbb{P}}^{t,x,\nu}$. The value function associated with the {\it auxiliary dual problem} is defined as 
\begin{align}\label{def:auxi_dual_valuefunc}
	\hat{\Lambda}(t,x) : =\sup_{\nu\in\Theta^t} \Lambda(t,x;\nu),~ (t, x) \in\overline{\mathcal{U}}_{T},
\end{align}
when $\kappa<0$ and $\sup$ in \eqref{def:auxi_dual_valuefunc} is replaced by $\inf$ otherwise.
Theorem~\ref{thm:1} is then equivalent to the following result.
\begin{theorem}\label{thm:main}
	Under \ref{cond:BLR},  $\hat{\Lambda}(t,x)\in C(\overline{\mathcal{U}}_T)\cap C^{1,2}({\mathcal{U}}_T)$ is the unique classical solution to the following HJB PIDE:
	\begin{align}\label{eq:HJB}
		&\partial_t \hat{\Lambda}+ \overline{\mu}(x)\partial_{x}\hat{\Lambda}+ \frac{1}{2}\overline{\sigma}(x)^2\partial_{xx}\hat{\Lambda}- d_0(x)\hat{\Lambda} + \mathcal{I}_{\beta}[\hat{\Lambda}] +1= 0, ~\text{ in } \mathcal{U}_T,\\
		&d_0(x): = \beta r+ \frac{1}{2}\beta(1-\beta)\hat{\theta}(x),\label{def:d}\\
		&\mathcal{I}_{\beta}[\hat{\Lambda}](t,x): = (1-\beta)\lambda\int_{\mathcal{Z}} \Big[ {\hat{\Lambda}(t,x)}^{\frac{\beta}{\beta-1}} {\hat{\Lambda}(t,\xi(x,z))}^{\frac{1}{1-\beta}}-{\hat{\Lambda}(t,x)}\Big]\hat{f}(x,z)dz.\label{def:I}
	\end{align}
	with boundary condition $\hat{\Lambda}(T,x)=1$, $x\in[0,1]$.
	Furthermore, $\hat{\Lambda}(t,x) = \Lambda(t,x;{\nu}^*) $ where ${\nu}^*\in\Theta^{t}$ is the Markov policy given by 
	\begin{align}\label{eq:optimal_nu_2}
		{\nu}^*_s:=\hat{\nu}(s,\pi_{s-},z)= \frac{1}{1-\beta}	\ln\Big[\frac{\hat{\Lambda}(s,\xi(\pi_{s-},z))}{\hat{\Lambda}(s,\pi_{s-})}\Big], ~s\in[t,T].
	\end{align}
\end{theorem}
The proof is divided into several steps that are organized into three subsections. One preliminary step is to show that the control processes in {\it auxiliary dual problem} \eqref{def:auxi_dual_valuefunc} can be restricted to those in $\Theta^t$ taking values in $[-M,M]$ for a fixed positive constant $M$ sufficiently large. We denote this set by $\Theta^{t,M}$, and the corresponding {\it constrained auxiliary dual value function} by $\Lambda^{M}(t,x)$.  We start with presenting the lower and upper bounds of $\hat{\Lambda}$. The estimates are used to verify that the restriction on $\nu$ can be removed. 
\begin{proposition} \label{prop:Lambda_bound} There are positive constants $C_{\ell}$ and $C_{u}$ that only depend on utility parameter $\kappa$, such that
	\begin{align}\label{eq:Lambda_bound}
		C_{\ell} \le \hat{\Lambda}(t, x) \leq C_{u},\quad \forall(t, x) \in\overline{\mathcal{U}}_{T}.
	\end{align}
\end{proposition}
\begin{proof}
	{\it Case $\kappa<0$.} Note that $0<\beta<1$, and therefore the function $h(d): = e^{\beta d}-1+\beta(1-e^d)$ satisfies $h(d)\le 0$ for $d \in \mathbb{R}$. Using \eqref{eq:Lambda2}, we have
	\begin{align*}
		\widetilde{\mathbb{E}}^{t,x,\nu} &\Big[e^{ \int_{t}^{T}-\beta\left(r+\frac{1}{2}(1-\beta)\theta_1^{2}\right) + \lambda \int_{\mathcal{Z}} h(\nu(u,z))\hat{f}(\pi_{u-},z)dz du  } \\
		&+ \int_{t}^{T} e^{ \int_{t}^{s}-\beta\left(r+\frac{1}{2}(1-\beta)\theta_{1}^{2}\right) + \lambda \int_{\mathcal{Z}} h(\nu(u,z))\hat{f}(\pi_{u-},z)dz d u} d s\Big] \\
		& \leq \Lambda(t,x;\nu) \leq 1+T, \quad  \forall(t, x) \in\overline{\mathcal{U}}_{T},
	\end{align*}
	and they imply that
	$$
	0< e^{-\beta\left(r+\frac{1}{2}(1-\beta)\theta_{1}^{2}\right) T}(1+T) \le  \hat{\Lambda} (t, x) \le 1+T.
	$$
	{\it Case $0<\kappa<1$.} Note that $\beta<0$, and therefore the function $h(d): = e^{\beta d}-1+\beta(1-e^d)$ satisfies $h(d)\ge 0$ for $d \in \mathbb{R}$.	Similar arguments give us the following estimates:
	\begin{align*}
		1 \le \hat{\Lambda} (t, x) \le e^{-\beta\left(r+\frac{1}{2}(1-\beta)\theta_{1}^{2}\right) T}(1+T).
	\end{align*}
	Since the above lower and upper bounds do not depend on the initial state of the filter process, $C_{\ell}$ and $C_{u}$ can be constructed  for given $\kappa<1$ and $\kappa\neq 0$. 
\end{proof}
We provide the following auxiliary lemma. 
\begin{lemma}\label{lem:main} When $\kappa<0$, suppose that the {\it constrained auxiliary dual value function} $\Lambda^{M}(t,x)$ is the unique classical ($C(\overline{\mathcal{U}}_T)\cap C^{1,2}({\mathcal{U}}_T)$) solution of the HJB equation: \begin{align}\label{eq:constrainted_HJBPIDE1}
		&{\partial_t}\Lambda^{M}(t, x)+\max\limits_{\nu\in[-M,M]} \Big\{\Gamma(x,\nu) \Lambda^{M}(t, x)+\mathcal{L}^{\nu} \Lambda^{M}(t, x)\Big\}+1 = 0, \text{ in }\mathcal{U}_{T},\\
		& \text{where }\mathcal{L}^{\nu} g(t, x):= \overline{\mu}(x)\partial_xg(t,x)+\frac{1}{2} \overline{\sigma}^2(x)\partial_{xx}g(t,x)\cr
		&\qquad\qquad\quad +\int_{\mathcal{Z}}\left\{g(t,\xi(x, z))-g(t, x)\right\}\lambda e^{\beta\nu(z)}\hat{f}(x,z) d z,\label{eq:generator_L_nu}
	\end{align} 
	subject to the boundary condition: $\Lambda^{M}(T,x) = 1$, $x\in[0,1]$;
	and suppose that similar statements hold for $0<\kappa<1$ by replacing $\max$ in \eqref{eq:constrainted_HJBPIDE1} by $\min$. Let $\Lambda^M$ be the constrained auxiliary dual value function with 
	\begin{align}\label{eq:Mvalue}
		M > \max\left[ \ln(C_{u}/C_{\ell}), -\ln(C_{\ell}/C_{u})\right]/(1-\beta).
	\end{align}
	Then $\Lambda^M(t,x) = \hat{\Lambda}(t,x)$, where $\hat{\Lambda}$ is the unconstrained value function in \eqref{def:auxi_dual_valuefunc}.
\end{lemma}
\begin{proof} 
	We prove for the case when $\kappa<0$ while the case $0<\kappa<1$ follows similarly. The maximum selector on the l.h.s of \eqref{eq:constrainted_HJBPIDE1} induces a Markov policy $\hat{\nu}^M$ defined, for $(s,x)\in\overline{\mathcal{U}}_T$ indexed by $\mathcal{Z}$, as follows:
	\begin{align}\label{def:optimalnu}
		&\hat{\nu}^M(s,x,z) := \arg \max_{\nu\in[-M,M]} \left\{\Gamma(x,\nu) \Lambda^{M}(s, x)+\mathcal{L}^{\nu} \Lambda^{M}(s, x) \right\}\\ 
		&=\begin{cases}
			&\frac{1}{1-\beta}	\ln\left[\frac{\Lambda^{M}(s,\xi(x,z))}{\Lambda^{M}(s,x)}\right],\text{ if }\frac{1}{1-\beta}	\left|\ln\left[\frac{\Lambda^{M}(s,\xi(x,z))}{\Lambda^{M}(s,x)}\right]\right| \le M\\
			&M \mathrm{sgn} \ln\left[\frac{\Lambda^{M}(s,\xi(x,z))}{\Lambda^{M}(s,x)}\right], \text{otherwise}
		\end{cases},  ~s\in[t,T]. \nonumber
	\end{align}
	Using \eqref{eq:Lambda_bound} (notice that the estimates also hold for the {\it constrained auxiliary dual value function} $\Lambda^M$), it follows that, for $M$ satisfying \eqref{eq:Mvalue}, $\frac{1}{1-\beta}	\left|\ln\left[\frac{\Lambda^{M}(s,\xi(x,z))}{\Lambda^{M}(s,x)}\right]\right| < M$, so the constraints in \eqref{eq:constrainted_HJBPIDE1} can be removed, i.e.,
	\begin{align}
		\left[\partial_{t}+\mathcal{L}^{\nu_s}+\Gamma(\pi_s,\nu_s) \right] {\Lambda}^M(s,\pi_s)   \le -1, ~s\in[t,T], ~\forall~ \nu\in\Theta^t. \label{proof:1}
	\end{align}
	This inequality together with the Feynman--Kac formula imply that for $\nu\in\Theta^t$:
	\begin{align}
		{\Lambda}^M(t,x) = & \widetilde{\mathbb{E}}^{t,x,\nu}\Big[ e^{\int_{t}^{T}\Gamma( \pi_u,\nu_u)du}{\Lambda}^M(T,\pi_T)  \label{pf:lem_main_1} \\
		&-\int_{t}^{T} e^{\int_{t}^{s}\Gamma( \pi_u,\nu_u)du} 	[\partial_{t}+\mathcal{L}^{\nu_s}+\Gamma(\pi_s,\nu_s)]{\Lambda}^M(s,\pi_s)  ds \Big] 
		\ge  \Lambda(t,x;\nu).\nonumber
	\end{align}
	Taking supreme over $\nu\in\Theta^t$, we have $\Lambda^M \ge \hat{\Lambda}$, and hence $\Lambda^M = \hat{\Lambda}$ by definition. Given that ${\Lambda}^M$ is continuous and bounded, the Markov policy $\hat{\nu}$ defined in \eqref{eq:optimal_nu_2} is bounded, continuous, and $x$-locally Lipschitz. Thus, the Markov control process  $\nu^*$ in \eqref{eq:optimal_nu_2} belongs to $\Theta^{t,M}\subset \Theta^{t}$. From the definition $\nu^*$, the inequalities in \eqref{proof:1} and \eqref{pf:lem_main_1} become equalities for $\nu={\nu}^*$. Hence, $\Lambda^M(t,x) = \hat{\Lambda}(t,x) = \Lambda(t,x;{\nu}^*)$. Finally, substituting the Markov policy  $\nu^*$ in HJB equation \eqref{eq:constrainted_HJBPIDE1} we obtain \eqref{eq:HJB}.  
\end{proof} \color{black}
In the remainder of this section, we prove for Theorem~\ref{thm:main}. We find it is convenient to restrict the control set to $\Theta^{t,M}$ with $M$ sufficiently large for a moment and remove this restriction later by Lemma~\ref{lem:main}.  To help
readers better understand the main idea of the proof, we provide an outline before discussing them in detail.\\ 
\indent  {\it Step 1:  $\Lambda^M$ is uniformly Lipschitz on the state space $(t,x)\in\overline{\mathcal{U}}_T$}. The analytical challenges come from the  L\'{e}vy-type jumps of the filter process in \eqref{eq:filter_measure_nu}, as the law of jump measure $N(dt,dz)$'s compensator depending on the filter itself. To overcome this difficulty, we must introduce an auxiliary process through the Radon-Nikodym derivatives and give the necessary estimates under \ref{cond:BLR}. Results are summarized in Sect.~\ref{sec:cts}. \\
\indent {\it Step 2:  $\Lambda^M$ is a viscosity solution of the HJB PIDE \eqref{eq:constrainted_HJBPIDE1}}.  We adopt a classical definition (Definition~\ref{def:viscosity1}) of the viscosity solution and show ${\Lambda}^M$ is a  viscosity solution of the \eqref{eq:constrainted_HJBPIDE1} in Theorem~\ref{thm:visc} in Sect.~\ref{sec:visc}.\\ 
\indent {\it  Step 3: From PIDE to PDE.} Let $M$ be sufficiently large, we change the notation and rewrite the HJB PIDE \eqref{eq:constrainted_HJBPIDE1} as a parabolic PDE:  
\begin{align}\label{eq:HJBPDE}
	&\Big(\partial_t + \overline{\mu}(x)\partial_{x}+ \frac{1}{2}\overline{\sigma}(x)^2\partial_{xx}-d_0(x)\Big)g(t,x)+ \mathcal{I}_{\beta}[{\Lambda}^M](t,x) +1= 0, \text{ in } {\mathcal{U}}_T,
\end{align} 
where functions $d_0(\cdot)$ and $\mathcal{I}_{\beta}[\cdot]$ are defined as in \eqref{def:d}-\eqref{def:I}.\\
\indent {\it  Step 4: $\Lambda^M$ is a viscosity solution to PDE \eqref{eq:HJBPDE}.}  We consider viscosity solution $g$ of the PDE \eqref{eq:HJBPDE}, which is interpreted as an equation for an ``unknown'' $g$ with the last term $\mathcal{I}_{\beta}[{\Lambda}^M]$ prespecified with ${\Lambda}^M$ characterized in {Step 2}. We aim to demonstrate that ${\Lambda}^M$ also solves PDE \eqref{eq:HJBPDE} in viscosity sense. we must show the equivalence of two definitions of viscosity solutions to HJB PIDE \eqref{eq:constrainted_HJBPIDE1} (i.e, Definitions~\ref{def:viscosity1} and \ref{def:viscosity2}; the former is the classical one while the latter has no replacement of the solution by a test function in the nonlocal integro-differential term associated with the jumps). The results are presented in Proposition~\ref{prop:visco_def_eq} and Corollary~\ref{cor:visco_pde}.\\
\indent {\it  Step 5: Uniqueness of the viscosity solution to the PDE \eqref{eq:HJBPDE}.} It is clear that $g = {\Lambda}^M$ is a viscosity solution for both the PDE \eqref{eq:HJBPDE} and PIDE \eqref{eq:constrainted_HJBPIDE1}, as the two equations are  essentially the same. However, if a function $g$ solves the PDE \eqref{eq:HJBPDE}, it does not mean that $g$ also solves the PIDE \eqref{eq:constrainted_HJBPIDE1}, because the term $\mathcal{I}_{\beta}[{\Lambda}^M]$ in the PDE \eqref{eq:HJBPDE} depends on ${\Lambda}^M$ regardless of the choice of $g$. Thus, we must show that PDE \eqref{eq:HJBPDE} admits a unique viscosity solution. This requires applying a comparison result for viscosity solutions to HJB equations with degenerate coefficients on the boundary, and we cite the relevant result from \cite{amari2012differential}. \\
\indent {\it  Step 6: Existence of a classical solution to the PDE \eqref{eq:HJBPDE}.} The PDE \eqref{eq:HJBPDE} is a parabolic type with $\mathcal{I}_{\beta}[{\Lambda}^M]$ regarded as an autonomous term. We refer to the literature on degenerate parabolic PDE (see e.g. \cite{fleming2012deterministic,bayraktar2012valuation}) to show the existence of a classical solution to the PDE \eqref{eq:HJBPDE}. The result is presented in Theorem~\ref{thm:exist}. \\
\indent Results in Steps 3 - 6 are summarized in Sect.~\ref{sec:PDE}. Finally, we conclude that ${\Lambda}^M$ is a classical ($C(\overline{\mathcal{U}}_T)\cap C^{1,2} (\mathcal{U}_{T})$) solution of \eqref{eq:constrainted_HJBPIDE1}, together with Lemma~\ref{lem:main}, the proof of Theorem~\ref{thm:main} and Theorem~\ref{thm:1} is complete.  
\subsection{Lipschitz continuity of auxiliary constrained dual value function ${\Lambda}^M$}\label{sec:cts}
\indent We first show the Lipschitz continuity of ${\Lambda}^M(t,x)$ in the state variable $x$. Without loss of generality, we consider the case in which $t=0$. Unlike \cite{frey2012portfolio}, where the authors reformulate the dynamics of the filter process to an exogenous Poisson random measure while maintaining the original filter process's law. We now establish other necessary estimates of the value function by introducing an auxiliary process through the Radon–Nikodym derivatives. This method effectively enables us to work under general alternative data signals satisfying \ref{cond:BLR}.\\
\indent  We denote by $D_T:=D([0,T],[0,1])$ the path space of $(\pi_{t})_{t\in[0,T]}$, and $\mathcal{D}_T$ the usual $\sigma$-field of $D_T$. $P_1$ denotes the probability distribution on $(D_T,\mathcal{D}_T)$ induced by $(\pi_{t})_{t\in[0,T]}$ under $\widetilde{\mathbb{P}}^{0,x,\nu}$ for a given control process $\nu \in \Theta^{0,M}$. Standard arguments show that, with $\mathcal{L}^{\nu}$ defined in \eqref{eq:generator_L_nu}, the functional
\begin{align}\label{def:kg1}
	K_g(t): = g(\pi_t) - g(x) - \int_{0}^{t}\mathcal{L}^{\nu_s} g(\pi_s) ds
\end{align}
is a martingale under $P_1$ for each point $x\in[0,1]$ and each function $g(x)\in C^2([0,1])$, and $P_1$ is the unique such probability distribution.\\
\indent We introduce an auxiliary process $\Upsilon_t$ under a reference probability measure $\overline{\mathbb{P}}$ that satisfies the following stochastic integro-differential equation:
\begin{align}\label{eq:Upsilon}
	d\Upsilon_t = \overline{\mu}(\Upsilon_t)dt+\overline{\sigma}(\Upsilon_t) dW_t^{\beta} + \int_{\mathcal{Z}} (\xi(\Upsilon_{t-},z)-\Upsilon_{t-}) N_2(dt,dz),
\end{align}
where functions $\overline{\mu}$, $\overline{\sigma}$ are defined in \eqref{eq:filter_measure_nu}; $W_t^{\beta}$ is a standard Brownian motion and $N_2$ is a Poisson random measure with an intensity measure given by $\lambda f_1(z)dzdt$ under $\overline{\mathbb{P}}$. Note that $\Upsilon$ is a jump-diffusion process with an exogenous Poisson random measure. We denote the process $\Upsilon_t$ with initial condition $x$ by $\Upsilon_t^x$. 
To ensure that SDE \eqref{eq:Upsilon} has a unique strong solution for each control $\nu\in\Theta^{0,M}$, the coefficients $\overline{\mu}$, $\overline{\sigma}$ and $\xi$ must satisfy certain Lipschitz and growth conditions (see e.g. \cite{pham2001optimal}). We verify these conditions in the following Lemma, with proof in Appendix~\ref{A.proof_Prop1}.
\begin{lemma}{(Lipschitz and growth conditions)}\label{lem:2} Under \ref{cond:BLR}, 
	there is positive constant $C$ and function $\rho:\mathcal{Z} \rightarrow \mathbb{R}_+$ with $\int_{\mathcal{Z}}\rho(z)^2 f_1(z)dz <\infty$ such that for all $x$, $y \in [0,1]$,
	\begin{align}
		&\left|\overline{\mu}\left(x\right)-\overline{\mu}\left(y\right)\right|+\left|\overline{\sigma}\left(x\right)-\overline{\sigma}\left(y\right)\right| \leq C \left|x-y\right|, ~|\overline{\mu}(x)|+\left|\overline{\sigma}(x)\right| \leq C(1+|x|),\label{cond:1}\\
		&\left|\xi\left(x, z\right)-\xi\left(y, z\right)\right| \leq \rho(z)\left|x-y\right|, ~|\xi(x, z)| \leq (1+|x|).\label{cond:3}
	\end{align}
\end{lemma}
We denote by $P_2$ the probability distribution of $(D_T,\mathcal{D}_T)$ induced by $(\Upsilon_{t})_{t\in[0,T]}$ under $\overline{\mathbb{P}}$. We show that $P_1$ is absolutely continuous with respect to $P_2$ and that the corresponding Radon–Nikodym derivative has the following form:
\begin{align}\label{eq:Xi_T}
	&\Xi_T(\Upsilon):= \frac{dP_1}{dP_2}(\Upsilon)=\prod_{i=1}^{n(T)}\frac{e^{\beta\nu(\tau_i,z_i)}\hat{f}(\Upsilon_{\tau_{i}},z_{i})}{f_1(z_{i})}\\
	&\exp\Big(-\sum_{i=0}^{n(T)}\int_{\tau_i}^{\tau_{i+1} \wedge T} \lambda[\Upsilon_s\mathbf{E}_1(e^{\beta\nu})+(1-\Upsilon_s)\mathbf{E}_2(e^{\beta\nu})-1 ]ds \Big),\nonumber
\end{align}
where $\mathbf{E}_i$ are expectation operators on $z\in\mathcal{Z}$ under density function $f_i$, $i=1,2$; $z_i$ are the sequence of jump size;
$\tau_i$ and $n(T)$ are the sequence of jump times and total jump times up to $T$ respectively: 
\begin{align*}
	\tau_0 = 0, \quad \tau_{i+1}=\inf\{s>\tau_i: \Upsilon_s \neq \Upsilon_{s-} \} \text{ and } n(T) = \max\{i:\tau_i\le T\}.
\end{align*}
Note that for $t>0$, when capital $T$ in \eqref{eq:Xi_T} is substituted by small $t$, we have
\begin{align*}
	\Xi_t(\Upsilon) - 1 = \int_{0}^{t}\int_{\mathcal{Z}} \Xi_{s-}(\Upsilon)(\frac{e^{\beta\nu(s,z)}\hat{f}(\Upsilon_{s-},z)}{f_1(z)}-1)\widetilde{N}_2(ds,dz),
\end{align*}
where $\widetilde{N}_2(ds,dz)= N_2(ds,dz)-\lambda f_1(z)dzdt$ is the compensated Poisson random measure under $\overline{\mathbb{P}}$. The operator $\widetilde{\mathcal{L}}^{\nu}$ associated with $\Upsilon$ is given by
\begin{align*}
	\widetilde{\mathcal{L}}^{\nu} g(x):=& \overline{\mu}(x) g'(x)+\frac{1}{2} \overline{\sigma}^2(x) g''(x) +\int_{\mathcal{Z}}\left\{g(\xi(x, z))-g(x)\right\}\lambda f_1(z) d z.
\end{align*}
It follows that the functional
$
\widetilde{K}_g(t): = g(\Upsilon_t) - g(x) - \int_{0}^{t}\widetilde{\mathcal{L}}^{\nu} g(\Upsilon_s) ds
$
is a martingale under $P_2$ for each point $x\in[0,1]$ and each function $g(x)\in C^2([0,1])$, and $P_2$ is the unique such probability distribution. Replacing $\pi$ by $\Upsilon$ in $K_g(t)$ defined in \eqref{def:kg1} and applying an integration by parts, we have
\begin{align*}
	&\Xi_t K_g(t) = \int_{0}^{t} K_g(s-)d\Xi_s + \int_{0}^{t}\Xi_{s-}d\widetilde{K}_g(s)+ \int_{0}^{t}\Xi_{s-}(dK_g(s)-d\widetilde{K}_g(s))\\
	&+\sum_{s \le t}(\Xi_s - \Xi_{s-})(K_g(s) - K_g(s-))\\
	&=\int_{0}^{t} K_g(s-)d\Xi_s + \int_{0}^{t}\Xi_{s-}d\widetilde{K}_g(s)\\
	&+\int_{0}^{t} \int_{\mathcal{Z}} \Xi_{s-}\Big[\frac{e^{\beta \nu(s,z)}\hat{f}(\Upsilon_{s-},z)}{f_1(z)} -1 \Big](g(\xi(\Upsilon_{s-},z))-g(\Upsilon_{s-}))\widetilde{N}_2(ds,dz).
\end{align*}
As both $\Xi_t$ and $\widetilde{K}_g(t)$ are martingales under $P_2$, it follows that $\Xi_tK_g(t)$ is a martingale under $P_2$. Now for each $A \in D_T$, we set $\widetilde{P}_1(A) = \int_{A} \Xi_T(\Upsilon) dP_2$. We can clearly see that $K_g(t)$ is a martingale under $\widetilde{P}_1$. Because of uniqueness, we conclude that $\widetilde{P}_1 = P_1$. \\
\indent  Having established the preparatory results above, now we provide the main result of this subsection. For the sake of  definiteness, we denote by $\pi^x$ and $\Upsilon^x$ the solutions to \eqref{eq:filter_measure_nu} and \eqref{eq:Upsilon}, respectively, starting from $x$. 
\begin{theorem}\label{thm:lipctsx}
	The  value function ${\Lambda}^M(t,x)$ is Lipschitz continuous in $x$.
\end{theorem} 
\begin{proof}  For $x,y\in[0,1]$, we have
	\begin{align*}
		&\Big\vert {\Lambda}^M(0,x)-\Lambda^M(0,y) \Big\vert \\
		& \le \sup_{\nu\in\Theta^{0,M}}  \Big \vert \widetilde{\mathbb{E}}^{0,x,\nu}[ \exp(\int_{0}^{T}\Gamma( \pi_u^{x},\nu)du )]  -\widetilde{\mathbb{E}}^{0,y,\nu}[\exp(\int_{0}^{T}\Gamma( \pi_u^{y},\nu)du  )  ]  \Big\vert \\
		&\quad + \Big| \widetilde{\mathbb{E}}^{0,x,\nu}[\int_{0}^{T} \exp(\int_{0}^{t}\Gamma( \pi_u^{x},\nu)du  ) dt ] - \widetilde{\mathbb{E}}^{0,y,\nu}[\int_{0}^{T} \exp(\int_{0}^{t}\Gamma( \pi_u^{y},\nu)du  ) dt ]\Big|\\ 
		&\le \sup_{\nu\in\Theta^{0,M}}   \overline{\mathbb{E}}   \Big\{ \mathbf{A}_T+\mathbf{B}_T + \int_{0}^{T} \mathbf{A}_t dt + \int_{0}^{T} \mathbf{B}_t dt  \Big\},
	\end{align*}
	with
	\begin{align*}
		&\mathbf{A}_t: = \left\vert \exp\Big[\int_{0}^{t}\Gamma( \Upsilon_u^{x},\nu)du\Big] \left(\Xi_t(\Upsilon^{x})- \Xi_t(\Upsilon^{y})\right)  \right\vert\\
		&\mathbf{B}_t:=  \Xi_t(\Upsilon^{y}) \left\vert \exp\Big[\int_{0}^{t}\Gamma( \Upsilon_u^{x},\nu)du\Big]  - \exp\Big[\int_{0}^{t}\Gamma( \Upsilon_u^{y},\nu)du\Big]  \right\vert.
	\end{align*}
	We first focus on term $\mathbf{A}$. As $\nu\in[-M,M]$, the function $\Gamma$ is bounded, and therefore, for some constant $R_A>0$,
	$	\overline{\mathbb{E}} (\mathbf{A}_T) \le R_A \overline{\mathbb{E}}  \left\vert \Xi_T(\Upsilon^{x})- \Xi_T(\Upsilon^{y})\right\vert.
	$
	Using the inequality: 
	$
	\left|\prod_{i=1}^{n} a_{i}-\prod_{i=1}^{n} b_{i}\right| \leq n\left(\max _{1 \leq i \leq n}\left\{a_{i}, b_{i}\right\}\right)^{n-1} \max _{1 \leq i \leq n}\left|a_{i}-b_{i}\right|
	$
	for any two positive sequences of $\left\{a_{i}\right\}_{i=1}^{n}$ and $\left\{b_{i}\right\}_{i=1}^{n}$, we obtain that
	\begin{align*}
		& \overline{\mathbb{E}} \left\vert \Xi_T(\Upsilon^{x}) - \Xi_T(\Upsilon^{y}) \right\vert \\
		&\le e^{\lambda T} \overline{\mathbb{E}}  \bigg\vert \prod_{i=1}^{n(T)}\frac{e^{\beta\nu(\tau_i,z_i)}\hat{f}(\Upsilon^{x}_{\tau_i},z_i)}{f_1(z_i)}\exp\Big(- \int_{0}^{T} \lambda[\Upsilon^{x}_s\mathbf{E}_1(e^{\beta\nu})+(1-\Upsilon^{x}_s)\mathbf{E}_2(e^{\beta\nu})]ds \Big)\\
		&\quad-
		\prod_{i=1}^{n(T)}\frac{e^{\beta\nu(\tau_i,z_i)}\hat{f}(\Upsilon^{y}_{\tau_i},z_i)}{f_1(z_i)}\exp\Big(- \int_{0}^{T} \lambda[\Upsilon^{y}_s\mathbf{E}_1(e^{\beta\nu})+(1-\Upsilon^{y}_s)\mathbf{E}_2(e^{\beta\nu}) ]ds  \Big) \bigg\vert\\
		&\le e^{\lambda T}  \overline{\mathbb{E}} \sum_{k=1}^{\infty}k(e^{\beta M}(1+b_{\max}))^{k-1}\mathds{1}_{n(T)=k-1}\max\Big\{\frac{e^{\beta\nu(\tau_i,z_i)}}{f_1(z_i)}\\
		&\quad \Big|\hat{f}(\Upsilon^{x}_{\tau_i}, z_i) -\hat{f}(\Upsilon^{y}_{\tau_i},z_i) \Big|, \Big|\exp(- \int_{0}^{T} \lambda[\Upsilon^{x}_s\mathbf{E}_1(e^{\beta\nu})+(1-\Upsilon^{x}_s)\mathbf{E}_2(e^{\beta\nu})]ds  \\
		&\quad-  \exp(- \int_{0}^{T} \lambda[\Upsilon^{y}_s\mathbf{E}_1(e^{\beta\nu})+(1-\Upsilon^{y}_s)\mathbf{E}_2(e^{\beta\nu}) ]ds \Big|, 1\le i \le k\Big\} \\
		&\le e^{\lambda T}  \overline{\mathbb{E}} \sum_{k=1}^{\infty}k(e^{\beta M}(1+b_{\max}))^{k-1}\mathds{1}_{n(T)=k-1} R_M \sup_{0\le s \le T}\left\vert \Upsilon^{x}_s - \Upsilon^{y}_s \right\vert,
	\end{align*}
	where the constant $R_M = 2\max(\lambda T,1+b_{\max})e^{\beta M}$. In the last inequality, we use the fact that $|\exp(-a)-\exp(-b)|\le |a-b|$ for any bounded $a,~b$. It is also clear the term in the last line is finite because $\Upsilon_t$ always stays in $[0,1]$. From Cauchy–Schwarz inequality, we further obtain
	\begin{align*}
		&\overline{\mathbb{E}} \left\vert \Xi_T(\Upsilon^{x}) - \Xi_T(\Upsilon^{y}) \right\vert \\
		& ~\le  e^{\lambda T}R_M\sum_{k=1}^{\infty}k(e^{\beta M}(1+b_{\max}))^{k-1}  \overline{\mathbb{E}}  [\mathds{1}_{n(T)=k-1}]\Big(  \overline{\mathbb{E}}  [\sup_{0\le s \le T}\left\vert \Upsilon^{x}_s - \Upsilon^{y}_s \right\vert^2 ] \Big)^{1/2}.
	\end{align*}
	Recall that $n(T)$ counts the total number of jumps of a Poisson process with the constant intensity rate $\lambda$ prior to $T$, it follows that for $C_M= (\lambda Te^{\beta M}(1+b_{\max})+1)$,
	\begin{align*}
		\sum_{k=1}^{\infty}k(e^{\beta M}(1+b_{\max}))^{k-1} \overline{\mathbb{E}} [\mathds{1}_{n(T)=k-1}] =  C_M\exp[\lambda T (e^{\beta M}(1+b_{\max})-1) ].
	\end{align*}
	It remains to show that there exists a constant $C>0$ such that 
	\begin{align}\label{eq:ycont}
		\overline{\mathbb{E}}  \Big[\sup_{0\le s \le T}\left\vert \Upsilon^{x}_s - \Upsilon^{y}_s \right\vert^2   \Big] \le C |x-y|^2.
	\end{align}
	Note that
	\begin{align*}
		d(\Upsilon^{x}_t - \Upsilon^{y}_t)=& \left(\overline{\mu}(\Upsilon^{x}_t) - \overline{\mu}(\Upsilon^{y}_t) \right) dt+ \left( \overline{\sigma}(\Upsilon^{x}_t)-\overline{\sigma}(\Upsilon^{y}_t)\right)dW_t^{\beta}\\
		&+ \int_{\mathcal{Z}} \left( \xi(\Upsilon_{t-}^x,z)-\xi(\Upsilon_{t-}^y,z) - \Upsilon_{t-}^x+\Upsilon_{t-}^y \right)N_2(dt,dz).
	\end{align*}
	Applying It\^{o}'s lemma to the function $|\Upsilon^{x}_t - \Upsilon^{y}_t|^2$ and  Corollary 2.12 in \cite{kunita2004stochastic}, we obtain a constant $C$ such that
	\begin{align*}
		\overline{\mathbb{E}}  \Big[\sup_{0\le s \le T}\left\vert \Upsilon^{x}_s - \Upsilon^{y}_s \right\vert^2   \Big]	\le& C\left\{  |x-y|^2 +  \overline{\mathbb{E}} \Big[ \int_{0}^{T} | \overline{\mu}(\Upsilon^{x}_t) - \overline{\mu}(\Upsilon^{y}_t)|^2 dt \Big]  \right.\\
		&+  \overline{\mathbb{E}} \Big[ \int_{0}^{T} | \overline{\sigma}(\Upsilon^{x}_t) - \overline{\sigma}(\Upsilon^{y}_t)|^2 dt \Big] \\
		&\left.+  \overline{\mathbb{E}} \Big[ \int_{0}^{T} \lambda| \xi(\Upsilon^{x}_t,z) - \xi(\Upsilon^{y}_t,z)|^2 f_1(z) dz dt \Big]\right\}.
	\end{align*}
	By the Lipshcitz conditions of $\overline{\mu}$, $\overline{\sigma}$, $\xi$ given by Lemma~\ref{lem:2}, we obtain the inequality
	\begin{align*}
		\overline{\mathbb{E}}  \Big[\sup_{0\le s \le T}\left\vert \Upsilon^{x}_s - \Upsilon^{y}_s \right\vert^2   \Big]
		\leq C'\Big\{|x-y|^{2}+\left(L_1+\lambda C_{\rho}' \right) \int_{0}^{T} \overline{\mathbb{E}} \Big[\sup _{0 \leq s \leq \tau}|\Upsilon_s^x -\Upsilon_s^y |^{2}\Big] d \tau \Big\},
	\end{align*}
	for some positive constant $C>0$ and $C_{\rho}': = \int_{\mathcal{Z}}\rho^2(z)f_1(z)dz$. Thus, from the Gronwall’s inequality, we obtain the desired inequality \eqref{eq:ycont}.\\
	\indent We next consider the term $\mathbf{B}$. From the Cauchy-Schwarz inequality, we obtain
	\begin{align*}
		\overline{\mathbb{E}}[\mathbf{B}_T] \le & \overline{\mathbb{E}}[\Xi_T(\Upsilon^{x})^2]  ^{1/2}\overline{\mathbb{E}}\Big[\Big\vert\exp(\int_{0}^{T}\Gamma( \Upsilon_u^{x},\nu)du)  - \exp(\int_{0}^{T}\Gamma( \Upsilon_u^{y},\nu)du)  \Big\vert^2\Big]^{1/2}\\
		\le & C_b\overline{\mathbb{E}}[\Xi_T(\Upsilon^{x})^2]  ^{1/2} \overline{\mathbb{E}} \Big[ \sup_{0\le s\le T}| \Gamma( \Upsilon_s^{x},\nu) - \Gamma( \Upsilon_s^{y},\nu)|^2 \Big]^{1/2}\\
		\le &  C_b^{\prime}\overline{\mathbb{E}}[\Xi_T(\Upsilon^{x})^2]  ^{1/2}\overline{\mathbb{E}}\Big[ \sup_{0\le s\le T}|\Upsilon^{x}_s - \Upsilon^{y}_s |^2 \Big]^{1/2},
	\end{align*}
	where in the second inequality, we use again $|\exp(-a)-\exp(-b)|\le |a-b|$; in the last inequality, we use that fact that $\Gamma(x,\nu)$ is Lipschitz continuous in the state variable $x$ for $\nu\in[-M,M]$. Recalling \eqref{eq:ycont}, it remains to show the following:
	\begin{align*}
		&	\overline{\mathbb{E}}[\Xi_T(\Upsilon^{x})^2]\\
		& \le e^{2\lambda T} 	\overline{\mathbb{E}}~\prod_{i=1}^{n(T)}\frac{e^{2\beta M}\hat{f}^2(\Upsilon^{x}_{\tau_i},z_i)}{f^2_1(z_i)}\exp\Big(- 2\int_{0}^{T} \lambda[\Upsilon^{x}_s\mathbf{E}_1(e^{\beta\nu})+(1-\Upsilon^{x}_s)\mathbf{E}_2(e^{\beta\nu})]ds \Big)\\
		& \le e^{2\lambda T} 	\overline{\mathbb{E}}\left[ (e^{2\beta M}(b_{\max}+1)^2)^{n(T)} \right]\\
		&\le \exp\left( \lambda T(e^{2\beta M}(b_{\max}+1)^2+1)\right) <+\infty.
	\end{align*}
	The above analysis can be easily extended to the other two terms $\int_0^T \mathbf{A}_t dt$ and $\int_0^T \mathbf{B}_t dt$. By the arbitrariness of $\nu \in \Theta^{0,M}$, we complete the proof. 
\end{proof}
\indent 	Next, we show the continuity of $\Lambda^M(t,x)$ in the time variable $t$. The following estimates of the filter process $\pi$ will be used, the proof is placed in Appendix~\ref{A.proof_Prop1}
\begin{proposition}\label{prop:pigrowth} For arbitrary  $\nu \in \Theta^{t,M}$, denote by $(\pi_s^{t,x,\nu})_{s\in[t,T]}$ the solution to \eqref{eq:filter_measure_nu} starting from $(t,x)\in \overline{\mathcal{U}}_T$. For any $k\in[0,2]$ and $0\le t \le s \le T$, there is a constant $C_{\pi}>0$ uniformly such that 
	\begin{align}
		&\widetilde{\mathbb{E}}^{t,x,\nu}\Big[ \sup\limits_{t \le u \le s}(1+|\pi_u^{t,x,\nu}|^k) \Big] \le C_{\pi}(1+|x|^k), \label{eq:prop5eq1}\\
		&\widetilde{\mathbb{E}}^{t,x,\nu}\Big[ \sup\limits_{t \le u \le s}|\pi_u^{t,x,\nu} - x|^k \Big] \le C_{\pi}(1+|x|^k)(s-t)^{\frac{k}{2}}.\label{eq:prop5eq2}
	\end{align}
\end{proposition}
\begin{proposition} \label{prop:cts} For $t,s\in[0,T]$ and $x,y\in[0,1]$, there is constant $C>0$:
	\begin{align*}
		|\Lambda^M(t,x) - \Lambda^M(s,y)|\le C[|s-t|^{\frac{1}{2}} + |x-y| ].
	\end{align*}
\end{proposition}
\begin{proof}
	Let $0\le t< s \le T$, applying Theorem~\ref{thm:DPPfeller} (a weaker form since the control set is closed for $\Lambda^M$), we obtain 
	\begin{align*}
		& |\Lambda^M(t,x) - \Lambda^M(s,x)|\\
		\le & \sup\limits_{\nu\in\Theta^{t,M}} \widetilde{\mathbb{E}}^{t,x,\nu} \Big\vert e^{\int_{t}^{s}\Gamma(\pi_u^{t,x,\nu},\nu)du}\Lambda^M(s,\pi^{t,x,\nu}_s) +\int_{t}^{s} e^{\int_{t}^{\tau}\Gamma(\pi_u^{t,x,\nu},\nu)du} d\tau -  \Lambda^M(s,x) \Big\vert\\
		\le &  \sup\limits_{\nu\in\Theta^{t,M}}  \widetilde{\mathbb{E}}^{t,x,\nu}\Big( e^{\int_{t}^{s}\Gamma(\pi_u^{t,x,\nu},\nu)du}\left|\Lambda^M(s,\pi^{t,x,\nu}_s) - \Lambda^M(s,x) \right| \\
		& + \Lambda^M(s,x) \left| e^{\int_{t}^{s}\Gamma(\pi_u^{t,x,\nu},\nu)du} - 1 \right| + \int_{t}^{s} e^{\int_{t}^{\tau}\Gamma(\pi_u^{t,x,\nu},\nu)du} d\tau \Big)=\text{(\rom{1})} + \text{(\rom{2})} + \text{(\rom{3})}.
	\end{align*}
	By virtue of the boundedness of $\Gamma$ for $\nu\in[-M,M]$, Lipschitz continuity of $\Lambda^M$ in $x$ by Theorem~\ref{thm:lipctsx}, and \eqref{eq:prop5eq2}, there is a  positive constant $C$ such that
	\begin{align*}
		\text{(\rom{1})} &\le C \sup_{\nu\in\Theta^{t,M}} \widetilde{\mathbb{E}}^{t,x,\nu}\Big[ |\pi_s^{t,x,\nu} - x| \Big] \le C(1+|x|)(s-t)^{\frac{1}{2}},\\
		\text{(\rom{2})}& \le  \Lambda^M(s,x) | e^{C(s-t)}-1| \le C|s-t|,\\
		\text{(\rom{3})}& \le C(s-t).
	\end{align*}
	Finally, we obtain
	$|\Lambda^M(t,x) - \Lambda^M(s,x)| \le (C+T^{\frac{1}{2}})|s-t|^{\frac{1}{2}}$, together with Theorem~\ref{thm:lipctsx}, the proof is complete.
\end{proof}
\subsection{ $\Lambda^M$ is a viscosity solution of the HJB PIDE \eqref{eq:constrainted_HJBPIDE1}}\label{sec:visc}
We adapt the notion of a viscosity solution introduced by \cite{barles2008second} to the case of integro-differential equations, which is based on the notion of a test function and interprets equation \eqref{eq:constrainted_HJBPIDE1} in a weaker sense. We focus on the case $\kappa<0$, and for  $0<\kappa<1$, we follow a similar argument.
\begin{definition}\label{def:viscosity1} {(viscosity solution (test functions))}
	A bounded function $g \in C(\overline{\mathcal{U}}_T)$ is a viscosity supersolution (subsolution) of equation \eqref{eq:constrainted_HJBPIDE1} if, for any bounded test function $\psi\in C^{1,2}(\overline{\mathcal{U}}_T)$ such that $(t_0,x_0)\in {\mathcal{U}}_T$ is a global minimum (maximum) point of $g-\psi$ with $g(t_0,x_0)=\psi(t_0,x_0)$, then we have
	\begin{align*}
		&(-{\partial_t}-\overline{\mu}(x_0)\partial_x-\frac{1}{2} \overline{\sigma}(x_0)^2\partial_{xx})\psi(t_0,x_0) -\max\limits_{\nu\in [-M,M]}  H_{\psi}(t_0,x_0,\nu) \geq 1 	\text{(resp. $\leq 1$)},\\
		&\text{ where }H_{\psi}(t,x,\nu): = \Gamma(x,\nu) \psi(t, x)+\int_{\mathcal{Z}}\left\{\psi(t,\xi(x, z))-\psi(t,x)\right\}\lambda e^{\beta\nu}\hat{f}(x,z) d z.
	\end{align*}
	A bounded function $g$  is a viscosity
	solution of \eqref{eq:constrainted_HJBPIDE1} if it is both a viscosity subsolution and supersolution of  \eqref{eq:constrainted_HJBPIDE1}. 
\end{definition}
We establish the following result.
\begin{theorem}\label{thm:visc}
	${\Lambda}^M$ is a bounded Lipshcitz continuous viscosity solution of HJB PIDE \eqref{eq:constrainted_HJBPIDE1} in ${\mathcal{U}}_T$ subject to the terminal condition ${\Lambda}^M(T,x)=1$, $x\in[0,1]$.
\end{theorem}
\begin{proof}
	{\it Step 1: Viscosity supersolution.} Let $(t_0,x_0)\in {\mathcal{U}}_T$ and $\psi\in C^{1,2}(\overline{\mathcal{U}}_T)$ such that
	$ 0 = (\Lambda^M-\psi)(t_0,x_0) = \min_{(t,x)\in \mathcal{U}_T}\left(  \Lambda^M(t,x)-\psi(t,x) \right)$,
	and hence $\Lambda^M \ge \psi$ on ${\mathcal{U}}_T$. Let $(t_k,x_k)$ be a sequence in ${\mathcal{U}}_T$ such that 
	$ \lim\limits_{k\rightarrow \infty} (t_k,x_k) = (t_0,x_0)$, 
	and define the sequence $\{ \varphi_k \}$ as $\varphi_k: = \Lambda^M(t_k,x_k) - \psi(t_k,x_k)$. 
	From the continuity of $\Lambda^M$ (Proposition~\ref{prop:cts}), we have $\lim\limits_{k\rightarrow \infty} \Lambda^M(t_k,x_k) = \Lambda^M(t_0,x_0)$, so $\lim\limits_{k\rightarrow \infty} \varphi_k =0 $.\\
	\indent Consider a given control ${\nu} \in \Theta^{t,M}$, denote the filter process (the solution to \eqref{eq:filter_measure_nu}) with the initial state $\pi_{t_k}^k = x_k$ by $\pi^k$, and define stopping times $\tau_k$  as
	\begin{align}\label{eq:betak}
		\tau_{k}: = \inf\{s>t_k: (s, \pi^k_s)\not\in[t_k,t_k+\beta_k) \times(x_k-\epsilon_0, x_k+\epsilon_0 )\cap \mathcal{U}_T  \}, 
	\end{align}
	for a given constant $\epsilon_0\in(0,1/2)$ and $\beta_k: = \sqrt{\varphi_k}\mathds{1}_{\varphi_k\neq0} + k^{-1}\mathds{1}_{\varphi_k= 0}$, so $\lim\limits_{k\rightarrow \infty} \tau_k = 0$. Using Theorem~\ref{thm:DPPfeller}, we obtain
	\begin{align*}
		\Lambda^M(t_k,x_k) 
		&\ge \widetilde{\mathbb{E}}^{t_k,x_k,\nu} \Big[ e^{\int_{t_k}^{\tau_k}\Gamma(\pi_u^{k},\nu_u)du}\psi(\tau_k,\pi^{k}_{\tau_k}) +\int_{t}^{\tau_k} e^{\int_{t}^{s}\Gamma(\pi_u^{k},\nu_u)du} ds \Big],
	\end{align*}
	and hence by the definition of $\varphi_k$, 
	\begin{align}\label{thm5.10pf1}
		\varphi_k\ge \widetilde{\mathbb{E}}^{t_k,x_k,\nu} \Big[ \zeta^k(\tau_k)\psi(\tau_k,\pi^{k}_{\tau_k}) - \psi(t_k,x_k) +\int_{t}^{\tau_k} \zeta^k(s) ds \Big],
	\end{align}
	where $\zeta^k(s):= \exp\{\int_{t_k}^{s}\Gamma(\pi_u^{k},\nu_u)du\}$. Applying It\^{o}'s lemma to $\zeta^k\psi$, we have 
	\begin{align*}
		&\zeta^k(\tau_k)\psi(\tau_k,\pi^{k}_{\tau_k})\\
		&\quad = \psi(t_k,x_k) + \int_{t_k}^{\tau_k} \Gamma(\pi_u^{k},\nu_u)\zeta^k(u) \psi(u,\pi^k_u)+ \zeta^k(u)\{ (\mathcal{L}^{\nu_u}+\partial_t)\psi(u,\pi^k_u)  \} du\\
		&\qquad+\int_{t_k}^{\tau_k}  \zeta^k(u)\overline{\sigma}(\pi_u^k)\partial_{x}\psi(u,\pi^k_u) dW^{\beta}_u\\
		&\qquad +\int_{t_k}^{\tau_k}  \zeta^k(u) \int_{\mathcal{Z}} \Big(\psi(u,\xi(\pi_{u-}^k,z)) - \psi(u, \pi_{u-}^k) \Big)\widetilde{N}^{\beta}(du,dz).
	\end{align*}
	By assumption, the last two terms are martingales under $\widetilde{\mathbb{P}}^{t_k,x_k,\nu}$. Thus,
	\begin{align*}
		&\widetilde{\mathbb{E}}^{t_k,x_k,\nu} \left[ \zeta^k(\tau_k)\psi(\tau_k,\pi^{k}_{\tau_k}) \right] = \psi(t_k,x_k) \\
		&\quad + \widetilde{\mathbb{E}}^{t_k,x_k,\nu}  \Big[\int_{t_k}^{\tau_k} \Gamma(\pi_u^{k},\nu_u)\zeta^k(u) \psi(u,\pi^k_u)+ \zeta^k(u)\{ (\mathcal{L}^{\nu_u}+\partial_t)\psi(u,\pi^k_u)  \} du \Big].
	\end{align*}
	Recalling \eqref{thm5.10pf1}, we obtain
	\begin{align}\label{thm5.10pf2}
		\varphi_k \ge \widetilde{\mathbb{E}}^{t_k,x_k,\nu}  \Big[\int_{t_k}^{\tau_k} \zeta^k(u)\{ (\mathcal{L}^{\nu_u}+\partial_t)\psi(u,\pi^k_u) +1+\Gamma(\pi_u^{k},\nu_u) \psi(u,\pi^k_u) \} du  \Big].
	\end{align}
	We now aim to let $k\rightarrow \infty$, but we cannot directly apply the mean-value theorem as $ u \rightarrow g(u, \pi^k_u,\nu_u)$ is not continuous for function $g$ in general. We first show that the last term of the r.h.s of \eqref{thm5.10pf2} satisfies the following estimates, 
	\begin{align}
		&\widetilde{\mathbb{E}}^{t_k,x_k,\nu} \Big[\int_{t_k}^{\tau_k} \zeta^k(u)\Gamma(\pi_u^{k},\nu_u) \psi(u,\pi^k_u)du
		\Big]\nonumber \\
		&\quad \ge \widetilde{\mathbb{E}}^{t_k,x_k,\nu} \Big[\int_{t_k}^{\tau_k} \Gamma(x_k,\nu_u) \psi(t_k,x_k) du  \Big] -\beta_k \epsilon(\beta_k), \label{eq:thm5stp1eq1}
	\end{align}
	for some $\epsilon(\beta_k)\rightarrow 0$ as $\beta_k \rightarrow 0$, where $\beta_k$ is defined in \eqref{eq:betak}. By choosing a sufficiently small $\epsilon_0$, and from the local Lipschitz continuity of bounded continuous function $\psi\in C^{1,2}(\overline{\mathcal{U}}_T)$, we have
	\begin{align*}
		|\psi(u,\pi_u^k) - \psi(t_k,x_k)| \le C_{\epsilon_0}\left(|u-t_k|+|\pi_u^k-x_k| \right),\quad \forall u\in [t_k,\tau_k].
	\end{align*}
	In addition, $\Gamma(x,\nu)$ is bounded and Lipschitz in $x$ for $\nu \in \Theta^{t,M}$ and therefore 
	\begin{align*}
		&|\Gamma(\pi^k_u,\nu_u) - \Gamma(x_k,\nu_u)| \le C |\pi_u^k - x_k|,\\
		& |\zeta^k(u) - 1| \le C |u-t_k|\sup\limits_{t_k \le s \le u}|\pi_s^k - x_k|.
	\end{align*}
	We denote by $	\| \psi \|$ the uniform norm of the function $\psi$, and we have
	\begin{align*}
		&\widetilde{\mathbb{E}}^{t_k,x_k,\nu}\Big[\int_{t_k}^{\tau_k} \zeta^k(u)\Gamma(\pi_u^{k},\nu_u) \psi(u,\pi^k_u) du
		\Big]\\
		&\quad  \ge \widetilde{\mathbb{E}}^{t_k,x_k,\nu}\Big[\int_{t_k}^{\tau_k} \Gamma(\pi^k_u,\nu_u) \zeta^k(u) \psi(t_k,x_k) du \Big]\\
		&\qquad-C_{\epsilon_0}\beta_k\left\{ \beta_k +  \widetilde{\mathbb{E}}^{t_k,x_k,\nu}\Big[ \sup\limits_{t_k \le u \le \tau_k}|\pi_u^{k} - x_k| \Big] \right\}\\
		&\quad \ge \widetilde{\mathbb{E}}^{t_k,x_k,\nu}\Big[\int_{t_k}^{\tau_k} \Gamma(x_k,\nu_u)  \psi(t_k,x_k) du \Big] \\ 
		& \qquad -C_{\epsilon_0}'	\| \psi \|\beta_k\left\{ \beta_k +  C'	\| \psi \|\widetilde{\mathbb{E}}^{t_k,x_k,\nu}\Big[ \sup\limits_{t_k \le u \le \tau_k}|\pi_u^{k} - x_k| \Big] \right\}.
	\end{align*}
	Together with \eqref{eq:prop5eq2} in Proposition~\ref{prop:pigrowth} gives \eqref{eq:thm5stp1eq1}.
	By using the continuity of $\partial_t \psi$, $\partial_{x}\psi$, and $\partial_{xx}\psi$, as well as \eqref{cond:1} in Lemma~\ref{lem:2}, similar arguments give
	\begin{align*}
		&\widetilde{\mathbb{E}}^{t_k,x_k,\nu} \Big[\int_{t_k}^{\tau_k} \zeta^k(u)\{\partial_t + \overline{\mu}(\pi^k_u)\partial_x+ \frac{1}{2}\overline{\sigma}^2(\pi^k_u)\partial_{xx}\}\psi(u,\pi^k_u) du  \Big] \\
		&\quad \ge \widetilde{\mathbb{E}}^{t_k,x_k,\nu} \Big[\int_{t_k}^{\tau_k} \{\partial_t + \overline{\mu}(x_k)\partial_x+ \frac{1}{2}\overline{\sigma}^2(x_k)\partial_{xx}\}\psi(t_k,x_k) du  \Big] -\beta_k \epsilon(\beta_k).
	\end{align*}
	Next, using \eqref{cond:1}, we have
	\begin{align*}
		&|	\psi(u,\xi(\pi^k_{u-},z)) - \psi(t_k, \xi(x_k,z)) | \\
		&\quad \le C_{\epsilon_0}\Big(|u-t_k| + | \xi(\pi^k_{u-},z) - \xi(x_k,z))  | \Big)\\
		&\quad \le C_{\epsilon_0}\Big(|u-t_k| + (\rho(z)+1)|\pi^k_{u-} - x_k|  \Big), ~\forall u\in [t_k,\tau_k].
	\end{align*}
	Note that $C_{\rho}: = \int_{\mathcal{Z}} \rho(z) (f_1(z) +f_2(z)) dz<\infty $ under \ref{cond:BLR}, therefore
	\begin{align*}
		&\widetilde{\mathbb{E}}^{t_k,x_k,\nu} \Big[ \int_{t_k}^{\tau_k}\zeta^k(u)\lambda\int_{\mathcal{Z}} \left\{\psi(u,\xi(\pi^k_{u-},z)) - \psi(u,\pi^k_{u-})\right\} \hat{f}(\pi^k_{u-},z)e^{\beta\nu_u} dz  du \Big] \\
		&\quad \ge \widetilde{\mathbb{E}}^{t_k,x_k,\nu} \Big[ \int_{t_k}^{\tau_k}\lambda\int_{\mathcal{Z}} \left\{\psi(t_k,\xi(x_k,z)) - \psi(t_k,x_k)\right\} \hat{f}(x_k,z)e^{\beta\nu_u} dz  du \Big] \\
		&\qquad - \beta_kC\bigg(\beta_k + C_{\rho}\widetilde{\mathbb{E}}^{t_k,x_k,\nu}\Big[ \sup\limits_{t_k \le u \le \tau_k}|\pi_u^{k} - x_k| \Big]  \bigg).
	\end{align*}
	By substituting these estimates back into \eqref{thm5.10pf2}, we have
	\begin{align*}
		\frac{\varphi_k}{\beta_k} \ge &\frac{1}{\beta_k}\widetilde{\mathbb{E}}^{t_k,x_k,\nu}  \Big[ \int_{t_k}^{\tau_k} \{{\partial_t}+\overline{\mu}(x_k)\partial_x +\frac{1}{2} \overline{\sigma}(x_k)^2\partial_{xx}\} \psi(t_k,x_k)+1\\
		&+H_{\psi}(t_k,x_k,\nu_u) du  \Big]- \epsilon(\beta_k).
	\end{align*} 
	Finally, we set $k\rightarrow \infty$, $t_k \rightarrow t_0$, ${\varphi_k}/{\beta_k}  \rightarrow 0$, $\epsilon(\beta_k) \rightarrow 0$, and from the mean-value theorem, the bounded convergence theorem, and when replacing $\nu$ by a constant strategy, we have
	\begin{align*}
		\{\partial_t+\overline{\mu}(x_0)\partial_x +\frac{1}{2} \overline{\sigma}(x_0)^2\partial_{xx}\} \psi(t_0,x_0)+1 +H_{\psi}(t_0,x_0,\nu) \le 0. 
	\end{align*}
	As $\nu$ is arbitrary, we obtain the supersolution viscosity inequality. 
	\begin{align*}
		\{\partial_t+\overline{\mu}(x_0)\partial_x +\frac{1}{2} \overline{\sigma}(x_0)^2\partial_{xx}\} \psi(t_0,x_0)+1 +\max\limits_{\nu\in [-M,M]}H_{\psi}(t_0,x_0,\nu) \le 0.
	\end{align*}
	{\it Step 2: Viscosity subsolution.}  Let $(t_0,x_0)\in{\mathcal{U}}_T$ and $\psi\in C^{1,2}(\overline{\mathcal{U}}_T)$ such that
	$
	0 = (\Lambda^M-\psi)(t_0,x_0) = \max_{(t,x)\in{\mathcal{U}}_T}\left(  \Lambda^M(t,x)-\psi(t,x) \right)$, 
	and thus $\Lambda^M \le \psi$ on ${\mathcal{U}}_T$.
	We aim to establish the subsolution viscosity inequality in $(t_0,x_0)$. We argue by contradiction and assume that there is $\ell>0$, such that 
	\begin{align*}
		\{\partial_t+ \overline{\mu}(x_0)\partial_x +\frac{1}{2} \overline{\sigma}(x_0)^2\partial_{xx}\} \psi(t_0,x_0)+1 +\max\limits_{\nu\in [-M,M]} H_{\psi}(t_0,x_0,\nu) < -\ell < 0.
	\end{align*}
	As $\mathcal{L}^{\nu} \psi$ is continuous, there exists an open set $\mathcal{N}_{\epsilon_0}$ surrounding $(t_0,x_0)$ defined for $\epsilon_0\in(0,1/2)$ as
	$$\mathcal{N}_{\epsilon_0}:=\{ (t,x): (t,x)\in(t_0-\epsilon_0,t_0+\epsilon_0)\times(x_0 - \epsilon_0,x_0+\epsilon_0)\cap \mathcal{U}_T \} $$  
	and such that for $x\in\mathcal{N}_{\epsilon_0}$,
	\begin{align*}
		\{\partial_t+ \overline{\mu}(x)\partial_x +\frac{1}{2} \overline{\sigma}(x)^2\partial_{xx}\} \psi(t_0,x_0)+1 +\max\limits_{\nu\in [-M,M]} H_{\psi}(t,x,\nu) < -\frac{\ell}{2}. 
	\end{align*}
	We let $\iota>0$ be such that 
	\begin{align*}
		&\max\limits_{(t,x)\in{\mathcal{U}}_T\backslash \mathcal{N}_{\epsilon_0}} (\Lambda^M - \psi)(t,x) \le -\iota e^{-\epsilon_0 C_{\Gamma}}<0,
	\end{align*}
	where $C_{\Gamma}:=\max\limits_{x\in[0,1], \nu\in[-M,M]}(-\Gamma(x,\nu),0) < \infty$ by the boundedness of $\Gamma$.\\ 
	\indent Let $(t_k,x_k)$ be a sequence in $\mathcal{N}_{\epsilon_0}$ such that 
	$$ \lim\limits_{k\rightarrow \infty} (t_k,x_k) = (t_0,x_0),$$
	and define the sequence $\{ \varphi_k \}$ as $\varphi_k: = \Lambda^M(t_k,x_k) - \psi(t_k,x_k)$. By continuity of $\Lambda^M$ and $\psi$, we have $\lim\limits_{k\rightarrow \infty} \varphi_k = 0$. For all $k \ge 1$ and $\epsilon_k >0$ with $\lim\limits_{k\rightarrow \infty} \epsilon_k = 0$,  consider the $\epsilon_k$-optimal control $\nu^{*,k}$, such that 
	\begin{align}\label{thm5.10pf3}
		\Lambda^M(t_k,x_k) \le \Lambda(t_k,x_k,\nu^{*,k}) + \epsilon_k.
	\end{align}
	Denote the filter process (the solution to \eqref{eq:filter_measure_nu}) by $\widetilde{\pi}^k$, with the initial state given by $\widetilde{\pi}_{t_k}^k = x_k$ and the control given by $\nu = \nu^{*,k}$, and we define the stopping time 
	\begin{align}
		\tau_{k}: = \inf\{s>t_k: (s, \widetilde{\pi}^k_s)\not\in \mathcal{N}_{\epsilon_0} \}.
	\end{align}
	By definition, we have $\Lambda^M(\tau_{k},\widetilde{\pi}^k_{	\tau_{k}}) -\psi(\tau_{k},\widetilde{\pi}^k_{	\tau_{k}}) \le -\iota e^{-\epsilon_0 C_{\Gamma}} $.\\
	\indent Let $\widetilde{\zeta}^k(s):= \exp\{\int_{t_k}^{s}\Gamma(\widetilde{\pi}_u^{k},\nu^{*,k}_u)du\}$, we have
	\begin{align*}
		&\widetilde{\zeta}^k(\tau_k) \Lambda^M(\tau_{k},\widetilde{\pi}^k_{	\tau_{k}}) + \int_{t_k}^{\tau_k}\widetilde{\zeta}^k(s)ds - \Lambda^M(t_k,x_k)\\
		&\quad \le \widetilde{\zeta}^k(\tau_k) \psi(\tau_{k},\widetilde{\pi}^k_{	\tau_{k}}) + \int_{t_k}^{\tau_k}\widetilde{\zeta}^k(s)ds - \psi(t_k,x_k)  -\iota e^{-\epsilon_0 C_{\Gamma}} \widetilde{\zeta}^k(\tau_k) - \varphi_k\\
		&\quad \le \int_{t_k}^{\tau_k}  \widetilde{\zeta}^k(u)\Big((\partial_t+\mathcal{L}^{\nu^{*,k}_u})\psi(u,\widetilde{\pi}^k_u)+1\Big)du-\iota - \varphi_k.
	\end{align*}
	From the above calculations, we have 
	\begin{align*}
		&\widetilde{\mathbb{E}}^{t_k,x_k,\nu^{*,k}}\Big[  \widetilde{\zeta}^k(\tau_k) \Lambda^M(\tau_{k},\widetilde{\pi}^k_{	\tau_{k}}) + \int_{t_k}^{\tau_k}\widetilde{\zeta}^k(s)ds  \Big]\\	
		&\qquad \le \Lambda^M(t_k,x_k) -\iota - \varphi_k - \frac{\ell}{2} \widetilde{\mathbb{E}}^{t_k,x_k,\nu^{*,k}}[\tau_k - t_k].
	\end{align*}
	However, from the optimality of $\Lambda^M$ and \eqref{thm5.10pf3} we have
	\begin{align*}
		\widetilde{\mathbb{E}}^{t_k,x_k,\nu^{*,k}}\Big[  \widetilde{\zeta}^k(\tau_k) \Lambda^M(\tau_{k},\widetilde{\pi}^k_{	\tau_{k}}) + \int_{t_k}^{\tau_k}\widetilde{\zeta}^k(s)ds  \Big]	 \ge \Lambda^M(t_k,x_k)-\epsilon_k.
	\end{align*}
	By selecting $\epsilon_k = \varphi_k$, we have 
	$\Lambda^M(t_k,x_k)\le \Lambda^M(t_k,x_k) - \iota$, which is a contradiction, and therefore we have shown the subsolution inequality. 
\end{proof}
\subsection{ $\Lambda^M$ is a classical solution of HJB \eqref{eq:HJBPDE}.}\label{sec:PDE}
\indent We introduce an alternative definition of the viscosity solution first suggested by \cite{pham1998optimal} and formalized in various contexts as in \cite{barles2008second,davis2010impulse,seydel2010general}, and show that this alternative definition is equivalent to Definition~\ref{def:viscosity1}.
\begin{definition}[viscosity solution (test functions in the local terms only).]\label{def:viscosity2} A bounded function $g \in C(\overline{\mathcal{U}}_T)$ is a viscosity supersolution (subsolution) of equation \eqref{eq:constrainted_HJBPIDE1} if, for any bounded test function $\psi\in C^{1,2}(\overline{\mathcal{U}}_T)$ such that $(t_0,x_0)\in {\mathcal{U}}_T$ is a global minimum (maximum) point of $g-\psi$ with $g(t_0,x_0)=\psi(t_0,x_0)$, we have
	\begin{align*}
		&(-\partial_t-\overline{\mu}(x_0)\partial_x-\frac{1}{2} \overline{\sigma}(x_0)^2\partial_{xx})\psi(t_0,x_0)-\max\limits_{\nu\in [-M,M]}  H_{g}(t_0,x_0,\nu) \geq 1 	\text{(resp.$\leq 1$)},\\
		&\text{ where }H_{g}(t,x,\nu): = \Gamma(x,\nu) g(t, x)+\int_{\mathcal{Z}}\left\{g(t,\xi(x, z))-g(t,x)\right\}\lambda e^{\beta\nu}\hat{f}(x,z) d z.
	\end{align*}
	A bounded function $g$  is a viscosity
	solution of \eqref{eq:constrainted_HJBPIDE1} if it is both a viscosity subsolution and supersolution of  \eqref{eq:constrainted_HJBPIDE1}.
\end{definition}
\begin{proposition}\label{prop:visco_def_eq}
	Definitions~\ref{def:viscosity1} and \ref{def:viscosity2} of viscosity solutions are 
	equivalent.
\end{proposition}
The proof is placed in Appendix~\ref{A.proof_Prop1}.\\
Together with Theorem~\ref{thm:visc}, we immediately conclude the following corollary corresponding to Step 4 in proof of Theorem~\ref{thm:main} in Sect.~\ref{sec:4.1}.
\begin{corollary}\label{cor:visco_pde}
	The function ${\Lambda}^M$ is a viscosity solution of PDE~\eqref{eq:HJBPDE}. 
\end{corollary}
Given the results above, we formally define the functional $\mathcal{I}_{\beta}[g]$:
\begin{align*}
	\mathcal{I}_{\beta}[g](t,x): = (1-\beta)\lambda\int_{\mathcal{Z}} {g(t,x)}^{\frac{\beta}{\beta-1}} \left[{g(t,\xi(x,z))}^{\frac{1}{1-\beta}}-{g(t,x)^\frac{1}{1-\beta}}\right]\hat{f}(x,z)dz.
\end{align*}
Under \ref{cond:BLR}, $\mathcal{I}_{\beta}[g]$ is well defined for the bounded function $g$.  We observe that for $M$ sufficiently large (as \eqref{eq:Mvalue} in Lemma~\ref{lem:main}), 	$$\mathcal{I}_{\beta}[\Lambda^M](t,x) = \max\limits_{\nu\in [-M,M]} H_{\Lambda^M}(t,x,\nu) + d_0(x).$$
Thus, we rewrite the HJB PIDE \eqref{eq:constrainted_HJBPIDE1} as the equivalent parabolic PDE \eqref{eq:HJBPDE}, as stated in Step 3 in Sect.~\ref{sec:4.1}. We provide the following Lemma on $\mathcal{I}_{\beta}$, which will be used when we prove the uniqueness and existence of the classical solution to the PDE \eqref{eq:HJBPDE}, the proof is placed in Appendix~\ref{A.proof_Prop1}.
\begin{lemma}\label{lem:3}
	The functional $\mathcal{I}_{\beta}[{\Lambda}^M](t,x)$ is bounded and Lipschitz continuous in $x$ on $\overline{\mathcal{U}}_T$, therefore it is H\"{o}lder continuous in $x$ with some exponent $0<\iota<1$.
\end{lemma}
As described in Step 5, we cite the following comparison result for the degenerate parabolic PDE to demonstrate the uniqueness of the solution to PDE \eqref{eq:HJBPDE}.
\begin{theorem}[Theorem 2 in \cite{amadori2007uniqueness}]\label{thm:unique}
	We take $u$ and $v$ as a bounded upper semicontinuous subsolution and a bounded lower semicontinuous supersolution, respectively, to \eqref{eq:HJBPDE}, subject to the terminal condition $u(T,x) = v(T,x) = 1$, $x\in[0,1]$. Then $u \le v$ on ${\mathcal{U}}_T$.
\end{theorem}
By the virtue of Proposition~\ref{prop:filter_boundary} (the boundaries of the filter process are unattainable), Proposition~\ref{prop:Lambda_bound} (boundedness of $\Lambda^M$), and Lemma~\ref{lem:2} (Lipschitz and growth conditions of coefficients in PDE), we straightforwardly check that the assumptions in \cite{amadori2007uniqueness} (Assumptions 1 and 2 therein) are verified for our case. 
\begin{corollary}\label{cor:unique} The value function ${\Lambda}^M$ is the unique viscosity solution of the\\ parabolic PDE \eqref{eq:HJBPDE} subject to the terminal condition ${\Lambda}^M(T,x)=1$.
\end{corollary}
Following Step 6, it remains to establish the existence result of PDE \eqref{eq:HJBPDE}. We provide the following theorem, with proof placed in Appendix~\ref{A.proof_thm_exist}.
\begin{theorem}\label{thm:exist}
	The PDE \eqref{eq:HJBPDE} admits a classical solution $g\in C^{1,2}(\mathcal{U}_T)$ subject to  the terminal condition $g(T,x) = 1$.
\end{theorem}
%----Section 6. 
\section{Conclusion}\label{sec:conclu}
In this study, we establish the first duality approach to the optimal investment--consumption problem with partial information and mixed-type observations. Interestingly, the inclusion of alternative data makes our problem part of the family of incomplete markets, and our dual problem is an optimization problem over a set of equivalent local martingale measures. We  comprehensively demonstrate its application in a bull--bear market regime economy by drawing on expert opinions as a complementary observation source. The analytically tractable results for the power utility case show that the optimal investment and consumption policies are determined by the solution of a PIDE, which takes into account the effect of the alternative observations.

%----Appendix.
\appendix\normalsize
\section{Proof of Theorem~\ref{thm:duality_main}}\label{A.proof_thm1}
Without loss of generality, we prove the result for starting time $t=0$ and arbitrary initial guess $x\in[0,1]$. For notation convenience, we suppress the index $t$ and $x$ in $\mathcal{A}(t,x,v)$, $V^{t,x,v,\varpi,c}$, $\Theta^t$, $\widetilde{L}(t,x,y;\nu)$, and $\hat{L}(t,x,y)$.
\begin{assump}\label{ass:utilitu_gen} 	 The time-dependent utility functions $U_i(t,c)\in C^{2}([0,T]\times\mathbb{R}_+)$, $i=1,2$, has the following properties for any given $t\in[0,T]$,
	\begin{enumerate} 	
		\item[i.] $U_i(t,c)$ is strictly concave with respect to $c$ and 
		\begin{align}
			\lim_{c\rightarrow 0+} \partial_c U_i(t,c) = +\infty,\quad \lim_{c\rightarrow +\infty} \partial_c U_i(t,c) = 0.
		\end{align}
		
		%		\item[i.] $U_i$ satisfies the growth condition
		%		\begin{align*}
		%		|	U_i(t,c)| \le K_1(1+x^{\alpha}),
		%		\end{align*}
		%		for some constants $K_1\in\mathbb{R}_+$ and $\alpha\in(0,1)$.
		%\item[iii.] $U(t, 0+) > -\infty$ and $U(t,\infty) = \infty$.
		%\item[iv.] $x\rightarrow xU_i'(t,x)$ is nondecreasing on $\mathbb{R}_+$.
		\item[ii.]  There is $c_0>0$,  $\zeta \in(0,1)$ and $\iota >1$ such that  $$\zeta \partial_cU_i(t,c) \ge \partial_c U_i(t, \iota c) \text{ for }c>c_0.$$
		\item [iii.] 
		There are positive constants $K$ and $\hat{\kappa}$ such that
		\begin{align}\label{eq:ass3eq1}
			\limsup_{c\rightarrow+\infty} \max_{t \in [0,T]} \partial_c U_i(t,c) c^{\hat{\kappa}} \le K.
		\end{align}
	\end{enumerate}
\end{assump}
Assumption~\ref{ass:utilitu_gen}i. is referred to as Inada conditions that are commonly applied in economic models \cite{inada1963two,kramkov1999asymptotic}. The two growth conditions in Assumption~\ref{ass:utilitu_gen}ii. and \ref{ass:utilitu_gen}iii. are standard ones in duality approaches \cite{karatzas1991martingale,cvitanic1992convex} and used to straightforwardly obtain the regularity of various functions in the duality treatment. Assumption~\ref{ass:utilitu_gen} is satisfied by most popular utility functions, such as CRRA and the constant absolute risk aversion utilities such as
\begin{align*}
	U_i(t,c) = \frac{c^{\kappa^{\prime}}}{\kappa^{\prime}}( 0\neq \kappa^{\prime} <1 ), ~~	U_i(t,c) = \ln c, \text{ or }	U_i(t,c) = 1 - e^{-\kappa^{\prime\prime}c}(\kappa^{\prime\prime}>0). 
\end{align*}
It is also satisfied by the class of utility functions with time-varying risk aversion, for example, $U_i(t,c) = 
c^{\kappa(t)}/\kappa(t)$, where $\kappa(t)$ is a deterministic function satisfying $\underline{\kappa}_{l}\le \kappa(t) \le \underline{\kappa}_{u}$, with $\underline{\kappa}_{l}<\underline{\kappa}_{u}<1$ being constants such that $\underline{\kappa}_{l}>0$ or $\underline{\kappa}_{u}<0$. In some applications, economists are interested in the power utility; particularly the case of negative power which is considered to be more realistic from the standpoint of agent behavior but usually requires a different treatment, and thus is rarely discussed in literature \cite{federico2015utility}.
Assumption~\ref{ass:utilitu_gen}ii. is thought to cover the case of a power utility with a negative power, which is the case arising in our optimal control problem \eqref{def:controlprob}-\eqref{eq:utilit_form} considered in Sect.~\ref{sec:1}.\\
\indent In addition, if $U_i(t,\infty)>0$ for all $t\in[0,T]$, then Assumption~\ref{ass:utilitu_gen}ii. implies Assumption~\ref{ass:utilitu_gen}iii. and that the utility function has asymptotic elasticity strictly less than 1 \cite{kramkov1999asymptotic}. In other words, for a given $t$,
$
\mathrm{AE}(U_i(t,\cdot)):= \limsup_{c\rightarrow\infty}\frac{x\partial_cU_i(t,c)}{U_i(t,c)} <1. 
$
The asymptotic elasticity can be interpreted as the ratio of the marginal utility $\partial_cU$ to the average utility $U(c)/c$, for a large $x>0$. With a financial application in mind, we may think of an agent comparing her marginal utility from very large wealth/consumption levels with her current average utility. As noted in \cite{schachermayer2004portfolio}, if the limit of the coefficient of relative risk aversion, $\lim\limits_{c\rightarrow\infty}-c\partial_{cc}U/\partial_cU$ exists and is strictly positive; then $\mathrm{AE}(U)<1$. As non-increasing relative risk aversion is considered to be common among economic agents, it follows that these agents have asymptotic elasticity less than one.\\
\indent For convenience of exposition, we list properties of convex dual functions $\widetilde{U}_i$ and $I_i$ (inverse function of $\partial_cU_i(t,\cdot)$) for the above general class of utility functions.
\begin{lem}\label{prop:I} Under Assumption~\ref{ass:utilitu_gen}, convex dual functions $\widetilde{U}_i$ and functions $I_i$, $i=1,2$ have the following properties. For any $t\in[0,T]$, 
	\begin{itemize}
		\item[i.] $\widetilde{U}_i:[0,T]\times\mathbb{R}_+ \rightarrow \mathbb{R}$ is strictly decreasing, strictly convex and satisfies
		\begin{align*}
			&\partial_y \widetilde{U}_i(t,y) = - I_i(t,y), ~y\in\mathbb{R}_+,\\
			&U_i(t,c) = \inf_{y\in\mathbb{R}_+} [\widetilde{U}_i(t,y) + xy] = \widetilde{U}_i(t,\partial_cU_i(t, c))+c \partial_cU_i(t, c),~ c\in\mathbb{R}_+.
		\end{align*}
		\item[ii.] 
		
		$\lim\limits_{y\rightarrow 0+} I_i(t,y)= +\infty,~\lim\limits_{y\rightarrow +\infty} I_i(t,y) = 0.
		$
		\item[iii.] For some constant $L_{I}>0$ and $\hat{\kappa}$ the constant in Assumption~\ref{ass:utilitu_gen} iii., 
		\begin{align*}
			I_i(t,y) \le L_{I}(1+y^{-\frac{1}{\hat{\kappa}}}),\quad \forall (t,y)\in[0,T]\times\mathbb{R}_+.
		\end{align*}
		\item[iv.]There exists $y_0>0$ such that for any $\zeta\in(0,1)$, there is a constant $\iota>1$:
		\begin{align}\label{rem2:eq2}
			I_i(t, \zeta y) \le \iota  I_i(t,y),~\forall 0<y<y_0.
		\end{align}
	\end{itemize}
\end{lem}
With a slight abuse of notation, let $\Theta$ be the set collecting those process $\nu:=(\nu_D,\nu_J)$ associated with $\overline{Z}^{\mathbb{Q}}$ in \eqref{eq:z_gene} for $\mathbb{Q}\in\mathcal{Q}$, and write $\overline{Z}^{\mathbb{Q}}$ as $Z^{\nu}$. We introduce  the function $\chi(y;\nu):\mathbb{R}_+\times\Theta\rightarrow\mathbb{R}$: 
\begin{align}\label{def:dualchi}
	\chi(y;\nu):= \mathbb{E}[e^{-rT}Z_T^{\nu} I_1(T, e^{-rT}yZ_T^{\nu})+ \int_{0}^{T} e^{-rt}Z_t^{\nu}I_2(t, e^{-rt}yZ_t^{\nu})dt].
\end{align}
If the condition 
\begin{align}\label{cond:7}
	\chi(y;\nu) < \infty,~\forall y\in\mathbb{R}_+ 
\end{align}
prevails, the Monotone Convergence Theorem and the Dominated Convergence Theorem together with Lemma~\ref{prop:I}ii. imply that $\chi(\cdot;\nu)$ is continuous and 
\begin{align}\label{eq:chiasy}
	\lim_{y\rightarrow \infty} \chi(y;\nu) = 0, \quad \lim_{y\rightarrow 0+} \chi(y;\nu) = \infty,
\end{align}
and $\chi(\cdot;\nu)$ is strictly decreasing on $\mathbb{R}_+$, for given $\nu\in \Theta$. We verify that the condition \eqref{cond:7} holds under Assumption~\ref{ass:utilitu_gen}. By Lemma~\ref{prop:I}iii., for $y >0$, we have
\begin{align*}
	\chi(y;\nu) \le L_I  +L_I y^{-\frac{1}{\hat{\kappa}}} \mathbb{E}[ e^{\frac{rT}{\underline{\iota}}}(Z_T^{\nu})^{-\frac{1}{\underline{\iota}}}+ \int_{0}^{T} e^{\frac{rt}{\underline{\iota}}}(Z_t^{\nu})^{-\frac{1}{\underline{\iota}}} dt],\quad {\underline{\iota}}: = \frac{\hat{\kappa}}{1-\hat{\kappa}}.
\end{align*}
Due to the boundedness of $\hat{\theta}$, i.e., that of $\theta$ (recall that $\theta(\alpha_t)= (\mu(\alpha_t)-r)/\sigma$ is bounded as defined in \eqref{eq:theta}), together with $\nu\in\Theta$ satisfying conditions \eqref{cond:5} and \eqref{cond:6}, we obtain the estimate that $\chi(y;\nu)< \infty$. \\
\indent As assumed in Theorem~\ref{thm:duality_main}, for $y\in\mathbb{R}_+$, dual optimizer  $\nu^{y}\in \Theta$ of \eqref{def:gene_dual_L} exists:
\begin{align}\label{eq:prop4eq0}
	\hat{L}(y) = \inf_{\nu \in \Theta} \widetilde{L}(y;\nu) =  \widetilde{L}(y;\nu^y) < \infty.
\end{align}
We first establish a useful fact for $\nu^y$.
\begin{lem} \label{lem:1} For given $y\in\mathbb{R}_+$, let $\nu^{y}$ be the dual optimizer of \eqref{def:gene_dual_L} for $y$, 
	\begin{align}\label{eq:prop4eq4}
		\sup_{\nu\in\Theta}\mathbb{E}[e^{-rT}Z_T^{\nu} I_1(T, e^{-rT}yZ_T^{\nu^y})+ \int_{0}^{T} e^{-rt}Z_t^{\nu}I_2(t, e^{-rt}yZ_t^{\nu^y})dt] \le \chi(y;\nu^y).
	\end{align}
\end{lem}
\begin{proof}
	Fix $\epsilon \in (0,1)$, and arbitrary $\nu'=(\nu_D', \nu_J') \in \Theta$, we define 
	\begin{align*}
		&G_{\epsilon}(t)= (1-\epsilon) Z^{\nu^y}_t + \epsilon  Z^{\nu'}_t, \quad \nu^{\epsilon}_D(t) = G_{\epsilon}(t)^{-1} ((1-\epsilon)Z_t^{\nu^y}\nu^y_D(t) + \epsilon  Z_t^{\nu'}\nu'_D(t)),\\
		&\nu^{\epsilon}_J(t,q) = \ln \{G_{\epsilon}(t)^{-1} ((1-\epsilon)Z_t^{\nu^y}e^{\nu^y_J(t,q)} + \epsilon  Z_t^{\nu'}e^{\nu'_J(t,q)}) \}.
	\end{align*}
	Then $\nu^{\epsilon}:=(\nu^{\epsilon}_D, \nu^{\epsilon}_J) \in \Theta$ and we have 
	\begin{align*}
		d	G_{\epsilon}(t) = -\hat{\theta}(\pi_t) 	G_{\epsilon}(t) d\widetilde{W}_t -\nu^{\epsilon}_D(t) G_{\epsilon}(t) d\widetilde{B}_t - \int_{\mathbb{R}} G_{\epsilon}(t)(1-e^{\nu^{\epsilon}_J(t,q)}) \overline{m}^{\pi}(dt,dq).
	\end{align*}
	By comparing the solutions to the above SDE and \eqref{eq:z_gene}, we find that $G_{\epsilon}(t) = Z^{\nu^{\epsilon}}_t$ by the uniqueness of the Dol\'{e}an's exponential. Since $\nu^y$ is optimal, we have $$\epsilon^{-1}\left(\widetilde{L}(y;\nu^y)-\widetilde{L}(y;\nu^{\epsilon})\right)\le0,$$ or equivalently, 
	\begin{align}
		&\mathbb{E}\Big[ \epsilon^{-1}\big(\widetilde{U}_1(T, y e^{-rT}Z_T^{\nu^y}) + \int_{0}^{T} \widetilde{U}_2(t,  y e^{-rt}Z_t^{\nu^y}) dt\cr
		&-\widetilde{U}_1(T, y e^{-rT}Z_T^{\nu^{\epsilon}}) - \int_{0}^{T} \widetilde{U}_2(t,  y e^{-rt}Z_t^{\nu^{\epsilon}}) dt\big)\Big] \le 0.\label{eq:thm3eq1}
	\end{align}
	Recalling the fact $\partial_y\widetilde{U}_i(t,y) = - I_i(t,y)$, we see that \eqref{eq:prop4eq4} can be obtained from taking the limit as $\epsilon\downarrow0$ inside the expectation sign of \eqref{eq:thm3eq1}. For a rigorous justification, we show that the random variable inside the expectation operator in \eqref{eq:thm3eq1} is bounded from below by uniformly integrable terms. For a given $t\in[0,T]$, fix $\omega \in \Omega$, suppressing the dependence on $\omega$. Once 
	$Z_t^{\nu^{\prime}} > Z_t^{\nu^{y}}$, the mean-value theorem implies that
	\begin{align}\label{eq:prop4eq3}
		\epsilon^{-1}  \left(\widetilde{U}_i(t, y e^{-rt}Z_t^{\nu^y}) - \widetilde{U}_i(t, y e^{-rt}Z_t^{\nu^{\epsilon}}) \right) &=I_i(t,e^{-rt}y F)e^{-rt}y\epsilon^{-1}(G_{\epsilon}(t) - Z_t^{\nu^{y}})\nonumber\\
		&= I_i(t,e^{-rt}y F)e^{-rt}y(Z_t^{\nu^{\prime}} - Z_t^{\nu^{y}}),
	\end{align}
	where $ Z_t^{\nu^{y}} \le F \le  Z_t^{\nu^{y}} + \epsilon( Z_t^{\nu^{\prime}}- Z_t^{\nu^{y}}) <  Z_t^{\nu^{\prime}}$, and in the first equality, we use the fact that $G_{\epsilon}(t) = Z^{\nu^{\epsilon}}_t$. As $I_i(t,y)$ is decreasing in $y$, we obtain
	\begin{align}\label{eq:prop4eq5}
		&I_i(t,e^{-rt}y F)e^{-rt}y(Z_t^{\nu^{\prime}} - Z_t^{\nu^{y}}) \ge -e^{-rt}yI_i(t,e^{-rt}y Z_t^{\nu^{y}})Z_t^{\nu^{y}}.
	\end{align}
	Alternatively, if $Z_t^{\nu^{\prime}} < Z_t^{\nu^{y}}$, we have 
	\begin{align*}
		&\epsilon^{-1} \left(\widetilde{U}_i(t, y e^{-rt}Z_t^{\nu^y}) - \widetilde{U}_i(t, y e^{-rt}Z_t^{\nu^{\epsilon}}) \right) = 	I_i(t,e^{-rt}y F )e^{-rt}y(Z_t^{\nu^{\prime}} - Z_t^{\nu^{y}}),\\
		& \ge -I_i(t,e^{-rt}y(1-\epsilon)Z_t^{\nu^y})Z_t^{\nu^y}\\
		&=-I_i(t,e^{-rt}y(1-\epsilon)Z_t^{\nu^y})Z_t^{\nu^y}(\mathds{1}_{\left\{ y Z_t^{\nu^y} < y_0\right\}}+\mathds{1}_{\left\{ y Z_t^{\nu^y} \ge y_0\right\}})\\
		&\ge -\iota I_i(t,e^{-rt}yZ_t^{\nu^y})Z_t^{\nu^y} -I_i(t,e^{-rt}y_0(1-\epsilon)Z_t^{\nu^y})Z_t^{\nu^y},	\end{align*}
	where we use the fact that $ (1-\epsilon)Z_t^{\nu^y} \le (1-\epsilon)Z_t^{\nu^y} + \epsilon Z_t^{\nu^{\prime}}\le F \le Z_t^{\nu^y}$ in the second line; the last inequality holds for a sufficiently small $\epsilon$ in the sense that $\epsilon<1-\xi$, where the constants $y_0>0$, $\xi>0$ and $\iota>1$ defined in \eqref{rem2:eq2} are used.\\
	\indent Repeat the proof of \eqref{cond:7}, we obtain that  the random variable inside the expectation operator in \eqref{eq:thm3eq1} is bounded from below by uniformly integrable terms when $\epsilon$ is sufficiently small. As a result, Fatou's lemma can be applied when taking the limit as $\epsilon \downarrow 0$ in \eqref{eq:prop4eq3}, which implies \eqref{eq:prop4eq4} by the arbitrary choice of $\nu'$. 
\end{proof}
Next we show that $\nu^y$ leads to a pair of admissible strategy $(\varpi,c)\in\mathcal{A}(\chi(y;\nu^y))$, as a corollary of the following theorem and Lemma~\ref{lem:1}.
\begin{thm}\label{thm:budget} {(Budget constraint)} 
	Let $\mathcal{V}$ be a nonnegative $\mathcal{H}_T$-measurable random variable and $c_t$ a consumption rate process such that 
	\begin{align}\label{eq:budget}
		\sup_{\nu\in\Theta} \mathbb{E}[e^{-rT}Z_T^{\nu}\mathcal{V} + \int_{0}^{T} e^{-rt}Z_t^{\nu}c_t dt] \le v.
	\end{align}
	Then there exists an investment process $\varpi$ such that $(\varpi,c)\in\mathcal{A}(v)$ and $V_T^{\varpi,c} \ge \mathcal{V}$. 
\end{thm}
\begin{proof}
	The techniques are similar to those in \cite{pham2001optimal}. That is, to show for a given candidate terminal wealth level $\mathcal{V}$ and consumption plan $c$, the superhedging price (l.h.s of \eqref{eq:budget}) satisfies certain dynamic characterization and, therefore, the martingale representation theorem in Proposition~\ref{prop:martg_represen} can be applied. The differences are both the investment and consumption strategies are considered, and the filtration $\mathbb{H}$ includes both Brownian filtration and jump filtration in our analysis, so we omit the details of the proof due to space constraints.
\end{proof}
\noindent Recall \eqref{eq:prop4eq4} in Lemma~\ref{lem:1} and set $\mathcal{V}=I_1(T, e^{-rT}yZ_T^{\nu^y})$ and $c_t = I_2(t, e^{-rt}yZ_t^{\nu^y})$. Theorem~\ref{thm:budget} implies that there exists $(\varpi,c)\in\mathcal{A}(\chi(y;\nu^y))$ such that $V_T^{\varpi,c} \ge \mathcal{V}$.\\
\indent Finally, we aim to show that for every $v\in\mathbb{R}_+$, there exists $y^*= y(v) \in\mathbb{R}_+$ such that $v = \chi(y^*;\nu^{y^*})$, where $\nu^{y^*}$ is the dual optimizer in \eqref{def:gene_dual_L} for $y^*$. This statement is a corollary of the following result.
\begin{lem}\label{lem:1.5}The dual value function $\hat{L}(y)$ defined in \eqref{def:gene_dual_L} is continuously differentiable and its derivative
	$\partial_y\hat{L}(y) = - \chi(y;\nu^{y})$, for $y\in\mathbb{R}_+$ and $\nu^y$ as the dual optimizer for $y$. In addition, 
	\begin{align}\label{eq:derivativeasy}
		\lim_{y\rightarrow 0+}  \partial_y\hat{L}(y) = -\infty, \quad \lim_{y\rightarrow +\infty}  \partial_y\hat{L}(y) = 0.
	\end{align} 
\end{lem}
\begin{proof}By properties of $\widetilde{U}_i$ in Lemma~\ref{prop:I}i., $\hat{L}$ is clearly decreasing and convex in $y$. Firstly, we show that $\hat{L}$ is differentiable w.r.t $y$, and hence is continuously differentiable by its convexity. For a fixed $\bar{y}>0$, let $\nu^{\bar{y}}$ be the corresponding minimizer such that $\hat{L}(\bar{y}) = \widetilde{L}(\bar{y};\nu^{\bar{y}})$. We consider the function $\overline{L}(y) : = \widetilde{L}(y;\nu^{\bar{y}})$, which is also convex and decreasing in $y$. We have that $\overline{L}(y) \ge \hat{L}(y)$ for all $y\in\mathbb{R}_+$ and $\overline{L}(\bar{y}) = \hat{L}(\bar{y})$. It then follows that 
	$$ \partial_{-}\overline{L}(\bar{y})\le \partial_{-}\hat{L}(\bar{y}) \le \partial_{+}\hat{L}(\bar{y}) \le\partial_{+}\overline{L}(\bar{y}), $$
	where $\partial_{\pm}$ denote the left and the right derivatives respectively, their existence is guaranteed by the convexity of $\overline{L}$ and $\hat{L}$. 
	By Monotone Convergence Theorem and the fact that $\partial_y\widetilde{U}_i(t,y)=-I_i(t,y)$, we have 
	$ \partial_{+}\overline{L}(\bar{y}) \le -\chi(\bar{y};\nu^{\bar{y}}).$
	On the other hand, by convexity,  
	\begin{align*}
		\partial_{-} \overline{L}(\bar{y}) \ge& \limsup_{\epsilon\rightarrow0+} \mathbb{E}\Big[-e^{-rT}Z_T^{\nu^{\bar{y}}} I_1(T, e^{-rT}(\bar{y}-\epsilon)Z_T^{\nu^{\bar{y}}}\\
		&- \int_{0}^{T} e^{-rt}Z_t^{\nu^{\bar{y}}}I_2(t, e^{-rt}(\bar{y}-\epsilon)Z_t^{\nu^{\bar{y}}})dt\Big],\end{align*}
	where the term inside the expectation operator is uniformly integrable when $\epsilon$ is sufficiently small by following the same arguments of proving Lemma~\ref{lem:1}. We conclude that $\partial_{-} \overline{L}(\bar{y}) \ge -\chi(\bar{y};\nu^{\bar{y}})$. Hence $\partial_y\hat{L}(\bar{y}) = - \chi(\bar{y};\nu^{\bar{y}})$, for all $\bar{y}>0$.\\
	\indent Next, we prove \eqref{eq:derivativeasy}. For notation convenience, we write $\phi(0+): = \lim_{y\rightarrow 0+} \phi(y)$ for function $\phi$.
	Notice that $\hat{L}(0+) \ge \widetilde{U}_1(T,0+)+ \int_{0}^{T}\widetilde{U}_2(t,0+)dt $.  This follows by using Jensen's inequality, the convexity and decreasing properties of $\widetilde{U}_i(t,\cdot)$, and the (super)martingale property of $Z^{\nu}$ for arbitrary $\nu\in \Theta$:
	\begin{align*}
		\widetilde{L}(y;\nu) &\ge \widetilde{U}_1(T,ye^{-rT}\mathbb{E}[Z_T^{\nu}]) + \int_{0}^{T}\widetilde{U}_2(t,ye^{-rt}\mathbb{E}[Z_t^{\nu}])dt\\
		&\ge\widetilde{U}_1(T,ye^{-rT}) + \int_{0}^{T}\widetilde{U}_2(t,ye^{-rt})dt\xrightarrow{y\downarrow 0 }\widetilde{U}_1(T,0+)+ \int_{0}^{T}\widetilde{U}_2(t,0+)dt,
	\end{align*}
	and taking infimum over $\nu\in \Theta$ on both sides. 
	If $\widetilde{U}_1(T,0+)+ \int_{0}^{T}\widetilde{U}_2(t,0+)dt = \infty$, $\hat{L}(0+) = \infty$ and $\partial_y\hat{L}(0+)=-\infty$. Otherwise, if $\widetilde{U}_1(T,0+)+ \int_{0}^{T}\widetilde{U}_2(t,0+)dt < \infty$, we observe that 
	\begin{align}\label{eq:Lhatproof1}
		\hat{L}(y) \le \mathbb{E}\left[ \widetilde{U}_1(T,ye^{-rT}Z^0_T) + \int_{0}^{T}\widetilde{U}_2(t,ye^{-rt}Z^0_t)dt \right],
	\end{align}
	where
	$
	Z_t^0 := \exp\left(-\frac{1}{2}\int_{0}^{t}\hat{\theta}(\pi_s)^2ds-\int_{0}^{t}\hat{\theta}(\pi_s)d\widetilde{W}_s  \right)$, $t\in[0,T]
	$. 
	Since $(Z_t^0)_{t\in[0,T]}$ is a geometric Brownian motion with uniformly bounded $\hat{\theta}$, we have $$yZ_t^0\rightarrow0,\quad \forall t\in[0,T], \text{ a.s. }$$ 
	The last term inside the expectation operator in \eqref{eq:Lhatproof1} is bounded above by\\ $\widetilde{U}_1(T,0+)+ \int_{0}^{T}\widetilde{U}_2(t,0+)dt < \infty$ in the present case. By Dominated Convergence Theorem, we conclude $\hat{L}(0+) = \widetilde{U}_1(T,0+)+ \int_{0}^{T}\widetilde{U}_2(t,0+)dt<\infty$. Hence,
	\begin{align*}
		-\partial_y\hat{L}(0+) \ge \frac{\hat{L}(0+) - \hat{L}(y)}{y}&\ge \frac{1}{y}\Big[  \widetilde{U}_1(T,0+)+ \int_{0}^{T}\widetilde{U}_2(t,0+)dt -  \widetilde{L}(y;\nu')  \Big],
	\end{align*}
	where the last term is greater than $\chi(y;\nu')$ for all $y\in\mathbb{R}_+$ and $\nu'\in \Theta$. Use \eqref{eq:chiasy}, letting $y\rightarrow0$, we have $-\partial_y\hat{L}(0+)\ge \infty$, or $\partial_y\hat{L}(0+)=  -\infty$.\\
	\indent In addition, we observe that the function $-\widetilde{U}_i(t,y)$ is increasing in $y$ and 
	$\lim_{y\rightarrow\infty}-\partial_y\widetilde{U}_i(t,y) = 0$ for  $t\in[0,T]$. Therefore, for any $\epsilon>0$, there exists a constant $K(\epsilon)$ such that
	\begin{align*}
		-\widetilde{U}_1(T,y) \le  K(\epsilon) + \epsilon y, ~\sup_{t\in[0,T]}\{-\widetilde{U}_2(t,y)\} \le K(\epsilon) + \epsilon y, ~\forall y\in\mathbb{R}_+.
	\end{align*}
	By l’Hospital’s Rule, we have 
	\begin{align*}
		0&\le \lim\limits_{y\rightarrow \infty}-\partial_y\hat{L}(y) = \lim\limits_{y\rightarrow\infty}\frac{-\hat{L}(y)}{y} =  \lim\limits_{y\rightarrow\infty} \sup_{\nu\in\Theta} \frac{-\widetilde{L}(y;\nu)}{y}\\
		&\le  \lim\limits_{y\rightarrow\infty} \sup_{\nu\in\Theta} \mathbb{E}\Big[ \frac{K(\epsilon)(1+T)}{y} + \big( e^{-rT}Z_T^{\nu} + \int_{0}^{T}e^{-rt}Z_t^{\nu}dt \big)\epsilon \Big] \le 2\epsilon.
	\end{align*}
	Therefore we have $\lim\limits_{y\rightarrow \infty}-\partial_y\hat{L}(y) = 0$. 
\end{proof}
\begin{proof}[Proof of Theorem~\ref{thm:duality_main}]
	Lemma~\ref{lem:1.5} indicates that for every $v\in\mathbb{R}_+$, there exists $y^*\in\mathbb{R}_+$ such that $-\partial_y\hat{L}(y^*) = v$, or equivalently, $\chi(y^{*},\nu^{y^*})=v$. Theorem~\ref{thm:budget} and Lemma~\ref{lem:1} implies the existence of $(\varpi^*,c^*)\in\mathcal{A}(v)$ with $c^*_t=  I_2(t, e^{-rt}y^*Z_t^{\nu^{y^*}})$ and $V_T^{\varpi^*,c^*} \ge I_1(T, e^{-rT}y^*Z_t^{\nu^{y^*}})$. To  verify the optimality of $(\varpi^*,c^*)$ and that there is no duality gap,  we show the reverse inequality in \eqref{eq:dualgap},
	\begin{align*}
		&\widetilde{J}(v;\varpi^*,c^*) \ge \mathbb{E}\Big[U_1(T,I_1(T, e^{-rT}y^*Z_t^{\nu^{y^*}}))+ \int_{0}^{T}U_2(t,c^*_t)dt\Big] \\
		&=\mathbb{E}\Big[U_1(T,I_1(T, e^{-rT}y^*Z_t^{\nu^{y^*}}))+ \int_{0}^{T}U_2(t,c^*_t)dt\Big] -y^*\chi(y^*;\nu^{y^*})+ y^*v\\
		& = \mathbb{E}\Big[\widetilde{U}_1(T, y^* e^{-rT}Z_T^{\nu^{y^*}}) + \int_{0}^{T} \widetilde{U}_2(t,  y^* e^{-rt}Z_t^{\nu^{y^*}}) dt\Big] + y^*v\\
		&= \hat{L}(y^*) + y^*v \ge \inf_{y^{\prime}>0} \{ \hat{L}(y^{\prime}) + y^{\prime}v \}.
	\end{align*}
	The above calculation shows that $y^*$ attains $\inf_{y\in\mathbb{R}_+}\{\hat{L}(y) + vy\}$. 
\end{proof}

\section{Proof of Theorem~\ref{thm:exist}}\label{A.proof_thm_exist}
\begin{proof}
	To streamline the presentation, we introduce the following notation. Consider a cylindrical domain $\mathcal{O}:=(t_1,t_2)\times O \subset (0,T) \times (0,1)$.\\
	\noindent $\bullet$ $\partial^*\mathcal{O}$ the boundary of $\mathcal{O}$, i.e.,
	$\partial^*\mathcal{O}: = (\{t_1,t_2\}\times O) \cup ((t_1,t_2)\times \partial O).$\\
	\noindent $\bullet$ $\mathcal{L}^p(\mathcal{O})$, the space of $p$th power integrable functions on $\mathcal{O}$, $\| \cdot \|_{p,\mathcal{O}}$ the norm in $\mathcal{L}^p(\mathcal{O})$:
	$\| \varphi \|_{p,\mathcal{O}} = \Big(\int_{\mathcal{O}} |\varphi|^p dxdt\Big)^{1/p}$.\\
	\noindent $\bullet$ $W_p^{1,2}(\mathcal{O})$, $1<p<\infty$, the completion of $C^{\infty}(\mathcal{O})$ under the Sobolev-type norm: 
	\begin{align*}
		\| \varphi \|_{W_p^{1,2}(\mathcal{O})}:= \Big[ \int_{\mathcal{O}}  (|\varphi|^p+|\partial_t \varphi|^p+|\partial_x \varphi|^p+|\partial_{xx} \varphi |^p)dtdx\Big]^{1/p}.
	\end{align*}
	This is the space of all functions $\varphi$ such that for $\varphi(t,x)$ and all of its generalized partial derivatives $\partial_t\varphi$, $\partial_{x} \varphi$, and $\partial_{xx} \varphi$ are in $\mathcal{L}^p(\mathcal{O})$. \\
	\noindent $\bullet$ $|\cdot|_{C_{\iota,\iota/2}(\mathcal{O})}$, $0<\iota\le1$, the H\"{o}lder norm given by $|\varphi|_{C_{\iota,\iota/2}(\mathcal{O})}: = \sup\limits_{(t,x)\in\mathcal{O}} |\varphi(t,x)| +\sup _{\substack{(x, y) \in \overline{O}^2 \\t_1 \leq t \leq t_2}} \frac{\left|\varphi(t, x)-\varphi(t, y)\right|}{|x-y|^{\iota}}+\sup _{\substack{x \in \overline{O}\\
			t_1 \leq s, t \leq t_2}} \frac{\left|\varphi(s, x)-\varphi(t, x)\right|}{|s-t|^{\iota / 2}}$, and
	\begin{align*}
		&|\varphi|^{1}_{C_{\iota,\iota/2}(\mathcal{O})}: = |\varphi|_{C_{\iota,\iota/2}(\mathcal{O})} + 	|\partial_x\varphi|_{C_{\iota,\iota/2}(\mathcal{O})},\\
		& |\varphi|^{2}_{C_{\iota,\iota/2}(\mathcal{O})}: = 	|\varphi|^{1}_{C_{\iota,\iota/2}(\mathcal{O})} + 	|\partial_{xx}\varphi|_{C_{\iota,\iota/2}(\mathcal{O})}+|\partial_{t}\varphi|_{C_{\iota,\iota/2}(\mathcal{O})}.
	\end{align*} We also let	$\overline{C}^{\iota}(\mathcal{O})$ and $\overline{C}^{2+\iota}(\mathcal{O})$ denotes the H\"{o}lder space of all functions $\varphi$ such that $|\varphi|_{C_{\iota,\iota/2}(\mathcal{O})}<\infty$ and $|\varphi|^{2}_{C_{\iota,\iota/2}(\mathcal{O})}<\infty$ respectively.\\ 
	\noindent 	{\it Step 1. }
	The PDE \eqref{eq:HJBPDE} that we analyze has degenerate coefficients on the boundaries of the state space, i.e., on $x=0$ and $x=1$. Thus, we begin with the following auxiliary problem: for a fixed $\ell > 2$, consider a bounded domain $\mathcal{O}_{\ell}: = (0,T)\times (1/\ell, 1-1/\ell) $. The PDE for this auxiliary problem is expressed as 
	\begin{align}\label{thm5.13prob1}
		(\partial_t + \overline{\mu}(x)\partial_{x}+ \frac{1}{2}\overline{\sigma}(x)^2\partial_{xx}-d_0(x))g(t,x) + \mathcal{I}_{\beta}[{\Lambda}^M](t,x) +1= 0, ~\text{ in }\mathcal{O}_{\ell},
	\end{align}
	subject to boundary conditions
	\begin{align*}
		g(t,x) = \Psi(t,x), ~ (t,x) \in \partial^*\mathcal{O}_{\ell}:=( (0,T)\times\{\frac{1}{\ell},1-\frac{1}{\ell}\} \cup (\{T\}\times (\frac{1}{\ell}, 1-\frac{1}{\ell}) ),
	\end{align*}
	with $\Psi(T,x) = 1$ for all $x\in(1/\ell, 1-1/\ell)$ and $\Psi(t,x) = \psi(t,x)$, for all $(t,x)\in(0,T)\times\{\ell,1-\ell\}$, for some $\psi \in C^{1,2}(\overline{\mathcal{O}}_{\ell})$. As $\overline{\mathcal{O}}_{\ell}$ avoids the boundaries $x=0$ and $x=1$, together with the boundedness of $\overline{\mu}$, $\overline{\sigma}$, $d_0$, and the autonomous term $\mathcal{I}_{\beta}[\Lambda^M]$, it follows from the standard results on parabolic PDEs (see e.g. Appendix E in \cite{fleming2012deterministic}) that the boundary value problem \eqref{thm5.13prob1} has a unique solution in $W_p^{1,2}(\mathcal{O}_{\ell})$, for any $p>0$. Applying $W_p^{1,2}$ interior estimate with part of the boundary (see e.g. Chapter IV in \cite{ladyzhenskaia1988linear}), we have
	\begin{align}\label{thm5.15eqlocest1}
		\| g \|_{W_p^{1,2}(\mathcal{O}_{\ell})} \le C( 	\|\mathcal{I}_{\beta}[\Lambda^M]\|_{p,\mathcal{O}_{\ell}} + \| \Psi \|_{W_p^{1,2}(\mathcal{O}_{\ell})} ) \le C_1.
	\end{align}
	For $p>3$, the finiteness of $\| g \|_{W_p^{1,2}(\mathcal{O}_{\ell})}$ implies the finiteness of $|g|_{C_{\iota,\iota/2}(\mathcal{O}_{\ell})}$, for some $\iota>0$. For $p>3$, by estimate (E.9) in Appendix E in \cite{fleming2012deterministic}, we have
	\begin{align*}
		|g|^{1}_{C_{\iota,\iota/2}(\mathcal{O}_{\ell})} \le C\| g \|_{W_p^{1,2}(\mathcal{O}_{\ell})} \le C_2.
	\end{align*}
	We now consider an open subset $\mathcal{O}_{\ell}^{\prime}$ of $\mathcal{O}_{\ell}$ such that $\overline{\mathcal{O}}_{\ell}^{\prime} \subset \overline{\mathcal{O}}_{\ell}$. Recall Lemma~\ref{lem:3},  $\mathcal{I}_{\beta}[\Lambda^M](t,x)$ is H\"{o}lder continuous in $x$ with exponent $0<\iota<1$ . By estimate (E.10) in Appendix E of \cite{fleming2012deterministic}, we have
	\begin{align}\label{thm5.15eqlocest3}
		|g|^{2}_{C_{\iota,\iota/2}(\mathcal{O}_{\ell}^{\prime})} \leq C(|\mathcal{I}_{\beta}(\Lambda^M)|_{C_{\iota,\iota/2}(\mathcal{O}_{\ell})}+\sup\limits_{(t,x)\in\mathcal{O}_{\ell}} |g(t,x)|) \le C_{3},
	\end{align}
	for some constant $C_{3}$ depending solely on $\mathcal{O}_{\ell}^{\prime}$ and $\mathcal{O}_{\ell}$. The H\"{o}lder norm $|g|^{2}_{C_{\iota,\iota/2}(\mathcal{O}_{\ell}^{\prime})} $ is finite, and thus we have $g\in \overline{C}^{2+\iota}(\mathcal{O}_{\ell}^{\prime})$ for any compact subset $\mathcal{O}_{\ell}^{\prime}$ of $\mathcal{O}_{\ell}$. From Theorem 10.1 in Chapter IV of \cite{ladyzhenskaia1988linear}, we conclude $g \in C^{1,2}(\mathcal{O}_{\ell})$.\\
	\indent {\it Step 2. } For $\ell = 3,4,\cdots$, construct a function $\psi_{\ell}$ satisfying the following 
	\begin{align*}
		\psi_{\ell}(x) \in C^{\infty}, ~0 \le 	\psi_{\ell}(x) \le 1, ~  	|\psi_{\ell}'(x)| \le 2,
	\end{align*}
	with $\psi_{\ell}(x) = 1$ for $x\in\mathcal{O}_{\ell}$ and $\psi_{\ell}(x) = 0$ for $x\in\mathcal{U}_T\backslash\mathcal{O}_{\ell+1}$.
	For any fixed $\ell_0>2$, let $\widetilde{g}_{\ell}$ be a solution of PDE:
	\begin{align}\label{thm5.15prob2}
		\partial_t\widetilde{g}_{\ell}+ \overline{\mu}\partial_{x}\widetilde{g}_{\ell}+ \frac{1}{2}(\overline{\sigma}+ 1-\psi_{\ell_0}^2)\partial_{xx}\widetilde{g}_{\ell}-d_0\widetilde{g}_{\ell} + \mathcal{I}_{\beta}[\Lambda^M] +1= 0,~\text{in }\mathcal{U}_T,
	\end{align}
	subject to boundary conditions $\widetilde{g}_{\ell}(T,x) = \psi_{\ell}(x)$. For $\ell>\ell_0$, applying $W_p^{1,2}$ interior estimate with part of the boundary to \eqref{thm5.15prob2},  we have
	$$\| \widetilde{g}_{\ell} \|_{W_p^{1,2}(\mathcal{O}_{\ell_0})} \le C_4, ~\forall  p>1, \ell > \ell_0, $$
	with $C_4$ solely depends on $\ell_0$ but is independent of $\ell$. In addition, for $\ell > \ell_0$, $\widetilde{g}_{\ell}$ uniformly 
	satisfies the PDE 
	\begin{align*}
		\partial_t v+ \overline{\mu}\partial_{x}\widetilde{g}_{\ell}+ \frac{1}{2}\overline{\sigma}^2\partial_{xx}\widetilde{g}_{\ell}-d_0\widetilde{g}_{\ell} + \mathcal{I}_{\beta}[\Lambda^M]+1= 0, ~\text{in }\mathcal{O}_{\ell_0}.
	\end{align*}
	Using the arguments developed in Step 1, we show that the above PDE has a unique solution that coincides with $\widetilde{g}_{\ell}$ in $\mathcal{O}_{\ell_0}$. Taking into account the estimates in \eqref{thm5.15eqlocest3}, we see that $\partial_{t}\widetilde{g}_{\ell}$, $\partial_{x}\widetilde{g}_{\ell}$ and $\partial_{xx}\widetilde{g}_{\ell}$ also satisfy a uniform H\"{o}lder condition on $\mathcal{O}_{\ell_0}$. Note that the coefficients of these equations are the same for all $\ell$ and that $\widetilde{g}_{\ell}$ are uniformly bounded above on $\mathcal{O}_{\ell_0}$. It then follows from Theorem 15 on page 80 in \cite{fleming2012deterministic}  that for any subsequence $\{ \widetilde{g}_{\ell'}\}$ of $\{\widetilde{g}_{\ell}\}$, there exists a further subsequence $\{ \widetilde{g}_{\ell''}\}$, such that itself (and its derivatives $\{ \partial_{t}\widetilde{g}_{\ell''}\}$, $\{ \partial_{x}\widetilde{g}_{\ell''}\}$, and $\{ \partial_{xx}\widetilde{g}_{\ell''}\}$) tends to the limit $\widetilde{g}$ (resp. $\{ \partial_{t}\widetilde{g}\}$, $\{ \partial_{x}\widetilde{g}\}$, and $\{ \partial_{xx}\widetilde{g}\}$) uniformly on each compact subset of $\overline{\mathcal{O}}_{\ell_0}$ (resp. $\mathcal{O}_{\ell_0}$). It then follows from the continuity of $\widetilde{g}_{\ell''}$ and the uniform convergence that $\widetilde{g}\in C^{1,2}(\mathcal{U}_T)$ as $\mathcal{O}_{\ell_0}$ is arbitrarily chosen. To conclude, $\widetilde{g}$ is a classical solution to PDE \eqref{eq:HJBPDE} with the terminal condition $\widetilde{g}(T,x) = 1$. 
\end{proof}
\section{Other proofs}\label{A.proof_Prop1}
\begin{proof}[Proof of Proposition~\ref{prop:filter_boundary}]
	Consider the infinitesimal generator $\mathcal{L}$ of \eqref{eq:filter2} that operates on $\phi\in C^2([0,1])$:
	\begin{align*}
		\mathcal{L}\phi(x):= &(a_2-(a_1+a_2)x)\phi'(x)+ \frac{1}{2}x^2(1-x)^2 (\theta_1-\theta_2)^2\phi{''}(x) \\
		&+ \lambda\int_{\mathcal{Z}} [\phi(\xi(x,z)) - \phi(x)] \hat{f}(x,z)dz.
	\end{align*}
	We show that boundary $0$ is unattainable from the interior of the state space and the arguments for boundary $1$ are similar. Without loss of generality, we prove for the process $\pi^{x_0}:=(\pi^{x_0}_t)_{t\ge0}$ which is defined as the solution of \eqref{eq:filter2} starting from time 0 and a given starting point $x_0\in(0,1)$.
	Consider the function $\phi(x)=1/x$, we have $\phi'(x)= - 1/x^2$, $\phi''(x) = 2/x^3$, and $\phi(x)\rightarrow\infty$ as $x\rightarrow 0$. Consequently, 
	\begin{align*}
		\mathcal{L}\phi(x)
		\le& \frac{1}{x}\Big[ a_1+a_2 + (1-x)^2(\theta_1-\theta_2)^2 +\lambda\int_{\mathcal{Z}} \frac{(f_2(z)-f_1(z))(1-x)\hat{f}(x,z)}{f_1(z)} dz \Big]\\
		\le & \frac{1}{x}\left[ a_1+a_2 + (\theta_1-\theta_2)^2 +\lambda(b_{\max}-1) \right] = L_{0} \phi(x),
	\end{align*}
	by setting $L_{0}: = a_1+a_2 + (\theta_1-\theta_2)^2 +\lambda(b_{\max}-1)>0$. Define $\tau_{n}:=\inf\{t>0:\pi_t^{x_0} \le n \}$, for $0\le n<1$. From the above calculations, we have 
	\begin{align*}
		\mathbb{E}\left[ \phi(\pi^{x_0}_{t \wedge \tau_n }) \right] \le \phi(x_0) + L_{0}\int_{0}^{t}\mathbb{E}\left[ \phi(\pi^{x_0}_{s \wedge \tau_n }) \right]ds,~\forall t>0,~n\in[0,1).
	\end{align*}
	By Gronwall's lemma, we get 
	\begin{align}\label{eq:bound1}
		\mathbb{E}\left[ \phi(\pi^{x_0}_{t \wedge \tau_n }) \right] \le e^{L_0t}\phi(x_0), ~\forall t>0.
	\end{align}
	Assume the contrary that $0$ is attainable, that is $\mathbb{P}(\tau_0 < \infty) >0$. Then, for a large $T_0>0$, $\mathbb{P}(\tau_0 \le T_0) > 0$. Taking $t = T_0$ in \eqref{eq:bound1}, we get 
	\begin{align}\label{eq:bound2}
		\mathbb{E}\left[ \mathds{1}_{(\tau_0\le T_0)}  \phi(\pi^{x_0}_{\tau_0})\right] \le e^{L_0 T_0}\phi(x_0).
	\end{align}
	As $\phi(\pi^{x_0}_{\tau_0}) = \phi(0) = +\infty$ on a positive measure subset $\{ \tau_0 \le T_0\}$, the left-hand side of \eqref{eq:bound2} is infinite while the right-hand side is finite, which is a contradiction. Therefore,  $\mathbb{P}(\tau_0 < \infty) = 0$. 
\end{proof}
\begin{proof}[Proof of Proposition~\ref{prop:filter_feller}]
	This proposition follows by slightly modifying the proof of Theorem~\ref{thm:lipctsx} in Sect.~\ref{sec:cts}. Consider the case when $t=0$, from the definitions of {\it Feller process}, we need only to show that $|\mathbb{E}^{0,x}[f(\pi^x_s)]-\mathbb{E}^{0,y}[f(\pi_s^y)]|$ tends to 0 as $|x-y|\rightarrow 0$ for any given $s>0$ and bounded continuous function $f$. The proof here is even simpler since we can use the uniform norm of the function $f$ to establish appropriate estimates as what we have done in the proof of Theorem~\ref{thm:lipctsx}. 
\end{proof}

\begin{proof}[Proof of Proposition~\ref{prop:back_primal}] For given $(t,x,v) \in \mathcal{U}_T\times\mathbb{R}_+$, by Theorem~\ref{thm:1} and Theorem~\ref{thm:duality_main}, there is no duality gap, that is, 
	\begin{align*}
		J(t,x,v) = \inf_{y\in\mathbb{R}_+} \left(\hat{L}(t,x,y) +vy\right) = \frac{1}{\kappa} v^{\kappa} \Lambda^M(t,x)^{1-\kappa}.
	\end{align*}
	Now we derive the optimal controls for the primal problem. Given $\hat{\Lambda}$ is smooth and the terminal condition $\hat{\Lambda}(T,x) = 1$, it is standard to verify that the process 
	\begin{align*}
		\mathcal{M}_s: =\left(  e^{-r(s-t)}Z_s^{\nu^*}\right)^{\beta} \hat{\Lambda}(s,\pi_s) + \int_{t}^{s} \left(  e^{-r(u-t)}Z_u^{\nu^*}\right)^{\beta} du,~s\in[t,T],
	\end{align*}
	is a $(\mathbb{P}^{t,x},\mathbb{H})$-martingale, for $Z^{\nu^*}$ defined as in \eqref{eq:z_2}. Thus, $\mathbb{E}^{t,x}[\mathcal{M}_T|\mathcal{H}_s] = \mathcal{M}_s$. Following Theoerem~\ref{thm:duality_main}, the candidate optimal wealth process $V^*$ is defined as 
	\begin{align*}
		{V}^*_s :=  \mathbb{E}^{t,x}\Big[ e^{-r(T-s)}Z^{\nu^*}_T\mathcal{V}^* + \int_{s}^{T}  e^{-r(u-s)}Z^{\nu^*}_u{c}^*_u du \Big\vert \mathcal{H}_s \Big],~s \in[t,T],
	\end{align*}
	and $\mathcal{V}^*= {v}(  e^{-r(T-t)}Z_T^{\nu^*})^{\beta-1}/{\hat{\Lambda}(t,x)}$ and $c^*_s = {v}(  e^{-r(s-t)}Z_s^{\nu^*})^{\beta-1}/{\hat{\Lambda}(t,x)}$. Hence,
	\begin{align*}
		{V}^*_s 
		= & \frac{ e^{r(s-t)}v}{\hat{\Lambda}(t,x) Z_s^{\nu^*}}\mathbb{E}^{t,x}\Big[  \mathcal{M}_T - \int_{t}^{s} \Big(  e^{-r(u-t)}Z_u^{\nu^*}\Big)^{\beta} du  \Big \vert \mathcal{H}_s \Big]\\
		=& v (e^{-r(s-t)}Z_s^{t,{\nu}^*})^{\beta-1} \frac{\hat{\Lambda}(s,\pi_s)}{\hat{\Lambda}(t,x)}.
	\end{align*}
	Applying It\^{o}'s lemma to the discounted wealth process $e^{-r(s-t)}V_s^*$ together with the fact that $\hat{\Lambda}$ solves PDE \eqref{eq:HJB} (Theoerem~\ref{thm:main}) and the form of $\nu^*$ in \eqref{eq:optimal_nu_2},  we have
	\begin{align*}
		de^{-r(s-t)}V_s^*
		=& \frac{ve^{-\beta r(s-t)}(Z_s^{\nu^*})^{\beta-1}}{\hat{\Lambda}(t,x) }\Big( -1ds + \hat{\Lambda}\big[ (1-\beta)\hat{\theta} + \frac{\partial_{x} \hat{\Lambda}}{\hat{\Lambda}}\big](s,\pi_s) d\widetilde{W}^{\mathbb{Q}}_s \Big),
	\end{align*}
	where $d\widetilde{W}^{\mathbb{Q}}_s: =\hat{\theta}(\pi_s)ds+d\widetilde{W}_s$. Therefore, 
	\begin{align*}
		d ~	e^{-r(s-t)} V^*_s  + e^{-r(s-t)}c^*_sds =  e^{-r(s-t)}{{\varpi}^*_s}\sigma  d\widetilde{W}^{\mathbb{Q}}_s,
	\end{align*}	
	for $\varpi^*_s =\hat{\varpi}(s,\pi_s,V_s^*)= {{V}^*_s}\Big[(1-\beta)\hat{\theta}(\pi_s) + {\partial_{x} \hat{\Lambda}(s,\pi_s)}/{\hat{\Lambda}(s,\pi_s)} \Big]/{\sigma}$, $s\in[t,T]$, which is the candidate investment strategy. 	In addition, the candidate optimal consumption process ${c}^*$ can be written in a feedback form as 
	\begin{align*}
		{c}^*_s&=\frac{v}{\hat{\Lambda}(t,x)}\Big(  e^{-r(s-t)}Z_s^{\nu^*}\Big)^{\beta-1} = v\Big(  e^{-r(s-t)}Z_s^{\nu^*}\Big)^{\beta-1} \frac{\hat{\Lambda}(s,\pi_s)}{\hat{\Lambda}(t,x)}\frac{1}{\hat{\Lambda}(s,\pi_s)}\\
		&= \hat{c}(s,\pi_s, {V}^*_s), \text{ with } \hat{c}(s,x,v): = \frac{v}{\hat{\Lambda}(s,x)},~ s\in[t,T].
	\end{align*}
	From the boundedness and continuity of $\hat{\Lambda}$ (Proposition~\ref{prop:Lambda_bound} and Theorem~\ref{thm:main}), together with the formula of ${V}^*$, we conclude that $({\varpi}^*,{c}^*) \in \mathcal{A}(t,x,v)$ and thus is the optimal control pair for the primal problem \ref{def:primal_con_pro}. 
\end{proof}

\begin{proof}[Proof of Lemma~\ref{lem:2}]
	By definitions in \eqref{eq:filter_measure_nu}, $\overline{\mu}(\cdot)$ and $\overline{\sigma}(\cdot)$ are continuously differentiable functions of the state variable on the bounded interval $[0,1]$. Hence, the Lipschitz and growth conditions \eqref{cond:1} are fulfilled. Recall the definition of $\xi$ in \eqref{def:xi}, the second part of
	\eqref{cond:3} holds by observing that	$0\le {xf_1(z)}/({xf_1(z)+(1-x)f_2(z)}) \le 1$, for all $x\in [0,1]$, while the equalities hold at $x=0$ and $x=1$ respectively. Next, for the first part of \eqref{cond:3}. We have
	\begin{align*}
		\left|\xi\left(x, z\right)-\xi\left(y, z\right)\right|
		\le  \max\left(\frac{f_1(z)}{f_2(z)},\frac{f_2(z)}{f_1(z)}\right)|x-y|. 
	\end{align*}
	Set $\rho(z) = \max\left({f_1(z)}/{f_2(z)},{f_2(z)}/{f_1(z)}\right)$, it is clear that $\int_{\mathcal{Z}}\rho^2(z)f_1(z)dz <\infty$ under \ref{cond:BLR}. 
\end{proof}
\begin{proof}[Proof of Proposition~\ref{prop:pigrowth}]
	We prove for $k=2$, and the assertions for $k\in[0,2]$ follow from the H\"{o}lder inequality. By corollary 2.12 in \cite{kunita2004stochastic} and $\nu \in \Theta^{t,M}$, there exists a constant $C$ uniformly such that:
	\begin{align*}
		\widetilde{\mathbb{E}}^{t,x,\nu}\Big[ \sup\limits_{t \le u \le s}|\pi_u^{t,x,\nu}|^2 \Big] &\le C \bigg\{  |x|^2 +  	\widetilde{\mathbb{E}}^{t,x,\nu} \Big[ \int_{t}^{s} | \overline{\mu}(\pi_u^{t,x,\nu})|^2 du \Big]  \\
		&+  	\widetilde{\mathbb{E}}^{t,x,\nu} \Big[ \ \int_{t}^{s} | \overline{\sigma}(\pi_u^{t,x,\nu})|^2 du \Big] \\
		&+ 	\widetilde{\mathbb{E}}^{t,x,\nu}\Big[ \int_{t}^{s} \int_{\mathcal{Z}} \lambda| \xi(\pi_{u-}^{t,x,\nu},z) - \pi_{u-}^{t,x,\nu}|^2 \hat{f}(\pi_{u-}^{t,x,\nu},z) dz du \Big]\bigg\}.
	\end{align*}
	We use the linear growth condition of $\overline{\mu}$, $\overline{\sigma}$, $\xi$ given by Lemma \ref{lem:2} to obtain that       
	\begin{align*}
		\widetilde{\mathbb{E}}^{t,x,\nu}\Big[ \sup\limits_{t \le u \le s}|\pi_u^{t,x,\nu}|^2 \Big] \le& C^{\prime} \Big\{  |x|^2 +  	\widetilde{\mathbb{E}}^{t,x,\nu} \big[ \int_{t}^{s}  1+ \sup\limits_{t \le u \le s}|(\pi_u^{t,x,\nu})|^2 du \big]  \Big\}.
	\end{align*}   
	\eqref{eq:prop5eq1} follows by Gronwall's inequality, $C_{\pi} = \max(C,C')$. A similar argument applies to \eqref{eq:prop5eq2}. 
\end{proof}

\begin{proof}[Proof of Proposition~\ref{prop:visco_def_eq}]
	The proof is motivated by Proposition 1.3 in \cite{barles2008second} and Proposition 5.4 in \cite{seydel2010general}. We extend these arguments to the current L\'{e}vy-type jump setting. We begin with an analysis of the operator $H_g[t,x,\nu]$. Let $\nu^{(k)}\in[-M,M]$ with $\nu^{(k)}\rightarrow\nu$,  $x_k\in[0,1]$ with $x_k \rightarrow x $ and $g_k$ is a sequence of uniformly bounded functions, such that $|g_k| \le \phi$ with bounded $\phi \in C(\overline{\mathcal{U}}_T)$ and $\lim\limits_{k \rightarrow \infty} g_k = g$. The Dominated Convergence Theorem gives
	\begin{align}
		&\lim\limits_{k\rightarrow \infty} \int_{\mathcal{Z}} (g_k(t,\xi(x_k,z)) - g_k(t, x_k) ) e^{\beta \nu^{(k)}}\hat{f}(x_k,z) dz \nonumber\\
		&=  \int_{\mathcal{Z}} (g(t,\xi(x,z)) - g(t, x) )e^{\beta \nu}\hat{f}(x,z) dz, ~t\in [0,T], \label{prop5.11pf1}
	\end{align}
	in which we have used the continuity of $\xi$ and $\hat{f}$ in $x$.\\
	\indent Assume that $g$ is a viscosity subsolution according to Definition~\ref{def:viscosity2}. Let $(t_0,x_0)\in {\mathcal{U}}_T$, $\psi\in C^{1,2}(\overline{\mathcal{U}}_T)$ such that $v-\psi$ has a global maximum at $(t_0,x_0)$. For any given $\nu\in[-M,M]$, we have $H_g(t_0,x_0,\nu) \le H_{\psi}(t_0,x_0,\nu)$, and therefore 
	\begin{align*}
		(-{\partial_t}-\overline{\mu}(x_0)\partial_x-\frac{1}{2} \overline{\sigma}(x_0)^2\partial_{xx})\psi(t_0,x_0)-1 -\max\limits_{\nu\in [-M,M]}  H_{\psi}(t_0,x_0,\nu)\leq 0, 
	\end{align*}
	which implies that $g$ is also a viscosity subsolution according to Definition~\ref{def:viscosity1}. \\
	\indent Assume $g$ is a viscosity subsolution according to Definition~\ref{def:viscosity1}. Let $(t_0,x_0)\in {\mathcal{U}}_T$, $\psi\in C^{1,2}(\overline{\mathcal{U}}_T)$ such that $g-\psi$ has a global maximum at $(t_0,x_0)$. Consider the function for a sufficiently small $\epsilon_0\in(0,1)$, 
	$$ \varphi(t,x): = \mathds{1}_{x\in[x_0-\epsilon_0,x_0+\epsilon_0]} \psi(t,x) + \mathds{1}_{x\in[0,1]\backslash[x_0-\epsilon_0,x_0+\epsilon_0]} g(t,x).$$
	$\varphi$ is clearly a continuous function, we then construct a bounded sequence $\varphi_k\in C^{1,2}(\overline{\mathcal{U}}_T)$ such that $|\varphi_k|\le \psi$ with $\lim\limits_{k\rightarrow\infty}\varphi_k = \varphi $. By construction, $g\le \varphi_k$ with equality holds at $(t_0,x_0)$. From Definition~\ref{def:viscosity1}, we find that for all $k$,
	\begin{align*}
		(-{\partial_t}-\overline{\mu}(x_0)\partial_x-\frac{1}{2} \overline{\sigma}(x_0)^2\partial_{xx})\psi(t_0,x_0)-1 -\max\limits_{\nu\in [-M,M]}  H_{\varphi_k}(t_0,x_0,\nu) \leq 0.
	\end{align*}
	For each $k$, the max in l.f.s of above equation is attained by a $\nu^{(k)}$. Choose a subsequence such that $\nu^{(k)} \rightarrow \nu \in [-M,M]$ and use the limit in \eqref{prop5.11pf1}, we have
	\begin{align*}
		(-{\partial_t}-\overline{\mu}(x_0)\partial_x-\frac{1}{2} \overline{\sigma}(x_0)^2\partial_{xx})\psi(t_0,x_0)-1 -\max\limits_{\nu\in [-M,M]}  H_{\varphi}(t_0,x_0,\nu) \leq 0.
	\end{align*}
	Finally, by sending $\epsilon_0$ to 0, we complete the proof.  
	%In particular, we have 
	%$ H_{\varphi_k}(t_0,x_0,\nu) = \Gamma(x_0,\nu) v(t_0, x_0)+\int_{\mathcal{Z}}\left\{\psi(t_0,x_0+\xi(x, z))-v(t,x)\right\}\lambda e^{\beta\nu}\hat{f}(x,z) d z$
\end{proof}
\begin{proof}[Proof of Lemma~\ref{lem:3}] For $t\in[0,T)$, $x,~y\in[0,1]$, we have
	\begin{align*}
		&\Big |\mathcal{I}_{\beta}[\Lambda^{M}](t,x) - \mathcal{I}_{\beta}[\Lambda^{M}](t,y)\Big|\le (1-\beta)\lambda \int_{\mathcal{Z}} \Big|{\Lambda^{M}(t,\xi(x,z))}^{\frac{1}{1-\beta}}-{\Lambda^{M}(t,x)^{\frac{1}{1-\beta}}}\Big|\\
		&\quad \Big| \Lambda^{M}(t,x)^{\frac{\beta}{\beta-1}}\hat{f}(x,z) -\Lambda^{M}(t,y)^{\frac{\beta}{\beta-1}}\hat{f}(y,z) \Big| dz \\
		&\quad+ (1-\beta)\lambda   \int_{\mathcal{Z}} \Lambda^{M}(t,y)^{\frac{\beta}{\beta-1}}\hat{f}(y,z) \left\{ \Big| \Lambda^{M}(t,\xi(x,z))^{\frac{1}{1-\beta}} -\Lambda^{M}(t,\xi(y,z))^{\frac{1}{1-\beta}}\Big| \right. \\
		&\quad+ \left.\Big| \Lambda^{M}(t,x)^{\frac{1}{1-\beta}} -\Lambda^{M}(t,y)^{\frac{1}{1-\beta}} \Big|\right\} dz.
	\end{align*}
	We get that the above terms are bounded above 
	through repeatedly using the Lipschitz continuity and boundedness of $\Lambda^{M}$, the properties of $\xi$ as given in Lemma~\ref{lem:2}, and \ref{cond:BLR}. 
\end{proof}	

\section{Supplementary notations and conditions}\label{App:E}
\begin{defn}\label{A.def:func_index}
	For any filtration $\mathbb{G}$, we denote the predictable $\sigma$-field on the product space $[0, T] \times \Omega$ by $\overline{\mathcal{P}}(\mathbb{G})$. Let $\mathcal{B}(\mathcal{Z})$ be the Borel $\sigma$-algebra on $\mathcal{Z}$. Any mapping $H:[0, T] \times \Omega \rightarrow \mathcal{Z} $, which is $\overline{\mathcal{P}}(\mathbb{G}) \times \mathcal{B}(\mathcal{Z})$-measurable, is referred to as a $\mathbb{G}$-predictable process indexed by $\mathcal{Z} $.
\end{defn}
Letting
$
\mathcal{F}_{t}^{N}:=\sigma\{N((0, s] \times A): 0 \leq s \leq t, A \in \mathcal{B}(\mathcal{Z})\},
$
we denote by $\mathbb{F}^{N}:=\left(\mathcal{F}_{t}^{N}\right)_{0 \le t \le T}$ the filtration which is generated by the random measure $N(d t, d z)$.
\begin{defn}\label{A.def:dual_projec} Given any filtration $\mathbb{G}$ with $\mathbb{F}^{N} \subseteq \mathbb{G}$, the $\mathbb{G}$-dual predictable projection of $N$, denoted by $N^{\mathbb{P}, \mathbb{G}}(d s, d z)$, is the $\mathbb{G}$-predictable random measure, such that for any nonnegative $\mathbb{G}$-predictable process $\Phi$ indexed by $\mathcal{Z}$, 
	\begin{align*}
		\mathbb{E}[\int_{0}^{\infty} \int_{\mathcal{Z}} \Phi(s, z) N(d s, d z)]=\mathbb{E}[\int_{0}^{\infty} \int_{\mathcal{Z}} \Phi(s, z) N^{\mathbb{P}, \mathbb{G}}(d s, d z)].	
	\end{align*}
\end{defn}
\begin{assump}\label{A.ass:eta} 
	For each $ i \in \mathcal{S}$, we assume $\gamma$  is a L\'{e}vy kernel such that  $\gamma(i,dz)$  is a non-negative $\sigma$-finite measure on $(z,\mathcal{Z})$. In addition, there exists a constant $L_H > 0$ such that
	$$  \sup_{i \in \mathcal{S}} \int_{\mathcal{Z}} \gamma(i,dz) \le L_H< \infty.$$ 
	The functions $b_{1}(q,i)$, $\sigma_{1}(q)$, $\sigma_{2}(q)$ and $b_2(q,\cdot)$ are continuous in $q$.  In addition, we assume the following two conditions hold.
	\begin{itemize}
		\item[(i)  ] (Lipschitz conditions) For all $ i \in \mathcal{S}$, and $q_1,q_2\in \mathbb{R}$, there is a constant $C>0$, 
		\begin{align*}
			&\left|b_{1}(q_1, i)-b_{1}\left(q_2, i\right)\right|^2+\left|\sigma_{1}(q_1)-\sigma_{1}\left(q_2\right)\right|^2\\
			&+\left|\sigma_{2}(q_1)-\sigma_{2}\left(q_2\right)\right|^2 + \int_{\mathcal{Z}} |b_2(q_1,z) -b_2(q_2,z)|^2 \gamma(i,dz)\leq C\left|q_1-q_2\right|^2, 
		\end{align*}
		\item[(ii)  ] (growth conditions) For all $ i \in \mathcal{S}$, $q\in \mathbb{R}$, there is a constant $C>0$ such that
		$$
		\left|b_{1}(q, i)\right|^{2}+\left|\sigma_{1}(q)\right|^{2}+\left|\sigma_{2}(q)\right|^{2} +\int_{\mathcal{Z}} |b_2(q,z)|^2 \gamma(i,dz) \leq C\Big(1+|q|^{2}\Big).
		$$
	\end{itemize}
\end{assump}
The above conditions ensure strong existence and uniqueness for solutions to \eqref{eq:eta} (with a standard localizing argument, such as those used in the proof of theorem 2.1 in \cite{xi2017feller}, and theorem 3.6 in \cite{xi2018martingale}  and employing theorem 5.2 in \cite{komatsu1973markov}.)
\begin{prop}\label{prop:martg_represen}
	Under Assumption~\ref{A.ass:eta}, let $Y$ be any $(\mathbb{P},\mathbb{H})$-local martingale with $Y_0 = 0$. Then, there exists $\mathbb{H}$-predictable processes $\psi$, $\Psi$ and $\varphi$, such that 
	\begin{align*}
		&Y_t = \int_{0}^{t} \psi_u d \widetilde{W}_u  + \int_{0}^{t} \Psi_u \widetilde{B}_u + \int_{0}^{t}\int_{\mathbb{R}} \varphi(u,q) \overline{m}^{\pi}(du,dq), \quad 0\le t \le T, \nonumber\\
		&\text{ and }{\int_{0}^{T} \left(\psi_u^2 + \Psi_u^2\right) du} + \int_{0}^{T}\int_{\mathbb{R}} |{\varphi}(u,q)| \hat{\lambda}(\pi_{u-})\hat{\phi}_u(\pi_{u-},q)(du,dq) <+\infty, ~ \mathbb{P}\text{-a.s.}
	\end{align*} 
\end{prop}
The proof follows by modifying those in the seminal papers of \cite{kurtz1988unique,ceci2006risk,ceci2012nonlinear,callegaro2020optimal}.

\end{document}